\documentclass[a4paper,onecolumn,11pt,accepted=2025-06-06]{quantumarticle}
\pdfoutput=1
\usepackage[utf8]{inputenc}
\usepackage[english]{babel}
\usepackage[T1]{fontenc}
\usepackage{amsmath}
\usepackage{hyperref}
\usepackage{fixltx2e}
\usepackage[numbers,sort&compress]{natbib}
\usepackage{amssymb,amscd,amsthm, verbatim,amsmath,color,fancyhdr, mathrsfs}
\usepackage{graphicx}
\usepackage{turnstile}
\usepackage[letterpaper, left=2.5cm, right=2.5cm, top=2.5cm,
bottom=2.5cm,dvips]{geometry}
\usepackage{float}
\usepackage{epsfig}
\usepackage{listings}
\usepackage{xcolor}
\usepackage[braket]{qcircuit}
\usepackage{amsmath}
\usepackage{tikz}
\usepackage{cleveref}
\usepackage{lipsum}
\usepackage{tabularx}
\usepackage{array}
\usepackage{booktabs}

\newtheorem{thm}{Theorem}[section]

\newtheorem{lem}[thm]{Lemma}
\newtheorem{rem}[thm]{Remark}
\newtheorem{cor}[thm]{Corollary}

\newtheorem{defn}[thm]{Definition}

\newcommand{\x}{\boldsymbol x}
\newcommand{\A}{\boldsymbol A}
\newcommand{\LL}{\boldsymbol{L}}
\newcommand{\T}{\boldsymbol T}
\newcommand{\W}{\boldsymbol W}

\newcommand{\N}{\boldsymbol N}
\newcommand{\U}{\boldsymbol U}
\newcommand{\D}{\boldsymbol{D}}
\newcommand{\M}{\boldsymbol M}
\newcommand{\C}{\boldsymbol C}
\newcommand{\bb}{\boldsymbol b}
\newcommand{\y}{\boldsymbol y}
\newcommand{\B}{\boldsymbol B}
\newcommand{\CC}{\boldsymbol C}
\newcommand{\I}{\boldsymbol I}
\newcommand{\0}{\boldsymbol O}
\newcommand{\z}{\boldsymbol z}
\newcommand{\e}{\boldsymbol e}
\newcommand{\vv}{\boldsymbol v}
\newcommand{\V}{\boldsymbol V}
\newcommand{\E}{\boldsymbol E}
\newcommand{\G}{\boldsymbol G}
\newcommand{\F}{\boldsymbol F}
\newcommand{\HH}{\boldsymbol H}
\newcommand{\poly}{\mathrm{poly}}
\newcommand{\X}{\boldsymbol X}

\newcommand{\diff}{\mathrm{d}}
\newcommand{\bbC}{\mathbb{C}}
\newcommand{\bbR}{\mathbb{R}}

\newcommand{\calT}{\mathcal{T}}

\begin{document}
\title{A quantum algorithm for linear autonomous differential equations using Pad\'e approximation}

\author{Dekuan Dong}
\email{dkdong21@m.fudan.edu.cn}
\orcid{0009-0001-5553-0231}
\affiliation{School of Mathematical Sciences, Fudan University}

\author{Yingzhou Li}
\email{yingzhouli@fudan.edu.cn}
\orcid{0000-0003-1852-3750}
\thanks{The research was supported in part by NSFC under grant 12271109,
STCSM under grant 22TQ017 and 24DP2600100, and SIMIS under grant
SIMIS-ID-2024-(CN).}
\affiliation{School of Mathematical Sciences, Fudan University}
\affiliation{Shanghai Key Laboratory for Contemporary Applied Mathematics}

\author{Jungong Xue}
\email{xuej@fudan.edu.cn}
\thanks{The research  was supported in part by the National Science Foundation of China Grant 12171101 and the Laboratory of Mathematics for Nonlinear Science, Fudan University.}
\affiliation{School of Mathematical Sciences, Fudan University}

\maketitle

\begin{abstract}
    We propose a novel quantum algorithm for solving linear autonomous ordinary differential equations (ODEs) using the Pad\'e approximation. For linear autonomous ODEs, the discretized solution can be represented by a product of matrix exponentials. The proposed algorithm approximates the matrix exponential by the diagonal Pad\'e approximation, which is then encoded into a large, block-sparse linear system and solved via quantum linear system algorithms (QLSA). The detailed quantum circuit is given based on quantum oracle access to the matrix, the inhomogeneous term, and the initial state. The complexity of the proposed algorithm is analyzed. Compared to the method based on Taylor approximation, which approximates the matrix exponential using a $k$-th order Taylor series, the proposed algorithm improves the approximation order $k$ from two perspectives: 1) the explicit complexity dependency on $k$ is improved, and 2) a smaller $k$ suffices for the same precision. Numerical experiments demonstrate the advantages of the proposed algorithm comparing to other related algorithms.
\end{abstract}


\section{Introduction}

Quantum computing, which fundamentally differs from
classical computing, leverages quantum mechanics to perform computations using
quantum states and quantum gates. For many problems, quantum computing offers a
theoretical advantage over classical computing. In certain cases,
quantum algorithms executed on quantum computers can achieve exponential
speedup compared to their classical
counterparts~\cite{doi:10.1137/S0097539795293172, PhysRevLett.103.150502,
Montanaro2016}. Solving large-scale ordinary differential equations (ODEs) is one such problem where quantum algorithms show significant potential. 

Solving large-scale ODEs is a fundamental problem in science, with numerous real-world applications across various domains, including weather and
climate prediction, chemical reaction modeling, fluid dynamics,
and finance, such as option pricing models, etc. ODEs can generally be divided into
two categories: linear ODEs and nonlinear ODEs. In each category, ODEs can further be divided into autonomous and non-autonomous ODEs. A
general first-order linear ODEs admit,
\begin{equation} \label{eq:1.1}
    \begin{cases}
        \frac{\diff \x}{\diff t} = \A(t) \x(t) + \bb(t), \quad t \in [0,T],\\
        \x(0) = \x_0,
    \end{cases}
\end{equation}
where $\A(t) \in \bbC^{n\times n}$ is a matrix or discretized operator, $\bb(t) \in \bbC^n$ is the inhomogeneous term, $\x_0 \in \mathbb C^n$ is the initial state, $n$ is the dimension of the system, and $T$ is a final time. When both $\A(t)$ and $\bb(t)$ are time-independent, i.e., $\A(t) = \A$ and $\bb(t) = \bb$, the ODE system \eqref{eq:1.1} is referred to as a first-order autonomous ODE system. Given suitable quantum oracles to access $\A$, $\bb$, and $\x_0$, the goal of this paper is to design an efficient quantum algorithm that produces a quantum state that is $\epsilon$-close to the final state $\ket{\x(T)} =
\frac{\x(T)}{\|\x(T)\|_2}$.

Many quantum algorithms for solving linear ODEs \eqref{eq:1.1} have been proposed in the past decades. These algorithms can be grouped into three categories: a) discretized time steps with Quantum Linear System Algorithm~(QLSA); b) time-ordered integral form with Linear Combination of Hamiltonian Simulation~(LCHS); c) Schr{\"o}dingerization with Hamiltonian simulation. The first category differs from the latter two categories in its use of Hamiltonian simulation. Methods in the first category rely on Hamiltonian simulation underlying the QLSA, whereas methods in the latter two categories apply Hamiltonian simulation directly. Mathematically, methods in the latter two categories yield similar final expressions, although their derivation procedures differ significantly.

Methods in the first category consist of three steps: discretizing the time variable, encoding the discretized linear differential equation into an enlarged linear system, and solving the resulting linear system using QLSA. This approach was first introduced by \cite{berry2014high}, which uses linear multi-step methods to discretize the time variable. However, due to the limitations of the discretization method, the query complexity in \cite{berry2014high} is $\poly(1/\epsilon)$ even if the most efficient QLSA is applied. To improve the dependence on the precision $\epsilon$, \citet{berry2017quantum} applies a truncated Taylor series to approximate the propagator $\exp(\A t)$ for autonomous ODEs and encodes the evaluation process of the truncated Taylor series into a linear system. The dependence on $\epsilon$ is improved to $\poly(\log(1/\epsilon))$. The analysis in \cite{berry2017quantum} requires $\A$ to be diagonalizable, i.e., there exists an invertible matrix $\V$ such that $\V^{-1}\A\V$ is a diagonal matrix, and the final complexity depends linearly on the condition number of $\V$, $\kappa_{\V} = \|\V\|_2\cdot \|\V^{-1}\|_2$. To eliminate the dependence on $\kappa_{\V}$, \citet{krovi2023improved} proposes a modified linear system that also encodes the evaluation process of the truncated Taylor series into a linear system, and improves the analysis in \cite{berry2017quantum}. In~\cite{krovi2023improved}, the resulting complexity depends on $C(\A) := \max_{0 \le t \le T} \|\exp(\A t)\|_2$ instead of $\kappa_{\V}$. More recently, \citet{dong2024investigationquantumalgorithmlinear} revisited the analysis in \cite{berry2017quantum} and obtained a similar complexity result to \cite{krovi2023improved} without modifying the linear system. \citet{wu2025structurepreservingquantumalgorithmslinear} proposed structure-preserving quantum algorithms for both linear and nonlinear Hamiltonian systems. These algorithms maintain the symplectic properties of the underlying dynamics, thereby enhancing the robustness and stability of the numerical methods. Additionally, \citet{Childs2019QuantumSM} and \citet{berry2024quantum} generalize methods in this category to address non-autonomous linear ODEs, i.e., time-dependent $\A(t)$ and $\bb(t)$, using spectral method and Dyson's series, respectively.

The second category of methods originates from \cite{PhysRevLett.131.150603}. These methods are based on a novel integral expression for the non-unitary evolution operator $\calT e^{\int_0^t \A(s) \diff s}$, where $\calT$ denotes the time-ordering operator. By decomposing $\A(t)$ into its Hermitian and anti-Hermitian parts, i.e., $\A(t) = \LL(t) + \imath \HH(t)$, where $\LL(t) = \frac{\A(t) + \A^\dagger(t)}{2}$ and $\HH(t) =
\frac{\A(t) - \A^\dagger(t)}{2 \imath}$, the non-unitary evolution operator can be expressed as an integral,
\begin{equation} \label{eq:identity}
    \calT e^{\int_0^t \A(s) \diff s} = \int_{\bbR}
    \frac{1}{\pi (1 + k^2)} \calT e^{\imath \int_0^t
	(\HH(s) + k\LL(s)) \diff s} \diff k.
\end{equation}
This identity holds under the assumption that $\LL(t)$ is negative semi-definite, i.e.,  $\LL(t) \preceq 0$ for all $t\in[0, T]$. Later, in \cite{An2023QuantumAF}, a large family of identities similar to \cref{eq:identity} is unveiled, and the method in this category is improved to achieve near-optimal dependence on all parameters. The extension to the inhomogeneous case is also possible via Duhamel's principle. Compared to the methods in the first category, methods in this category consider the non-autonomous cases directly and achieve the optimal state preparation cost. However, methods in this category only give the terminal state, whereas methods in the first category provide information throughout the evolution. 

The third category of methods is known as Schr\"odingerization \cite{jin2022quantumsimulationpartialdifferential, PhysRevA.108.032603}. These methods use a warped phase transformation to map the linear differential equation into a higher-dimensional system, where it manifests as a Schr\"{o}dinger equation in the Fourier space of the extra variable. In this way, the linear ODEs with non-unitary dynamics are converted into a system that evolves under unitary dynamics, i.e., a Schr\"{o}dinger equation, which can be simulated on quantum computing via Hamiltonian simulation. Schr\"odingerization has been applied to a wide range of problems, including differential equations with various boundary conditions~\cite{hu2023dilationtheoremschrodingerisationapplications,
jin2024schrodingerisationbasedcomputationallystable,
jin2024quantumsimulationfokkerplanckequation,
jin2023quantumsimulationmaxwellsequations, doi:10.1137/23M1563451,
JIN2024112707, doi:10.1137/23M1566340,
cao2023quantumsimulationtimedependenthamiltonians}, iterative methods in numerical linear algebra~\cite{jin2023quantumsimulationdiscretelinear}, and more. In \cite{jin2024schrodingerizationbasedquantumalgorithms}, Schr\"{o}dingerization is extended to solve linear differential equations with inhomogeneous terms, which is the same problem addressed in this paper. The resulting query complexity depends linearly on $1/\epsilon$. More recently, methods in this category have been adapted for analog quantum simulation~\cite{jin2024analogquantumsimulationparabolic, Jin_2024}. Explicit quantum circuit implementations corresponding to Schr\"odingerization have been explored in~\cite{Hu2024quantumcircuits, jin2025quantumcircuitsheatequation}.

We propose a novel quantum algorithm for solving linear autonomous ODEs~\eqref{eq:1.1} in this paper, which falls within the first category of existing quantum algorithms. Our key contributions and innovations can be summarized as follows.
\begin{itemize}
    \item We replace the Taylor series approximation with the Pad\'e approximation of the matrix exponential in related quantum algorithms~\cite{berry2017quantum, krovi2023improved, dong2024investigationquantumalgorithmlinear}. The Pad\'e approximation, a rational function commonly used in classical computing to approximate matrix exponentials~\cite{higham2005scaling, doi:10.1137/09074721X, moler2003nineteen}, achieves similar accuracy to the Taylor approximation but with a lower order, though it requires a matrix inversion. In \cite{wu2025structurepreservingquantumalgorithmslinear}, the Pad\'e approximation is also applied to preserve the symplectic structure. Their results suggest that our method also preserves the symplectic property when applied to the same Hamiltonian systems. However, our encoding of the Pad\'e approximation differs from theirs.

    \item We propose a matrix encoding for the Pad\'e approximation, denoted as $\LL$. The encoding avoids explicit matrix inversion and maintains a comparable level of complexity to the matrix encoding used for Taylor approximation, as presented in~\cite{berry2017quantum, krovi2023improved, dong2024investigationquantumalgorithmlinear}. In contrast to the Taylor approximation, whose corresponding matrix encoding uses a block lower bidiagonal form, the matrix encoding for the Pad\'e relies on a block upper Hessenberg form. 
    
    \item The condition number of the Pad\'e approximation encoded matrix $\LL$ is estimated, which plays a crucial role in the complexity analysis of the algorithm. Additionally, we propose the block-encoding of the Pad\'e approximation encoded matrix, along with its detailed quantum circuit implementations.

    \item We analyze the overall complexities of the proposed quantum algorithm, including both the oracle query complexity and the quantum gate complexity. When $\A$ is Hermitian and negative semi-definite, our analysis relaxes the common restriction $\|\A h\|_2 \le 1$ found in other related works~\cite{berry2024quantum,berry2017quantum,krovi2023improved,dong2024investigationquantumalgorithmlinear}, where $h$ represents the time step. Compared to other methods in the first category, the proposed algorithm outperforms them both theoretically and numerically.
\end{itemize}

Overall, our proposed quantum algorithm and the corresponding analysis address linear autonomous ODEs. Extensions to the non-autonomous and/or mildly nonlinear ODEs are feasible using similar techniques as those in~\cite{dong2024investigationquantumalgorithmlinear,cao2023quantumsimulationtimedependenthamiltonians,Childs2019QuantumSM,berry2024quantum}. The complexity analysis for the extended algorithm will need to be rederived accordingly, and the dependence on parameters requires further investigation.

The rest of the paper is organized as follows. In \Cref{sec:prel}, we briefly review the method proposed in \cite{berry2017quantum}. \Cref{sec:Poly2Linsys} introduces the linear system encoding via Pad\'e approximation, analyzes the condition number of the linear system, and details the block-encoding of the linear system. \Cref{sec:CompAna} presents the complexity analysis of the proposed algorithm. In \Cref{sec:NumExp}, we compare our method to the one based on Taylor approximation \cite{berry2017quantum,krovi2023improved,dong2024investigationquantumalgorithmlinear} and the LCHS-based method~\cite{PhysRevLett.131.150603,An2023QuantumAF} theoretically. Moreover, numerical experiments are conducted to illustrate the advantages of our method over the method based on Taylor approximation. Finally, \Cref{sec:conclusion} concludes the paper with a discussion on future directions.

\section{Preliminaries}
\label{sec:prel}
Consider a linear differential equation of the form:
\begin{equation}\label{eq:main}
    \begin{cases}
        \frac{\diff \x}{\diff t} = \A\x(t) + \bb,\quad t\in[0,T],\\
        \x(0) = \x_0,
    \end{cases}
\end{equation}
where $\A\in \mathbb C^{n\times n}$ and $\bb\in \mathbb C^{n}$ are time-independent. When $\A$ is non-singular, the solution is given by
\[ \x(t) = \exp(\A t) \x(0) + \left(\exp(\A t) - \I_n\right)\A^{-1} \bb, \quad t\in [0, T]. \]
In order to prepare the quantum state $|\x(T)\rangle=\frac{\x(T)}{\|\x(T)\|_2}$, reference \cite{berry2017quantum} approximates the exponential function $\exp(\cdot)$ by its truncated Taylor series 
\[T_k(z) := \sum_{j=0}^k \frac{z^j}{j!}.\]
Given a suitable number of time steps $m\in \mathbb N_+$ and the time step size $h = T / m$, the following procedure is used to approximate $\x(T)$: 
\begin{equation}\label{eq:proc_taylor}
    \begin{aligned}
        \x(h) &\approx T_k(\A h) \x(0) + \left(T_k(\A h) - \I_n\right)\A^{-1}\bb =: \widehat \x(h),\\
        \x(2h) &\approx T_k(\A h) \widehat \x(h) + \left(T_k(\A h) - \I_n\right)\A^{-1}\bb =: \widehat \x(2h),\\
               &\quad \vdots \\
        \x(mh) &\approx T_k(\A h) \widehat \x\left((m-1)h\right) + \left(T_k(\A h) - \I_n\right)\A^{-1}\bb =: \widehat \x(T).\\ 
    \end{aligned}
\end{equation}
Then, reference \cite{berry2017quantum} shows that each step of the above procedure can be encoded into a linear system. The $s$-th step of the above procedure is encoded in the following linear system:
\begin{equation}\label{eq:berry2017}
    \begin{cases}
        \M_k(\A h) \z^{(s)} = \y_s,\\
        \widehat \x(sh) - (\boldsymbol 1^T \otimes \I_n) \z^{(s)} = 0,
    \end{cases}\quad \forall s = 1, \dots, m,
\end{equation}
where 
\[\M_k(\A h):= \begin{bmatrix}
    \I_n\\ -\A h & \I_n \\
    & \ddots & \ddots & \\
    && -\frac{\A h}k& \I_n
\end{bmatrix}, \quad \y_s :=  \begin{bmatrix}\widehat \x((s-1)h)\\ h\bb \\ \boldsymbol 0 \\\vdots \\ \boldsymbol 0\end{bmatrix} \in \mathbb C^{n(k+1)}, \quad \boldsymbol 1 := \begin{bmatrix}1\\ \vdots \\1\end{bmatrix} \in \mathbb R^{k+1}.\]
Combining the $m$ steps and repeating $\widehat \x(T)$ for $p$ times to boost the success probability, we get the linear system 
\begin{equation}\label{eq:full_M}
    \begin{tikzpicture}
        \draw(0,0)node{$\underset{\displaystyle =:\C_{m,k,p}(\A h)}{\underbrace{\left[
            \begin{array}{cccccccc}
                \M_k(\A h) & & & & & & & \\
                \T & \M_k(\A h) & & & & & & \\
                & \ddots & \ddots & & & & &\\
                && \T & \M_k(\A h) & & & &\\
                &&& -{\boldsymbol 1}^T \otimes \I_n& \I_n& & &\\
                &&&& -\I_n & \I_n& &\\
                &&&&& \ddots & \ddots & \\
                &&&&&& -\I_n & \I_n  
            \end{array}
            \right]}} \begin{bmatrix} \z^{(1)} \\ \z^{(2)} \\ \vdots \\ \z^{(m)}\\ \widehat \x \\ \widehat \x\\\vdots \\ \widehat \x\end{bmatrix} = \begin{bmatrix}\y_1 \\\y_0\\\vdots \\\y_0\\ \boldsymbol 0 \\\boldsymbol 0\\\vdots \\ \boldsymbol 0\end{bmatrix}
            $};
        \draw(-6.5,1.5)node[left]{$m \text{ blocks}\left\{\rule{0mm}{10mm}\right.$};
        \draw(-1.8,-0.7)node[left]{$p \text{ blocks}\left\{\rule{0mm}{10mm}\right.$};
    \end{tikzpicture}
\end{equation}
where
\[\T := -\begin{bmatrix}\I_n& \cdots & \I_n\\
\0_n & \cdots & \0_n\\ \vdots && \vdots \\\0_n & \cdots & \0_n\end{bmatrix}\in \mathbb R^{n(k+1)\times n(k+1)}, \quad \y_1 := \begin{bmatrix}\x(0)\\ h\bb \\ \boldsymbol 0 \\\vdots \\ \boldsymbol 0\end{bmatrix} \in \mathbb C^{n(k+1)}, \quad \y_0 := \begin{bmatrix}\boldsymbol 0\\ h\bb \\ \boldsymbol 0 \\\vdots \\ \boldsymbol 0\end{bmatrix} \in \mathbb C^{n(k+1)}.\]
In the rest of this paper, we refer to the coefficient matrix in \cref{eq:full_M} as $\C_{m,k,p}(\A h)$. With appropriately chosen parameters $m$, $k$, and $p$, we can solve the linear system using QLSA~\cite{costa2022optimal}, and then perform measurements on the resulting quantum state. This allows us to obtain a quantum state that approximates $|\x(T)\rangle$ within the required precision, with constant positive probability. The most time-consuming step in this procedure is solving the linear system using QLSA. Assuming the block-encoding of $\C_{m,k,p}(\A h)$ is available, the query complexity of QLSA depends at least linearly on the condition number of $\C_{m,k,p}(\A h)$. Therefore, it is crucial to control the condition number of the resulting linear system.
\section{Quantum algorithm based on Pad\'{e} approximation}\label{sec:Poly2Linsys}
In this section, we present our quantum algorithm for solving linear differential equations using Pad\'{e} approximation. We begin by outlining the Pad\'{e} approximation to matrix exponential.
\begin{defn}
    The $(p,q)$ Pad\'e approximation to $e^{\A}$ is defined by 
    \begin{equation}\label{eq:pade}
        R_{pq}(\A) = [D_{pq}(\A)]^{-1} N_{pq}(\A) = N_{pq}(\A) [D_{pq}(\A)]^{-1},
    \end{equation}
    where 
    \begin{equation}\label{eq:pade_N}
        N_{pq}(\A) = \sum_{j = 0}^p n_j \A^j, \quad n_j = \frac{(p+q - j)! p!}{(p+q)! j! (p - j)!}
    \end{equation}
    and 
    \begin{equation}\label{eq:pade_D}
        D_{pq}(\A) = \sum_{j=0}^q d_j (-\A)^j , \quad d_j = \frac{(p+q - j)! q!}{(p+q)! j! (q - j)!}.
    \end{equation}
    The non-singularity of $D_{pq}(\A)$ is assured if $p$ and $q$ are large enough or if the eigenvalues of $\A$ are all negative \cite{moler2003nineteen,saff1975zeros}. 
\end{defn}
Since Pad\'{e} approximation is a rational function, the matrix encoding method from~\cite{berry2017quantum} cannot be applied directly. We propose a novel approach to encode rational polynomials into linear systems with a special structure. In the case where $\A$ is Hermitian and negative semi-definite, this structure allows us to bound the condition number of the resulting linear system without relying on the usual restriction $\|\A h\|_2 \le 1$, which is typically required in related methods~\cite{berry2014high,berry2017quantum,krovi2023improved,dong2024investigationquantumalgorithmlinear}.

\subsection{Constructing the linear system}
Replacing the Taylor approximation by Pad\'e approximation, the procedure \eqref{eq:proc_taylor} becomes
\begin{equation}\label{eq:proc}
    \begin{aligned}
        \x(h) &\approx R_{pq}(\A h) \x(0) + \left(R_{pq}(\A h) - \I_n\right)\A^{-1}\bb =: \widehat \x(h),\\
        \x(2h) &\approx R_{pq}(\A h) \widehat \x(h) + \left(R_{pq}(\A h) - \I_n\right)\A^{-1}\bb =: \widehat \x(2h),\\
        &\vdots \\
        \x(mh) &\approx R_{pq}(\A h) \widehat \x\left((m-1)h\right) + \left(R_{pq}(\A h) - \I_n\right)\A^{-1}\bb =: \widehat \x(T).\\ 
    \end{aligned}
\end{equation}
We first consider the encoding in one step. Given $\widehat \x((s-1)h)$, the encoding of the $s$-th step is given by,
\begin{equation}\label{eq:first_step}
    \widehat \x(sh) = R_{pq}(\A h) \widehat \x((s-1)h) + \left(R_{pq}(\A h) - \I_n\right)\A^{-1}\bb, \quad \forall s = 1, \dots, m.
\end{equation}
Substituting $R_{pq}(\A h) = N_{pq}(\A h)D_{pq}^{-1}(\A h)$ into \cref{eq:first_step}, we obtain
\[\begin{aligned}
    \widehat \x(sh) &= N_{pq}(\A h)\left(D_{pq}^{-1}(\A h) \widehat \x((s-1)h) + (D^{-1}_{pq}(\A h)-\I_n)\A^{-1}\bb\right) + (N_{pq}(\A h) - \I_n)\A^{-1}\bb,
\end{aligned}\]
Introducing an auxiliary vector
\[\vv := D_{pq}^{-1}(\A h) \widehat \x((s-1)h) + (D_{pq}^{-1}(\A h)-\I_n)\A^{-1}\bb,\]
we can rewrite \cref{eq:first_step} as a pair of equations:
\begin{equation}\label{eq:122}
    \begin{cases}
        \widehat \x((s-1)h) = D_{pq}(\A h) \vv + (D_{pq}(\A h) - \I_n)\A^{-1}\bb,\\
        \widehat \x(sh) = N_{pq}(\A h) \vv + (N_{pq}(\A h) - \I_n)\A^{-1}\bb.
    \end{cases}
\end{equation}
Following the approach from~\cite{berry2017quantum}, we can encode the pair of equations into two linear systems,
\begin{equation}\label{eq:lin1}
    \left[\begin{array}{cccc|c}
        \I_n&&&&\\
        \hline \beta_1\A h & \I_n &&&\\
        & \ddots & \ddots &&\\
        && \beta_q\A h& \I_n&\\
        -\I_n &\cdots & -\I_n& -\I_n&\I_n
    \end{array}\right]\begin{bmatrix} \z_0^{(s)} \\  \z_1^{(s)} \\ \vdots \\ \z_p^{(s)} \\ \widehat \x((s-1)h) \end{bmatrix} =  \begin{bmatrix}\vv\\ -d_1 h\bb\\ \boldsymbol 0\\ \vdots \\ \boldsymbol 0\end{bmatrix}
\end{equation}
and 
\begin{equation}\label{eq:forw}
    \begin{bmatrix}\I_n\\ -\alpha_1\A h & \I_n\\
        & \ddots & \ddots \\
        && -\alpha_p\A h& \I_n\\
        -\I_n &\cdots & -\I_n& -\I_n&\I_n\end{bmatrix}\begin{bmatrix}\widetilde \z_0^{(s)} \\ \widetilde \z_1^{(s)} \\ \vdots \\ \widetilde \z_p^{(s)} \\ \widehat \x(sh) \end{bmatrix} =  \begin{bmatrix}\vv\\ n_1 h\bb\\ \boldsymbol 0\\ \vdots \\ \boldsymbol 0\end{bmatrix},
\end{equation}
respectively, where we use the following notations for convenience,
\begin{equation}\label{eq:alpha_beta}
    \begin{aligned}
        \begin{cases}
            \alpha_0 := 1,\\
            \alpha_{j+1} := \frac{n_{j+1}}{n_j} = \frac{p-j}{(j+1)(p+q-j)}, &j = 0,\dots, p-1,
        \end{cases} \\
        \begin{cases}
            \beta_0 := 1,\\
            \beta_{j+1} := \frac{d_{j+1}}{d_j} = \frac{q-j}{(j+1)(p+q-j)}, &j = 0,\dots, q-1.
        \end{cases}
    \end{aligned} 
\end{equation}
To combine the linear systems from \cref{eq:lin1} and \cref{eq:forw}, we begin by rewriting \cref{eq:lin1} as follows
\begin{equation} \label{eq:back}
    \begin{aligned}
        \begin{cases}
            \z_0^{(s)} = \vv,\\
            \begin{bmatrix} 
                \beta_1\A h& \I_n\\ &\ddots&\ddots \\ &&\beta_q\A h& \I_n\\ -\I_n & \cdots &-\I_n&-\I_n\end{bmatrix} 
            \begin{bmatrix}
                \z_0^{(s)}\\ \z_1^{(s)}\\ \vdots \\ \z_q^{(s)}
            \end{bmatrix} = \begin{bmatrix}
                -d_1h\bb \\ \boldsymbol 0 \\ \vdots  \\ \boldsymbol 0 \\ -\widehat \x((s-1)h)
            \end{bmatrix},
        \end{cases}
    \end{aligned}
\end{equation}
which symbolically treats $\vv$ as an unknown and $\widehat \x((s-1)h)$ as a known component. Next, combining \cref{eq:forw} with the second system from \cref{eq:back}, we obtain the following unified system
\begin{equation}\label{eq:backfor}
    \begin{bmatrix}
        \I_n & \I_n &\cdots &\I_n&\\ \I_n & \beta_q\A h&&&\\ &\ddots&\ddots& \\ && \I_n&\beta_1\A h&\\
        &&& -\alpha_1 \A h & \I_n\\
        &&&& \ddots & \ddots\\
        &&&&& -\alpha_p \A h & \I_n\\
        &&& -\I_n & \cdots & -\I_n & -\I_n & \I_n
    \end{bmatrix}\begin{bmatrix}\z_q^{(s)} \\ \z_{q-1}^{(s)}\\ \vdots\\ \z_0^{(s)} = \vv \\ \widetilde \z_1^{(s)} \\ \vdots \\ \widetilde \z_p^{(s)} \\\widehat \x(sh)\end{bmatrix} = \begin{bmatrix}\widehat \x((s-1)h)\\ \boldsymbol 0\\ \vdots \\ -d_1 h\bb \\ n_1 h \bb \\ \boldsymbol 0 \\ \vdots \\ \boldsymbol 0\end{bmatrix}.
\end{equation}
Since the diagonal Pad\'{e} approximation ($p=q$) is usually preferred over the off-diagonal cases ($p\not=q$) \cite{moler2003nineteen}, we consider the case $p=q=k$ in the rest of the paper. We present the following finding, whose proof is provided in \Cref{pf:lem3.2}.
\begin{lem}\label{obs:zz}
    In the case when $p = q = k$, we have $\z_j^{(s)} = (-1)^j \widetilde \z_j^{(s)}$ for all $j = 1, \dots, k$ and $s = 1, \dots, m$.
\end{lem}
Using \Cref{obs:zz}, we can omit the computation of the terms $\widetilde \z_j^{(s)}$, simplifying the linear system in \cref{eq:backfor} to
\begin{equation}\label{eq:sim_backfor}
    \begin{bmatrix}
        \I_n & \I_n &\cdots &\I_n&\\ \I_n & \beta_k\A h&&&\\ &\ddots&\ddots&&\\ && \I_n&\beta_1\A h&\\
        (-\I_n)^{k+1}& (-\I_n)^k&\cdots& -\I_n& \I_n
    \end{bmatrix}\begin{bmatrix}\z_k^{(s)} \\ \z_{k-1}^{(s)}\\ \vdots\\       \z_0^{(s)}\\ \widehat \x(s h)\end{bmatrix} = \begin{bmatrix}\widehat \x((s-1)h)\\ \boldsymbol 0\\ \vdots \\\boldsymbol 0\\ -d_1 h\bb \\ \boldsymbol 0 \end{bmatrix}.
\end{equation}
To ensure that the norm of the coefficient matrix in \cref{eq:sim_backfor} is bounded independent of $k$, we multiply both the first and the last block rows of the matrix by the factor $\frac1{\sqrt{k+1}}$. This modification also benefits the block-encoding and the condition number of the final linear system. With this modification, the $s$-th step \cref{eq:first_step} of the procedure is encoded into the following linear system
\begin{equation}
    \begin{cases}
        \W_k(\A h) \z^{(s)} = \widetilde \y_s,\\
        \frac{\widehat \x(sh)}{\sqrt{k+1}} + \frac{\widetilde{\boldsymbol 1}^T \otimes \I_n}{\sqrt{k+1}} \z^{(s)} = \boldsymbol 0,
    \end{cases}
\end{equation}
where 
\begin{equation}\label{eq:WT}
    \W_k(\A h):= \begin{bmatrix} \frac{1}{\sqrt{k+1}}\I_n & \frac{1}{\sqrt{k+1}}\I_n &\cdots &\frac{1}{\sqrt{k+1}}\I_n\\ \I_n & \beta_k\A h\\ &\ddots&\ddots \\ && \I_n&\beta_1\A h \end{bmatrix}
\end{equation}
and
\begin{equation}\label{eq:zy1}
    \z^{(s)} = \begin{bmatrix}
        \z_k^{(s)}\\ \z_{k-1}^{(s)}\\ \vdots \\ \z_0^{(s)}
    \end{bmatrix}\in \mathbb C^{n(k+1)}, \quad \widetilde \y_s := \begin{bmatrix}\frac{\widehat \x((s-1)h)}{\sqrt{k+1}}\\\boldsymbol 0\\ \vdots \\ \boldsymbol 0\\ -d_1 h \bb\end{bmatrix} \in \mathbb C^{n(k+1)}, \quad \boldsymbol {\widetilde 1} := \begin{bmatrix}(-1)^{k+1}\\(-1)^k\\ \vdots \\ -1\end{bmatrix} \in \mathbb R^{k+1}.
\end{equation}
Finally, we combine $m$ steps in the procedure \cref{eq:proc} together and obtain the entire linear system
\begin{equation}\label{eq:full_W}
    \begin{tikzpicture}
        \draw(0,0)node{$\underbrace{\left[
            \begin{array}{cccccccc}
                \W_k(\A h) & & & & & & & \\
                \widetilde \T & \W_k(\A h) & & & & & & \\
                & \ddots & \ddots & & & & &\\
                && \widetilde \T & \W_k(\A h) & & & &\\
                &&& \frac{\widetilde {\boldsymbol 1}^T \otimes \I_n}{\sqrt{k+1}}& \frac{\I_n}{\sqrt{k+1}}& & &\\
                &&&& -\I_n & \I_n& &\\
                &&&&& \ddots & \ddots & \\
                &&&&&& -\I_n & \I_n  
            \end{array}
            \right]}_{\displaystyle =:\LL_{m,k,p}(\A h)} \begin{bmatrix} \z^{(1)} \\ \z^{(2)} \\ \vdots \\ \z^{(m)}\\ \widehat \x \\ \widehat \x\\\vdots \\ \widehat \x\end{bmatrix} = \begin{bmatrix}\widetilde \y_1 \\\widetilde \y_0\\\vdots \\\widetilde \y_0\\ \boldsymbol 0 \\\boldsymbol 0\\\vdots \\ \boldsymbol 0\end{bmatrix}
            $};
        \draw(-6.5,1.5)node[left]{$m \text{ blocks}\left\{\rule{0mm}{10mm}\right.$};
        \draw(-1.8,-0.7)node[left]{$p \text{ blocks}\left\{\rule{0mm}{10mm}\right.$};
    \end{tikzpicture}
\end{equation}
where the vector $\widehat \x$ is repeated $p$ times to boost the success probability, and
\begin{equation}\label{eq:1111}
    \quad \widetilde \T := \frac{\left(\e_1 \boldsymbol {\widetilde 1}^T\right) \otimes\I_n}{\sqrt{k+1}}, \quad \e_1 := \begin{bmatrix}1\\0\\\vdots \\0\end{bmatrix}\in \mathbb R^{k+1}, \quad \widetilde \y_1 := \begin{bmatrix}\frac{\x(0)}{\sqrt{k+1}}\\\boldsymbol 0\\ \vdots \\ \boldsymbol 0\\ -d_1 h \bb\end{bmatrix} \in \mathbb C^{n(k+1)}, \quad \widetilde \y_0 := \begin{bmatrix}\boldsymbol 0\\ \boldsymbol 0\\ \vdots \\ \boldsymbol 0\\ -d_1 h \bb\end{bmatrix}\in \mathbb C^{n(k+1)}. 
\end{equation}
In the following discussion, we use $\LL_{m,k,p}(\A h)$ to refer to the coefficient matrix in the linear system given by \cref{eq:full_W}. It is important to note that $\LL_{m,k,p}(\A h)$ and $\CC_{m,k,p}(\A h)$, as defined in \cref{eq:full_M}, share the same block structure and block size. 
\subsection{Upper bounds of condition number}\label{sec:cond_main}
In this section, we provide upper bounds for the condition number of $\LL_{m,k,p}(\A h)$, under the assumption that either $\A$ is Hermitian and negative semi-definite, or $\|\A h\|_2 \le 1$. The key step is to derive an upper bound for the inverse of $\W_k(\A h)$. The proofs of the relevant lemmas are provided in Appendices \ref{sec:construction} and \ref{pf:lem3.45}. To simplicity, we drop $h$ in $\W_k(\A h)$ in these lemmas. 
\begin{lem}\label{lem:invert}
    The matrix $\W_k(\A)$ is non-singular if and only if $D_{kk}(\A)$ is non-singular.
\end{lem}
\begin{lem}\label{lem:W_inv_norm}
    Suppose $\A \in \mathbb C^{n\times n}$ is Hermitian and negative semi-definite, then 
    \begin{equation}\label{eq:W_inv_norm}
        \left\|\W_k(\A)^{-1}\right\|_2 \le \sqrt{(k+1)(4\log(k+1) + 1)}.
    \end{equation}
\end{lem}
\begin{lem}\label{lem:le1_W_inv_norm}
    Suppose $\A\in \mathbb C^{n\times n}$ satisfies $\|\A\|_2\le 1$, then
    \[\left\|\W_k(\A)^{-1}\right\|_2\le \frac{2\sqrt e}{3-e}\sqrt{(k+1)(4\log(k+1) + 1)}.\]
\end{lem}
Using the above lemmas, we derive upper bounds for the inverse and the condition number of $\LL_{m,k,p}(\A h)$ in the following theorem, whose proof is given in \Cref{sec:cond}.
\begin{thm}\label{thm:cond}
    Let $\A\in \mathbb C^{n\times n}$, $T>0$, and the parameters $m$ and $k\ge 3$ be chosen such that 
    \begin{equation}\label{eq:cond_cond}
        \left\|\I - e^{-i\A h}R_{kk}^i(\A h)\right\|_2 \le 1, \quad \forall i=1, \dots, m,
    \end{equation}
    where $h = T/m$. Denoting by $\kappa$ the condition number of $\LL_{m,k,p}(\A h)$. Then
    \begin{enumerate}
        \item if $\A$ is Hermitian and negative semi-definite, then
        \begin{equation}
            \left\|\LL_{m,k,p}(\A h)^{-1}\right\|_2 \le 6(m+p)\sqrt {k\log  k}
        \end{equation}
        and
        \begin{equation}
            \kappa \le 3(m+p) \sqrt{k\log  k}\left(6 + \|\A h\|_2 \right).
        \end{equation}
        \item if $\|\A h\|_2 \le 1$, then
        \[ \left\|\LL_{m,k,p}(\A h)^{-1}\right\|_2 = \mathcal O\left(C(\A)(m+p)\sqrt {k\log  k}\right), \]
        and 
        \[\kappa = \mathcal O\left(C(\A) (m+p) \sqrt{k \log  k}\right),\]
        where $C(\A) := \max_{t\in[0, T]}\left\|\exp(\A t)\right\|_2$.
    \end{enumerate}
    In summary, both of the two cases lead to
    \[\kappa = \mathcal O\left(C(\A) (m + p)\sqrt{k \log k} \|\A h\|_2 \right).\]
\end{thm}
\begin{rem}\label{rem:cond}
    One advantage of the Pad\'e approximation over the Taylor approximation is that when $\A$ is Hermitian and negative semi-definite, the spectral norm of the inverse of $\W_k(\A h)$ and $\LL_{m,k,p}(\A h)$ are independent of $\|\A h\|_2$. In contrast, for their counterparts, $\M_k(\A h)$ and $\C_{m,k,p}(\A h)$, the spectral norms may grow exponentially with $\|\A h\|_2$. To illustrate this, observe that
        \[ \left\|\C_{m,k,p}^{-1}(\A h)\right\|_2 \ge \left\|\M_k^{-1}(\A h)\right\|_2, \]
        and
        \[ \M_k^{-1}(\A h)\cdot\left(e_1\otimes \I_n\right) = \begin{bmatrix}
            \I_n\\
            \A h\\
            \frac{(\A h)^2}{2!}\\
            \vdots\\
            \frac{(\A h)^k}{k!}
        \end{bmatrix}.\]
        When $\A$ is Hermitian, we then have 
        \begin{equation}\label{eq:exp}
            \|\M_k^{-1}(\A h)\|_2 \ge \sqrt{\sum_{j=0}^k \frac{\|\A h\|_2^{2j}}{(j!)^2}}\ge \frac1{\sqrt{k+1}} \sum_{j=0}^k \frac{\|\A h\|_2^{j}}{j!} \approx \frac{\exp\left(\|\A h\|_2\right)}{\sqrt{k+1}}.
        \end{equation}
\end{rem}

\subsection{Block-encoding of \texorpdfstring{$\LL_{m,k,p}(\A h)$}{}}\label{sec:block_encode}
In this section, we implement the block-encoding of $\LL_{m,k,p}(\A h)$. The definition of block-encoding is provided in \Cref{sec:lem_block_encoding}. We assume:
\begin{itemize}
    \item The parameters $n$, $m$, $k+1$, and $p$ are powers of two, specifically $n = 2^{\mathfrak n}$, $m = 2^{\mathfrak m}$, $k + 1 = 2^{\mathfrak k}$, and $p = m \cdot(k+1) = 2^{\mathfrak m + \mathfrak k}$. The corresponding registers are denoted as $|\cdot\rangle_{\mathfrak n}$, $|\cdot\rangle_{\mathfrak m}$, and $|\cdot\rangle_{\mathfrak k}$, where the subscript indicates the number of qubits in this register. Additionally, a superscript $a$ denotes that the register contains ancilla qubits.
    \item An $(\alpha,\mathfrak d)$-block-encoding of $\A$ is available, which is denoted by $U_{\A}$.
\end{itemize}
Before constructing the block-encoding, we summarize the total cost in the following theorem.
\begin{thm}\label{thm:block_encoding}
    Given an $(\alpha, \mathfrak d)$-block-encoding of $\A$, there exists a $\left(4\cdot \max\{\alpha h, 1\}, \mathfrak d + 5\right)$-block-encoding of $\LL_{m,k,p}(\A h)$. The gate complexity is
    \begin{equation}\label{eq:gate_comp}
        \mathcal O\left(k + \poly  \log \left(m\cdot  k\right)\right),
    \end{equation}
    and it requires one query to the block-encoding of $\A$.
\end{thm}
Compared to the block-encoding presented in \cite{krovi2023improved}, ours is simpler, with both the normalization factor and the number of additional qubits being independent of $k$ and $m$. With minor modifications, the block-encoding proposed in this paper can also be applied to $\C_{m,k,p}(\A h)$, which shares a similar block structure to $\LL_{m,k,p}(\A h)$.

We begin by constructing the block-encoding for the case $h = 1$ and $\alpha = 1$, and will later generalize it for arbitrary $h > 0$ and $\alpha \ge \|\A \|_2$. To start, we rewrite the matrix as a sum of two matrices:
\begin{equation}\label{eq:L}
    \begin{aligned}
        \LL_{m,k,p}(\A) &={\footnotesize \begin{bmatrix}\W\\&\ddots\\&&\W \\&&&\frac{\I_n}{\sqrt{k+1}}\\&&& -\I_n & \I_n\\&&&&\ddots&\ddots\\ &&&&&-\I_n &\I_n \end{bmatrix} + \begin{bmatrix}
            \0\\\widetilde \T&\0\\&\ddots &\ddots \\&& \widetilde \T&\0\\&&& \frac{\widetilde {\boldsymbol 1}^T \otimes \I_n }{\sqrt{k+1}}& \0_n\\
            &&&&\0_n&\0_n\\&&&&&\ddots&\ddots\\ &&&&&& \0_n&\0_n  \end{bmatrix}}\\
        &=: \LL_1 + \LL_2,
    \end{aligned}
\end{equation}
where we use $\W$ to represent $\W_k(\A)$ for simplicity. Note that the nonzero components of $\widetilde \T$ are confined to its first block row, which equals $\frac{\widetilde {\boldsymbol 1}^T \otimes \I_n }{\sqrt{k+1}}$, allowing us to express $\LL_2$ as a tensor product. In the following two subsections, We will discuss the block-encoding of $\LL_1$ and $\LL_2$, and thus the block-encoding of $\LL_{m,k,p}(\A)$ can be implemented using LCU (Linear Combination of Unitaries).
\subsubsection{Block-encoding of \texorpdfstring{$\LL_1$}{}}
The first term $\LL_1$ can be written as
\[ \LL_1 = \begin{bmatrix}
    \mathcal W\\
    & \mathcal B    
\end{bmatrix}, \]
where $\mathcal W, \mathcal B \in \mathbb R^{mn(k+1)\times mn(k+1)}$, and explicitly
\begin{equation}\label{eq:WB}
    \mathcal W = \I_{m} \otimes \W, \quad \mathcal B = \begin{bmatrix}\frac{1}{\sqrt{k+1}}\\-1&1\\&\ddots&\ddots\\&&-1&1\end{bmatrix}\otimes \I_n . 
\end{equation}
The matrix $\W$ can be decomposed as
\[ \W = \M_1\otimes \I_n + \M_2\otimes \I_n + \M_3 \otimes \A, \]
where
\[ \M_1 = \begin{bmatrix}
    0\\
    1&0\\
    &\ddots & \ddots\\
    &&1&0
\end{bmatrix}, \quad \M_2 = \begin{bmatrix}
    \frac{1}{\sqrt{k+1}}&\frac{1}{\sqrt{k+1}}&\cdots & \frac{1}{\sqrt{k+1}}\\
    0&0\\
    &\ddots & \ddots\\
    &&0&0
\end{bmatrix}, \quad \M_3 = \begin{bmatrix}0\\&\beta_k\\ &&\ddots \\ &&& \beta_1\end{bmatrix}. \]
We now consider the block-encoding for $\M_1$, $\M_2$, and $\M_3$:
\begin{itemize}
    \item For matrix $\M_1$, its block-encoding $U_{\M_1}$ should satisfy
    \[ U_{\M_1}|0\rangle_1^a |j\rangle_{\mathfrak{k}} = \begin{cases}
        |0\rangle_1^a \left|(j+1)\bmod (k+1)\right\rangle_{\mathfrak k}, & j\not= k,\\
        |1\rangle_1^a \left|(j+1)\bmod (k+1)\right\rangle_{\mathfrak k}, & j = k. 
    \end{cases} \]    
    The corresponding quantum circuit is shown in \Cref{bc:m1}.
    \begin{figure}[H]
        \centering
        \[\Qcircuit @C=1em @R=.7em{
            \lstick{\ket{\cdot}_1^a}& \gate{X} & \qw & \qw\\
            \lstick{\ket{\cdot}_{\mathfrak k}}& \ctrl{-1} & \gate{ADD}& \qw
        }\]
        \caption{A quantum circuit implementing the block-encoding of $\M_1$, denoted by $U_{\M_1}$.}
        \label{bc:m1}
    \end{figure}
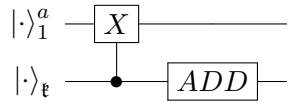
    \noindent The addition module performs $(j+1)\bmod (k+1)$, which is a unitary operation. Thus, we can verify that 
    \[ \langle 0|_1^a\langle i|_{\mathfrak{k}} U_{\M_1}|0\rangle_1^a |j\rangle_{\mathfrak{k}} = \begin{cases}0& i\not = (j+1)\bmod (k+1),\\
        0& i=0, j=k,\\
        1& i=j+1, j\not= k.\end{cases} \]
    Therefore, $U_{\M_1}$ is a $(1,1)$-block-encoding of $\M_1$, and the gate complexity for implementing $U_{\M_1}$ is $O(\log(k))$.
    \item For matrix $\M_2$, we have 
    \[ \M_2 = \begin{bmatrix}1\\&0\\&&\ddots\\ &&& 0\end{bmatrix} \begin{bmatrix}\frac{1}{\sqrt{k+1}}&\frac{1}{\sqrt{k+1}}& \cdots & \frac{1}{\sqrt{k+1}}\\ *&*&\cdots&*\\\vdots&\vdots&\cdots&\vdots\\ *&*&\cdots&* \end{bmatrix}. \] 
    Since $k+1 = 2^{\mathfrak k}$ is a power of $2$, the second matrix can be chosen as $H^{\otimes \mathfrak{k}}$, where $H$ is the Hadamard gate. The block-encoding of the first matrix is straightforward to implement in quantum circuits. Therefore, the full block-encoding $U_{\M_2}$ can be constructed using \Cref{lem:prod}, and the corresponding quantum circuit is shown in \Cref{bc:m2}. It is a $(1,1)$-block-encoding of $\M_2$, with gate complexity $O(\log(k))$.
    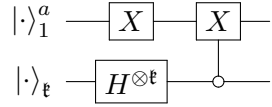
\begin{figure}[H]
        \centering
        \[\Qcircuit @C=1em @R=.7em{
            \lstick{\ket{\cdot}_1^a} & \gate{X} & \gate{X} & \qw\\
            \lstick{\ket{\cdot}_{\mathfrak k}} & \gate{H^{\otimes \mathfrak k}} & \ctrlo{-1}& \qw
        }\]
        \caption{A quantum circuit that implements the block-encoding of $\M_2$, denoted by $U_{\M_2}$.}
        \label{bc:m2}
    \end{figure}
    \item For matrix $\M_3$, note that $|\beta_j| < 1$, and we can compute $\beta_j$ explicitly. In fact, 
    \[ \beta_j = \frac{k-j+1}{j (2k-j+1)}, \quad j = 1, \dots, k+1. \]
    Thus, we can construct a quantum circuit to perform
    \[\begin{aligned}
        O_\beta |0\rangle_1^a |j\rangle_{\mathfrak k} &= \left(\beta_{k+1-j} |0\rangle_1^a + \sqrt{1 -\beta_{k+1-j}^2} |1\rangle_1^a \right) |j\rangle_{\mathfrak k}, \forall j = 0, \dots, k,
    \end{aligned}
    \]
    using a uniformly controlled rotation. This operation can be further simplified using the method proposed in \cite{PhysRevLett.93.130502,9951292}. Specifically, we need $k+1$ single-qubit gates and $k+1$ CNOT gates. It follows that $O_\beta$ is indeed a block-encoding of $\M_3$, as
    \[ \langle 0|_1^a\langle i|_{\mathfrak k} O_\beta |0\rangle_1^a |j\rangle_{\mathfrak k} = \delta_{ij} \beta_{k+1-j}. \]
    Setting $U_{\M_3} = O_{\beta}$, we conclude that $U_{\M_3}$ is a $(1, 1)$-block-encoding of $\M_3$ with gate complexity $O(k)$.
\end{itemize}
Using the unitaries $U_{\M_1}$, $U_{\M_2}$ and $U_{\M_3}$, we obtain a $(3, \mathfrak d + 3)$-block-encoding of $\W$ through LCU, as illustrated in \Cref{bc:W}, where $\zeta = 2\arccos \frac{\sqrt 6}{3}$. The gate complexity of this block-encoding is $O(k)$, except for the implementation of the controlled $U_{\A}$.
\begin{figure}[H]
    \centering
    \[\Qcircuit @C=1em @R=.7em{
        \lstick{\ket{\cdot}_1^a}&\gate{Z} & \gate{R_Y(\zeta)} & \ctrlo{1} & \ctrlo{1} & \ctrl{2} & \ctrl{4} & \gate{Z} & \gate{R_Y(\zeta)} & \qw\\
        \lstick{\ket{\cdot}_1^a}&\qw & \gate{H} & \ctrlo{1} & \ctrl{1} &\qw & \qw & \qw& \gate{H} & \qw\\
        \lstick{\ket{\cdot}_1^a}&\qw &\qw & \multigate{1}{U_{\M_1}} & \multigate{1}{U_{\M_2}} & \multigate{1}{U_{\M_3}} &\qw &\qw &\qw &\qw\\
        \lstick{\ket{\cdot}_{\mathfrak k}}&\qw &\qw & \ghost{U_{\M_1}} & \ghost{U_{\M_2}} & \ghost{U_{\M_3}} & \qw &\qw &\qw & \qw\\
        \lstick{\ket{\cdot}_{\mathfrak d}^a}&\qw &\qw &\qw &\qw &\qw & \multigate{1}{U_{\A}} &\qw &\qw & \qw\\
        \lstick{\ket{\cdot}_{\mathfrak n}}&\qw &\qw &\qw &\qw &\qw & \ghost{U_{\A}} &\qw &\qw & \qw
    }\]
    \caption{A quantum circuit that implements the block-encoding of $\W$, denoted by $U_{\W}$.}
    \label{bc:W}
\end{figure}
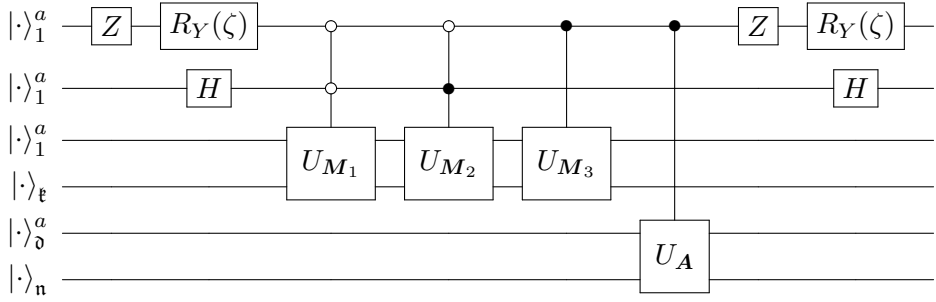

Next, we consider the block-encoding of $\mathcal B$ in \cref{eq:WB}.
The bi-diagonal matrix can be decomposed as
\[  \begin{bmatrix}\frac{1}{\sqrt{k+1}}\\-1&1\\&\ddots&\ddots\\&&-1&1\end{bmatrix} =  \begin{bmatrix}0\\-1&0\\&\ddots&\ddots\\&&-1&0\end{bmatrix} + \begin{bmatrix}\frac{1}{\sqrt{k+1}}\\&1\\&&\ddots\\&&&1\end{bmatrix} =: \M_4 + \M_5,\]
where $\M_4\in \mathbb R^{p\times p}$ is similar to $\M_1$ but with different size and opposite signs. To introduce the negative sign, we use a negative Pauli gate $Z$. The quantum circuit for $U_{\M_4}$ is shown in \Cref{bc:m4}
\begin{figure}[H]
    \centering
    \[\Qcircuit @C=1em @R=.7em{
        \lstick{\ket{0}_1^a} & \gate{X} & \gate{-Z} & \qw\\
        \lstick{\ket{\cdot}_{\mathfrak k + \mathfrak m}} & \ctrl{-1} & \gate{ADD} & \qw
    }\]
    \caption{A quantum circuit that implements the block-encoding of $\M_4$, denoted by $U_{\M_4}$.}
    \label{bc:m4}
\end{figure}
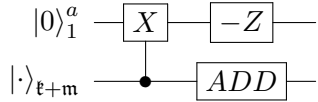
\noindent Thus, $U_{\M_4}$ is a $(1,1)$-block-encoding of matrix $\M_4$ with gate complexity $O(\log(m\cdot k))$.

The second matrix is diagonal and can be implemented using the circuit in \Cref{bc:m5}.
\begin{figure}[H]
    \centering
    \[\Qcircuit @C=1em @R=.7em{
        \lstick{\ket{0}_1^a} & \gate{R_Y(\theta_0)} & \qw\\
        \lstick{\ket{\cdot}_{\mathfrak k + \mathfrak m}} & \ctrlo{-1} & \qw
    }\]
    \caption{A quantum circuit that implements the block-encoding of $\M_5$, denoted by $U_{\M_5}$.}
    \label{bc:m5}
\end{figure}
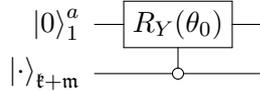
\noindent The parameter $\theta_0 = 2\arccos \frac{1}{\sqrt{k+1}}$, and $U_{\M_5}$ is a $(1,1)$-block-encoding of $\M_5$ with gate complexity $O(\log(m\cdot k))$. Next, using LCU, we obtain the block-encoding of $U_{\M_4} + U_{\M_5}$, which provides a $(2,2)$-block-encoding of $\mathcal B$. To match the $(3,\mathfrak d + 3)$-block-encoding of $\W$, we add a qubit, yielding a $(3, 3)$-block-encoding of $\mathcal B$ using \Cref{lem:adjust}, as shown in \Cref{bc:B}, where $\theta_1 = 2 \arccos \frac23$ such that $\cos \frac{\theta_1}2 = \frac23$. The gate complexity for the block-encoding of $\mathcal B$ is also $O(\log(m\cdot k))$.
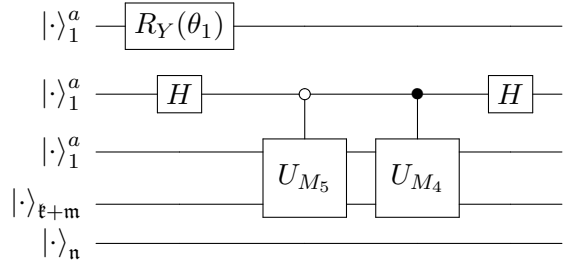
\begin{figure}[H]
    \centering
    \[\Qcircuit @C=1em @R=.9em{
        \lstick{\ket{\cdot}_1^a}& \gate{R_{Y}(\theta_1)} & \qw &\qw &\qw &\qw \\
        \lstick{\ket{\cdot}_1^a}& \gate{H} & \ctrlo{1} & \ctrl{1}& \gate{H} & \qw\\
        \lstick{\ket{\cdot}_1^a}& \qw & \multigate{1}{U_{M_5}} & \multigate{1}{U_{M_4}}& \qw & \qw \\
        \lstick{\ket{\cdot}_{\mathfrak k + \mathfrak m}}& \qw & \ghost{U_{M_5}} & \ghost{U_{M_4}} & \qw & \qw\\
        \lstick{\ket{\cdot}_{\mathfrak n}}& \qw & \qw &\qw &\qw &\qw
    }\]
    \caption{A quantum circuit that implements the block-encoding of $\mathcal B$, denoted by $U_{\mathcal B}$.}
    \label{bc:B}
\end{figure}
Combining the block-encoding of $\mathcal W$ and $\mathcal B$, we obtain a $(3, \mathfrak d + 3)$-block-encoding of $\LL_1$ with gate complexity $O(k + \log(m\cdot k))$. The corresponding quantum circuit is shown in \Cref{bc:l1}.
\begin{figure}[H]
    \centering
    \[\Qcircuit @C=1em @R=.7em{
        \lstick{\ket{\cdot}_{1}}& \ctrlo{1} & \ctrl{1} & \qw\\
        \lstick{\ket{\cdot}_{\mathfrak d + 3}^a} & \multigate{1}{U_{\mathcal W}} & \multigate{1}{U_{\mathcal B}} & \qw\\
        \lstick{\ket{\cdot}_{\mathfrak k + \mathfrak m + \mathfrak n}} & \ghost{U_{\mathcal W}} & \ghost{U_{\mathcal B}} & \qw
    }\]
    \caption{A quantum circuit that implements the block-encoding of $\LL_1$, denoted by $U_{\LL_1}$.}
    \label{bc:l1}
\end{figure}

\subsubsection{Block-encoding of \texorpdfstring{$\LL_2$}{}}
The second term $\LL_2$ in \cref{eq:L} can be written as 
\[ \begin{aligned}
    \LL_2 
    &= \left(\begin{bmatrix}
        \I_{m}\\&\0_{m}
    \end{bmatrix}\otimes \widetilde \T\right)\cdot \left(\begin{bmatrix}0&&&1\\ 1 & 0\\ & \ddots & \ddots \\ && 1 & 0\end{bmatrix}\otimes \I_{n(k+1)}\right) =: \LL_{2}^{(1)} \LL_2^{(2)}
\end{aligned}, \]
where 
\[ \widetilde \T = \frac1{\sqrt{k+1}}\begin{bmatrix}(-1)^{k+1}& (-1)^{k}&\cdots & -1\\0&0&\cdots&0\\\vdots&\vdots&\cdots&\vdots\\ 0&0&\cdots&0\end{bmatrix}\otimes \I_n=: \M_6\otimes \I_n, \quad \text{and}\quad  \M_7 := \begin{bmatrix}\I_m\\& \0_m\end{bmatrix}. \]
Similar to $\M_2$, we have 
\[ \M_6 = \begin{bmatrix}
    1\\&0\\&&\ddots \\&&&0
\end{bmatrix} \begin{bmatrix}\frac{(-1)^{k+1}}{\sqrt{k+1}}&\frac{(-1)^k}{\sqrt{k+1}}& \cdots & \frac{-1}{\sqrt{k+1}}\\ *&*&\cdots&*\\\vdots&\vdots&\cdots&\vdots\\ *&*&\cdots&* \end{bmatrix}\]
Since $k+1 = 2^{\mathfrak k}$, we have $(-1)^{k+1} = 1$, and the following observation.
\begin{rem}
    The second row of $H^{\otimes \mathfrak k}$ is of the form 
    \[\begin{bmatrix}
        \frac{1}{\sqrt{k+1}}& \frac{-1}{\sqrt{k+1}}&\cdots &\frac{1}{\sqrt{k+1}}& \frac{-1}{\sqrt{k+1}}
    \end{bmatrix}. \]
\end{rem}
Using this observation, we can express $\M_6$ as 
\[ \M_6 = \begin{bmatrix}1\\&0\\&&\ddots\\&&&0\end{bmatrix} \left(H^{\otimes (\mathfrak k - 1)} \otimes (XH)\right). \]
Define $\widetilde H_{\mathfrak k} := \left(H^{\otimes (\mathfrak k - 1)} \otimes (XH)\right)$, the quantum circuit in \Cref{bc:m6} implements a $(1, 1)$-block-encoding of $\M_6$ with gate complexity $O(\log(k))$.
\begin{figure}[H]
    \centering
    \[\Qcircuit @C=1em @R=.7em{
        \lstick{\ket{\cdot}_1^a}& \gate{X} & \gate{X} & \qw\\
        \lstick{\ket{\cdot}_{\mathfrak k}}& \gate{\widetilde H_{\mathfrak k}} & \ctrlo{-1}& \qw
    }\]
    \caption{A quantum circuit that implements the block-encoding of $\M_6$, denoted by $U_{\M_6}$.}
    \label{bc:m6}
\end{figure}
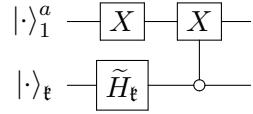
Moreover, the matrix $\M_7$ can be viewed as the top-left block of the following matrix 
\[ \left[\begin{array}{cc|cc}
    \I_m & & \0_m &\\
    & \0_m & & \I_m\\
    \hline
    \0_m & & \I_m & \\
    & \I_m & & \0_m
\end{array}\right] = \begin{bmatrix}1&0&0&0\\0&0&0&1\\0&0&1&0\\0&1&0&0\end{bmatrix} \otimes \I_m, \]
which can be implemented using a CNOT gate. The corresponding quantum circuit is shown in \Cref{bc:m7}.
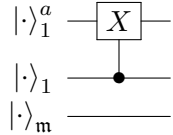
\begin{figure}[H]
    \centering
    \[ \Qcircuit @C=1em @R=1.2em{
        \lstick{\ket{\cdot}_1^a}& \gate{X} & \qw\\
        \lstick{\ket{\cdot}_1}& \ctrl{-1} & \qw\\
        \lstick{\ket{\cdot}_{\mathfrak m}}& \qw & \qw\\
    }\]
    \caption{A quantum circuit that implements the block-encoding of $\M_7$, denoted by $U_{\M_7}$.}
    \label{bc:m7}
\end{figure}
With the block-encodings of $\M_6$ and $\M_7$, we can obtain the block-encoding of matrix $\LL_2^{(1)}$ using \Cref{lem:2.1}. Noting that $\LL_2^{(2)}$ can be implemented by an addition module on the register $|\cdot\rangle_{\mathfrak m + 1}$, and applying \Cref{lem:prod}, the matrix $\LL_2$ can be implemented using the quantum circuit shown in \Cref{bc:l2}. This circuit provides a $(1, 2)$-block-encoding with gate complexity $O(\log(m\cdot k))$.
\begin{figure}[H]
    \centering
    \[\Qcircuit @C=1em @R=.7em{
        \lstick{\ket{\cdot}_1^a}& \qw &\qw \barrier[-1.8em]{6} & \qw &\qw \barrier[-1.8em]{6}& \gate{X} & \gate{X}& \qw& \qw \\
        \lstick{\ket{\cdot}_1^a}& \qw &\qw & \gate{X} & \qw & \qw & \qw &\qw & \qw\\
        \lstick{\ket{\cdot}_1}& \multigate{1}{ADD} &\qw & \ctrl{-1}& \qw &\qw & \qw &\qw & \qw\\
        \lstick{\ket{\cdot}_{\mathfrak m}}& \ghost{ADD}&\qw & \qw & \qw &\qw&\qw & \qw & \qw\\
        \lstick{\ket{\cdot}_{\mathfrak k}}& \qw &\qw & \qw &\qw & \gate{\widetilde H_{\mathfrak k}}& \ctrlo{-4} & \qw& \qw\\
        \lstick{\ket{\cdot}_{\mathfrak n}} &\qw &\qw & \qw & \qw&\qw & \qw & \qw& \qw\\
        & \dstick{U_{\LL_2^{(1)}}} & & \dstick{U_{\M_7}} & & & \dstick{U_{\M_6}}
    }\]
    \caption{A quantum circuit that implements the block-encoding of $\LL_2$, denoted by $U_{\LL_2}$.}
    \label{bc:l2}
\end{figure}
\subsubsection{Block-encoding of \texorpdfstring{$\LL_{m,k,p}(\A h)$}{}}
In summary, we now have a $(3,\mathfrak d + 3)$-block-encoding of $\LL_1$ and a $(1, 2)$-block-encoding of $\LL_2$. Using the following circuit, we can obtain a $(4, \mathfrak d + 4)$-block-encoding of $\LL_{m,k,p}(\A) = \LL_1 + \LL_2$ with gate complexity $O(k + \log(m\cdot k))$.
\begin{figure}
    \centering
    \[\Qcircuit @C=1em @R=.7em{
    \lstick{\ket{\cdot}_1^a}& \gate{Z} & \gate{R_{Y}(\theta_2)}& \ctrlo{1}& \ctrl{1}& \gate{Z}& \gate{R_Y(\theta_2)}& \qw\\
    \lstick{\ket{\cdot}^a_{\mathfrak d + 3}}& \qw & \qw & \multigate{1}{U_{\LL_1}}& \multigate{1}{U_{\LL_2}}& \qw & \qw& \qw\\
    \lstick{\ket{\cdot}_{\mathfrak k + \mathfrak m + \mathfrak n + 1}}& \qw & \qw & \ghost{U_{\LL_1}}& \ghost{U_{\LL_2}}& \qw & \qw &\qw}   
    \]
    \caption{A quantum circuit that implements the block-encoding of $\LL_{m,k,p}(\A)$.}
    \label{bc:L}
\end{figure}
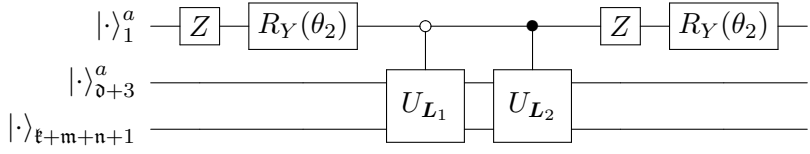
We choose $\theta_2 = \pi / 3$ such that 
\[ R_Y(\theta_2) Z = \begin{bmatrix}\frac{\sqrt3}{2} &-\frac12 \\ \frac12&\frac{\sqrt3}{2}\end{bmatrix}\begin{bmatrix}1&\\&-1\end{bmatrix} =\begin{bmatrix}\frac{\sqrt3}{2} &\frac12 \\ \frac12&-\frac{\sqrt3}{2}\end{bmatrix}. \]
Currently, we have implemented the block-encoding of $\LL_{m,k,p}(\A)$ based on a given $(1,\mathfrak d)$-block-encoding of $\A$. However, in our application, we need to implement the block-encoding of $\LL_{m,k,p}(\A h)$, and we may only have an $(\alpha, \mathfrak d)$-block-encoding of $\A$. If $h = \frac1\alpha$, we can use this block-encoding as before. However, if $h \not =  1/\alpha$, some modifications are needed:
\begin{itemize}
    \item If $h < \frac1\alpha$, using \Cref{lem:adjust}, we obtain a $(\frac1h, \mathfrak d + 1)$-block-encoding of $\A$, which can be directly applied in the framework discussed above. As a result, we will achieve a $(4, \mathfrak d + 5)$-block-encoding of $\LL_{m,k,p}(\A h)$.
    \item If $h > \frac1 \alpha$, we apply \Cref{lem:adjust} to all of the block-encodings in the above construction, except $U_{\A}$, scaling them by a factor of $\frac1{\alpha h}$. As a result, we will obtain a $(4\alpha h, \mathfrak d + 5)$-block-encoding of $\LL_{m,k,p}(\A h)$.
\end{itemize}
Hence, we have completed the constructive proof of \Cref{thm:block_encoding}.

\section{Complexity analysis}\label{sec:CompAna}
With the help of the upper bound of the condition number of $\LL_{m,k,p}(\A h)$, we proceed with the complexity analysis of the designed algorithm. We begin by discussing how to choose the parameters $m$ and $k$ to control the absolute error of the resulting approximation. Then, given fixed values of $m$ and $k$, we explore how to select $p$ such that the algorithm succeeds with constant positive probability.

\subsection{Approximation accuracy}\label{sec:accuracy}
Given $m$ and $k$, we aim to bound the errors
\begin{equation}\label{eq:abs_err}
    \left\|\widehat \x(ih) - \x(ih)\right\|_2, \quad \forall i = 1, \dots, m, 
\end{equation}
where $\x(ih)$ is the exact solution to \cref{eq:main} at time $ih$, and $\widehat \x(ih)$ is the approximation defined in \cref{eq:proc}. Throughout this section, we assume that $\|\A T\|_2 \ge 1$. We first bound these errors in a manner similar to the analysis in~\cite{berry2017quantum,krovi2023improved,dong2024investigationquantumalgorithmlinear}, for the case where $m = \lceil \|\A T\|_2\rceil$, implying that $\|\A h\|_2\le 1$. The result is summarized in the following lemma, with the proof provided in \Cref{pf:lem4.12}.
\begin{lem}\label{lem:sol_err0}
    Suppose $\|\A T\|_2 \ge 1$. Let $m = \left\lceil \|\A T\|_2 \right \rceil$ and $\delta\in \left(0,\frac{1}{m}\right)$. If $k$ satisfies  
    \begin{equation}\label{eq:cond_k}
        \frac{k! k!}{(2k)!(2k+1)!}\le \frac{\delta}{100},
    \end{equation}
    then we have 
    \begin{equation}
        \|\I - e^{-i\A h}R_{kk}^i(\A h)\|_2 \le \delta \|\A T\|_2, \quad \forall i =  1, \dots, m,
    \end{equation}
    and 
    \begin{equation}\label{eq:sol_err0}
        \left\|\widehat \x(ih) - \x(ih)\right\|_2 \le \delta T  \cdot \left( \|\A\|_2 \|\x(ih)\|_2 +  \|\bb\|_2\right), \quad \forall i = 1, \dots, m.
    \end{equation}
\end{lem}
Next, we bound these errors from another perspective. Given $k$ and the desired approximation accuracy, we aim to determine how large the step size $h$ can be. This may relax the condition that $\|\A h\|_2\le 1$. To achieve this, we draw a tighter bound on the remainder of Pad\'e approximation. As pointed out in \cite{higham2005scaling}, the remainder can be expressed as the following series,
\begin{equation}\label{eq:rho}
    \rho_k(x) := e^{-x} R_{kk}(x) - 1 = \sum_{j=2k+1}^\infty c_j x^j,
\end{equation}
where the series converges absolutely for $|x|<\min\{|t|: D_{kk}(t)=0\}=:\nu_k$, which is the smallest absolute value among the zero points of $D_{kk}(\cdot)$. Therefore, let 
\[\G := e^{-\A h}R_{kk}(\A h) - \I,\]
and we obtain the following bound
\begin{equation}\label{eq:cond_high-1}
    \left\|\G\right\|_2 = \left\|\rho_k\left(\A h\right)\right\|_2 \le \sum_{j=2k+1}^\infty |c_j| \theta^j =: f_k(\theta),
\end{equation}
where $\theta:=\|\A h\|_2 < \nu_k$. It is clear that if $\A$ is a general matrix and only $\|\A\|$ is known, then \cref{eq:cond_high-1} provides the tightest bound for $\|\G\|_2$. We state the following lemma, whose proof is provided in \Cref{pf:lem4.12}.

\begin{lem}\label{lem:sol_err}
    Let $\delta \in \left(0,\frac{1}{\|\A T\|_2}\right)$, with $\|\A T\|_2 \ge 1$, and let $m$ be chosen such that
    \begin{equation}\label{eq:cond_f}
        \frac{f_k(\theta)}{\theta} \le \frac{\delta}{e-1},
    \end{equation}
    where $\theta:= \|\A h\|_2$ and $f_k(\cdot)$ is defined in \cref{eq:cond_high-1}. Then, we have
    \begin{equation}
        \|\I - e^{-i\A h}R_{kk}^i(\A h)\|_2 \le \delta \|\A T\|_2, \quad \forall i =  1, \dots, m,
    \end{equation}
    and
    \begin{equation}\label{eq:sol_err}
        \left\|\widehat \x(ih) - \x(ih)\right\|_2 \le \delta T  \cdot \left( \|\A\|_2 \|\x(ih)\|_2 +  \|\bb\|_2\right), \quad \forall i = 1, \dots, m.
    \end{equation}
\end{lem}
In practice, we can evaluate $f_k(\theta)$ to high precision using symbolic computation, as demonstrated in \cite{higham2005scaling}. Let $\delta = 10^{-8}$. We apply the bisection method to find the largest $\theta$ that satisfies \cref{eq:cond_f}, denoted by $\theta_k$. The results, shown in \Cref{tab:theta_k}, indicate that $\|\A h\|_2$ can be significantly larger than $1$. Therefore, if the approximation order $k$ is predetermined, \Cref{lem:sol_err} allows us to choose $h$ without the restriction $\|\A h\|_2\le 1$. This is especially advantageous when $\A$ is Hermitian and negative semi-definite, as the condition number of $\LL_{m,k,p}(\A h)$ can be bounded without the $\|\A h\|_2\le 1$ condition.
\begin{table}[t]
    \renewcommand{\arraystretch}{2}
    \setlength{\tabcolsep}{4pt}
    \centering
    \begin{tabular}{l|c|c|c|c|c|c|c|c|c|c|c|c|c|c}\hline\hline
        $k$ & $5$ & $6$ & $7$ & $8$ & $9$ & $10$ & $11$ & $12$ & $13$ & $14$ & $15$ & $16$ & $17$ & $18$\\\hline 
        $\theta_k$ & $1.49$ & $2.36$ & $3.34$ & $4.40$ & $5.53$ & $6.69$ & $7.89$ & $9.11$ & $10.35$ & $11.61$ & $12.88$ & $14.16$ & $15.45$ & $16.74$ 
        \\\hline\hline
    \end{tabular}
    \caption{Maximal values $\theta_k$ of $\|\A h\|_2$ such that the condition \cref{eq:cond_f} is satisfied with $\delta = 10^{-8}$.}
    \label{tab:theta_k}
\end{table}

\subsection{Success probability}
In quantum computation, we are concerned with the success probability of the algorithm. In the case of solving \cref{eq:full_W}, the resulting vector is stored in a quantum superposition state. After measuring this superposition, we obtain the desired quantum state with a certain probability. The equation below illustrates how the resulting vector is stored in a quantum state,
\begin{equation}\label{eq:sol}
    \begin{aligned}
        \frac1C\begin{bmatrix} \z^{(1)} \\ \z^{(2)} \\ \vdots \\ \z^{(m)}\\ \widehat \x \\ \widehat \x \\ \vdots \\ \widehat \x
        \end{bmatrix} = \sum_{i=0}^{m-1}\sum_{j=0}^{k} \frac{\left\|\z_j^{(i)}\right\|_2}C \left|i(k+1) + j\right\rangle\left|\z_j^{(i)}\right\rangle + \sum_{a = 0}^{p-1} \frac{\left\|\widehat \x\right\|_2}C \left|m(k+1) + a\right\rangle \left|\widehat \x\right\rangle,
    \end{aligned}
\end{equation}
where $C = \sqrt{\sum_{i=1}^m \left\| \z^{(i)}\right\|_2^2 + p \|\widehat \x\|_2^2}$ is a normalization factor. In the rest of this paper, we refer to the qubits in the first ket as the index register, while the qubits in the second ket as the value register. We measure the index register and obtain a result greater than $m(k+1) - 1$ with probability 
\[\mathbb P_{succ} := \frac{ p \|\widehat \x\|_2^2}{\sum_{i=1}^m \left\| \z^{(i)}\right\|_2^2 + p \|\widehat \x\|_2^2}. \]
Conditioned on the measuring result of the index register being greater than $m(k+1)-1$, the quantum state of the value register is $|\widehat \x\rangle$, and the algorithm successfully provides the desired result.
\begin{thm}\label{thm:succ_prob}
    Let $\A\in \mathbb C^{n\times n}$, $T > 0$, and $\delta > 0$ satisfy the condition  
    \[\delta ' :=   \delta T \cdot \left(\left\|\A\right\|_2 + \frac{\|\bb\|_2}{\|\x(T)\|_2}\right) < \frac18.\]
    Then, the following results hold:
    \begin{enumerate}
        \item if $\A$ is Hermitian and negative semi-definite, and the conditions of \Cref{lem:sol_err} are satisfied, we have 
        \begin{equation}\label{eq:succ1}
            \mathbb P_{succ} \ge \frac12 \cdot \frac{p}{6 m g^2\left(h^2 + 1\right) + p},
        \end{equation}
        where 
        \[g := \frac{\max\left\{\max_{0\le t\le T}\left\|\x(t)\right\|_2, \|\bb\|_2\right\}}{\left\|\x(T)\right\|_2}.\]
        \item if $\A$ is an arbitrary matrix, $m$ is chosen such that $\|\A h\|_2 \le 1$, and the conditions of \Cref{lem:sol_err0} are satisfied, we have
        \begin{equation}\label{eq:succ2}
            \mathbb P_{succ} \ge \frac12 \cdot \frac{p}{204 m g^2 (h^2 + 1) + p}.
        \end{equation}
    \end{enumerate}
    In summary, for both scenarios, let $p = \left\lceil 6 m\left(1 + h^2 \right)\right\rceil$, then we have
    \[\mathbb P_{succ} \ge \frac12 \cdot \frac{6}{204 g^2 +6} \ge \frac{1}{70g^2}.\] 
\end{thm}
\noindent The proof of \Cref{thm:succ_prob} is provided in \Cref{sec:succ_prob}.
\subsection{Main Result}
Before stating our main result, we first fill a small gap between the QLSA proposed in \cite{costa2022optimal} and our application. The QLSA, as presented in \cite{costa2022optimal}, is as follows.
\begin{thm}[\cite{costa2022optimal}, Theorem 19]\label{thm:qlsa}
    Let $\LL$ be a matrix such that $\|\LL\|_2 = 1$ and $\|\LL^{-1}\|_2 = \kappa$. Given an oracle block-encoding of $\LL$ and an oracle for implementing $|\bb\rangle$, there exists a quantum algorithm that produces the normalized state $|\LL^{-1}\bb\rangle$ within an error $\epsilon$, using $\mathcal O(\kappa \log \frac1\epsilon)$ calls to the oracles.
\end{thm}
\Cref{thm:qlsa} requires $\|\LL\|_2 = 1$ and a $(1,*)$-block-encoding of $\LL$, but the matrix $\LL_{m,k,p}(\A h)$ does not satisfy these conditions. Specifically, its norm can exceed one, and its block-encoding, as given by \Cref{thm:block_encoding}, is not a $(1,*)$-block-encoding. To address this, suppose we have only an $(\alpha, *)$-block-encoding of $\LL$, where $\alpha \ge \|\LL\|_2$. We illustrate the impact of this modification in the case where $\LL$ is Hermitian and positive definite. To perform the adiabatic quantum simulation, which is the key step of the QLSA, the initial Hamiltonian is chosen as $H_0=\begin{bmatrix}0& Q_{\bb}\\ Q_{\bb}& 0\end{bmatrix}$, the same as in \cite{costa2022optimal}, where $Q_{\bb} = I - |\bb\rangle \langle \bb|$. However, the terminal Hamiltonian must be $H_1 = \begin{bmatrix}0& \frac1\alpha \LL Q_{\bb}\\ \frac1\alpha Q_{\bb}\LL& 0\end{bmatrix}$ to match the block-encoding of $\LL$. The Hamiltonian used in the adiabatic quantum computation is then of the form:
\[H(s)= (1-f(s))H_0 + f(s)H_1,\]
where $f(s):[0,1]\to [0,1]$ is the schedule function. Then the quantity $\Delta_0(s)$ (Definition 2 in \cite{costa2022optimal}), which indicates the eigenvalues' gap of $H(s)$, is given by 
\[\Delta_0(s) = 1 - f(s) + \frac{f(s)}{\alpha \|\LL^{-1}\|_2},\]
instead of $\Delta_0(s) = 1 - f(s) + \frac{f(s)}{\kappa}$ in \cite{costa2022optimal}. This shows that $\alpha \|\LL^{-1}\|_2$ plays the role of $\kappa$ in \Cref{thm:qlsa}. Therefore, under these settings, the query complexity for preparing the state $|\LL^{-1}\bb\rangle$ within error $\epsilon$ becomes $\mathcal O(\alpha \|\LL^{-1}\|_2 \log \frac 1\epsilon)$. In the ideal case, where the block-encoding satisfies $\alpha = \|\LL\|_2$, this simplifies to $\alpha \|\LL^{-1}\|_2 = \kappa$, and the query complexity matches the result of \Cref{thm:qlsa}.
\begin{thm}\label{thm:main}
    Suppose $\A\in \mathbb C^{n\times n}$ and an $(\alpha, \mathfrak d)$-block-encoding of $\A$, denoted by $U_{\A}$, is available. Let $\x_0$ and $\bb$ be $n$-dimensional vectors with known norms, and assume access to two controlled oracles, $O_{\x}$ and $O_{\bb}$, which prepare quantum states proportional to $\x_0$ and $\bb$, respectively. Let $\x(t)$ evolve according to the differential equation
    \begin{equation}
        \frac{\diff \x(t)}{\diff t} = \A \x(t) + \bb, \quad  t\in[0, T],
    \end{equation}
    with the initial condition $\x(0)=\x_0$. Without loss of generality, assume that $\|\A T\|_2 \ge 1$, and choose $m = \lceil \|\A T\|_2\rceil$. Under these conditions, the proposed algorithm produces a quantum state that is $\epsilon$-close to $\x(T)/\|\x(T)\|_2$ in $l^2$ norm, with a success probability of $\Omega (1)$, and includes a flag indicating success. The algorithm makes
    \begin{equation}\label{eq:query_complexity}
        \mathcal O\left((h^2 + 1)\cdot \alpha T \cdot C(\A) \cdot g \cdot \sqrt{k \log k} \cdot \log \frac{\sqrt {m(h^2 + 1)} g}{\epsilon}\right)
    \end{equation}
    queries to $U_{\A}$, $O_{\x}$, and $O_{\bb}$, where  
    \[\begin{aligned}
            & C(\A) := \max_{0\le t\le T}\left\|\exp(\A t)\right\|_2, \quad g := \frac{\max\left\{\max_{0\le t\le T}\left\|\x(t)\right\|_2, \|\bb\|_2\right\}}{\left\|\x(T)\right\|_2}, \\
            & k = \left\lceil\frac{\log M}{\log \log M} \right\rceil, \quad M = \frac{401 T }{\epsilon}\left( \|\A\|_2 + \frac{\|\bb\|_2}{\|\x(T)\|_2}\right), \quad \text{and} \quad \epsilon < \frac12. 
    \end{aligned}\]
    The gate complexity of this algorithm is larger than its query complexity by a factor of 
    \begin{equation}\label{eq:gate_complexity}
        \mathcal O\left( k + \poly\log \left(m\cdot k\right)\right).
    \end{equation}
\end{thm}
\begin{proof}
    Note that the choice of $k$ ensures that 
    \begin{equation}\label{cond:0}
        \frac{k! k!}{(2k)!(2k + 1)!} \le \frac1M,
    \end{equation}
    and $k\log k = \mathcal O(\log M)$. Next, let $\delta$ satisfy
    \[\frac{100}{M} \le \delta \le  \frac{\epsilon}{4T \left(\|\A\|_2 + \frac{\|\bb\|_2}{\|\x(T)\|_2}\right)} < \frac{1}{\|\A T\|_2},\]
    such that condition \cref{eq:cond_k} is satisfied. With the choice of $\delta$, we have 
    \[\begin{aligned}
        \quad \delta T \cdot \left(\left\|\A\right\|_2 + \frac{\|\bb\|_2}{ \|\x(T)\|_2}\right)\le \frac\epsilon 4 < \frac18.
    \end{aligned}\]
    Using \Cref{lem:sol_err0}, we have 
    \begin{equation}\label{eq:pri}
        \|\widehat \x(ih) - \x(ih)\|_2 \le \delta T \left(\|\A\|_2 \|\x(ih)\|_2 + \|\bb\|_2\right),\quad \forall i = 1, \dots, m,
    \end{equation}
    and 
    \begin{equation}\label{eq:pri2}
        \|\I - e^{-i\A h}R_{kk}^i(\A h)\|_2 \le \delta \|\A T\|_2 < 1, \quad \forall i =  1, \dots, m.
    \end{equation}
    Therefore, the condition \eqref{eq:cond_cond} of \Cref{thm:cond} is satisfied. We then set $p =  \left\lceil 6m\left(1 + h^2 \right)\right\rceil$ and construct the linear system \cref{eq:full_W}. Choosing $i = m$ in \eqref{eq:pri}, we obtain
    \[\frac{\left\|\widehat \x(T) - \x(T)\right\|_2}{\|\x(T)\|_2} < \frac\epsilon 4,\]
    where $\widehat \x(T)$ is the vectors corresponding to $\widehat \x$ in \cref{eq:full_W}. Using \Cref{lem:berry13}, we then have
    \begin{equation}\label{eq:err_noQLSA}
        \left\|\left|\x(T)\right\rangle - \left|\widehat \x(T)\right\rangle\right\|_2\le \frac\epsilon 2.
    \end{equation}
    
    We use the QLSA proposed in \cite{costa2022optimal} to solve the linear system, which may introduce errors. This method requires the block-encoding of the linear operator $\LL_{m,k,p}(\A h)$, and in \Cref{sec:block_encode}, we propose a way to block-encode it using a single query to $U_{\A}$. Additionally, the right-hand-side vector in the linear system \cref{eq:full_W} can be formed with a constant number of calls to $ O_{\x}$ and $O_{\bb}$, as shown in~\cite{berry2017quantum,krovi2023improved}. According to \Cref{thm:cond}, \Cref{thm:block_encoding}, and \Cref{thm:qlsa}, to obtain an $\epsilon'$-close solution to the linear system \cref{eq:full_W}, the query complexity of the QLSA is given by
    \begin{equation}\label{eq:complex_eps'}
        \mathcal O\left( \alpha h \cdot C(\A) \cdot  (m+p)\cdot \sqrt{k \log k} \cdot \log \frac{1}{\epsilon'}\right)= \mathcal O\left((h^2 + 1)\cdot \alpha T \cdot C(\A)\cdot \sqrt{k \log k} \cdot \log \frac{1}{\epsilon'} \right),
    \end{equation}
    where we set $p =  \left\lceil 6 m\left(1 + h^2 \right)\right\rceil$ and $h = T/m$ to arrive at the right-hand side. After solving the linear system, we measure the index register of the quantum state \cref{eq:sol} in the standard basis. Conditioned on the outcome being in 
    \[S:= \left\{m(k+1), m(k+1) + 1, \cdots, m(k+1) + p -1\right\},\]
    we output the state of the value register. We will show that it is sufficient to choose $\epsilon' = \mathcal O\left(\frac\epsilon{\sqrt{m(h^2 + 1)} g}\right)$, such that the probability of this event occurring is $\Omega(1/g^2)$, and the output state is $\epsilon / 2$-close to the state $|\widehat \x(T)\rangle$.
    
    \Cref{eq:sol} gives the normalized exact solution to the linear system \cref{eq:full_W}. Let $d = m (k+1) + p$, and define $\x_l = \z_j^{(i)}$ for $l = i(k+1) + j$. The normalized exact solution can then be expressed as
    \[|\x\rangle = \sum_{l=0}^d \gamma_l |l\rangle |\x_l\rangle,\]
    where $\gamma_l = \frac{\|\x_l\|_2}{\|\x\|_2}$. Note that $|\x_l\rangle = |\widehat \x(T)\rangle$ for any $l\in S$, and using \Cref{thm:succ_prob}, we have 
    \[\gamma_l = \frac{\|\widehat \x\|_2}{\|\x\|_2} \ge \frac{1}{c\sqrt m g}, \quad \forall l\in S,\]
    where $c = \sqrt{420(h^2 + 1)}$. Now suppose the QLSA outputs the state 
    \[|\x'\rangle = \sum_{l=0}^d \gamma_l' |l\rangle |\x_l'\rangle \]
    which satisfies
    \[\left\||\x\rangle - |\x'\rangle \right\|_2 \le \epsilon'. \]
    Then, for any $l\in S$, by \Cref{lem:berry14}, we have 
    \[\||\x_l\rangle - |\x_l'\rangle \|_2 \le \frac{2\epsilon'}{\gamma_l - \epsilon'}. \]
    Choosing $\epsilon' = \frac{\epsilon}{5c \sqrt m g}$, we get 
    \[ \frac{2\epsilon'}{\gamma_l - \epsilon'} \le \frac{\frac{2\epsilon}{5c \sqrt m g}}{\frac{1}{c \sqrt m g} - \frac{\epsilon}{5c \sqrt m g}} = \frac{2\epsilon}{5 - \epsilon} < \frac\epsilon 2.\]
    Thus, we have 
    \[ \||\x_l'\rangle - |\x(T)\rangle \|_2 \le \left\|\left|\x(T)\right\rangle - \left|\x_l\right\rangle\right\|_2 + \||\x_l\rangle - |\x_l'\rangle \|_2 < \epsilon, \quad \forall l \in S .\]
    Furthermore, by \Cref{lem:berry15}, we have 
    \[\gamma_l' \ge \gamma_l - \epsilon' \ge \frac{9}{10c\sqrt m g}.\]
    Therefore, if we measure the index register of $|\x'\rangle$ in the standard basis, the probability of getting outcome $l\in S$ is 
    \[\sum_{l\in S}|\gamma_l'|^2 \ge \frac{81p}{100c^2 m g^2} = \frac{81}{7000 g^2},\]
    and when this occurs, the state of the value register becomes $|\x_l'\rangle$, which is $\epsilon$-close to the desired state $|\x(T)\rangle$ in $l^2$ norm. Using amplitude amplification \cite{Brassard_2002}, we can raise this probability to $\Omega(1)$ with $\mathcal O(g)$ repetitions of the above procedure. 
    
    In summary, the query complexity is
    \[\mathcal O\left((h^2 + 1)\cdot \alpha T \cdot C(\A) \cdot g \cdot \sqrt{k \log k} \cdot \log \frac{\sqrt{m(h^2+1)} g}{\epsilon}\right).\]
    Finally, the total gate complexity is multiplied by the query complexity along with the gate complexity for block-encoding, as given in~\Cref{thm:block_encoding}. 
\end{proof}
\begin{cor}\label{cor}
    Suppose $\A\in \mathbb C^{n\times n}$ is Hermitian and negative semi-definite, and an $(\alpha, \mathfrak d)$-block-encoding of $\A$, denoted by $U_{\A}$, is available. Let $\x_0$ and $\bb$ be $n$-dimensional vectors with known norms, and assume access to two controlled oracles, $O_{\x}$ and $O_{\bb}$, which prepare quantum states proportional to $\x_0$ and $\bb$, respectively. Let $\x(t)$ evolve according to the differential equation
    \begin{equation}
        \frac{\diff \x(t)}{\diff t} = \A \x(t) + \bb, \quad  t\in[0, T],
    \end{equation}
    with the initial condition $\x(0)=\x_0$. Without loss of generality, assume that $\|\A T\|_2 \ge 1$, and let the parameter $m$ be chosen to satisfy the condition \eqref{eq:cond_f} in \Cref{lem:sol_err}. Under these conditions, the proposed algorithm produces a quantum state that is $\epsilon$-close to $\x(T)/\|\x(T)\|_2$ in $l^2$ norm, with a success probability of $\Omega (1)$, and includes a flag indicating success. The algorithm makes \begin{equation}\label{eq:query_complexity2}
        \mathcal O\left((h^2 + 1)\cdot \alpha T \cdot g \cdot \sqrt{k \log k} \cdot \log \frac{\sqrt {m(h^2 + 1)} g}{\epsilon}\right)
    \end{equation}
    queries to $U_{\A}$, $O_{\x}$, and $O_{\bb}$, where the parameters $g$ and $k$ are defined as in \Cref{thm:main}. The gate complexity is larger than its query complexity by a factor of
    \begin{equation}\label{eq:gate_complexity2}
        \mathcal O\left( k + \poly\log \left(m\cdot k\right)\right).
    \end{equation}
\end{cor}
\begin{rem}
    The difference between \Cref{thm:main} and \Cref{cor} lies in the assumption on $\A$. In \Cref{thm:main}, $\A$ is a general matrix where as in \Cref{cor}, $\A$ is a Hermitian and negative semi-definite matrix. When $\A$ is a general matrix, we choose $m = \lceil \|\A T\|_2\rceil$ such that $\|\A h\|_2 = \left\|\A \frac{T}{m}\right\|_2 \le 1$. When $\A$ is a Hermitian and negative semi-definite matrix, we apply \Cref{lem:sol_err} to choose $m \le \lceil \|\A T\|_2 \rceil$ such that the bounds in \eqref{eq:pri} and \eqref{eq:pri2} remain valid. By proceeding with the same construction and applying the results tailored to this special case, we obtain the same asymptotic complexity as in the general case, with a potentially smaller constant prefactor. Notably, in this setting $C(\A) = 1$, and thus it does not appear in the final query complexity.
\end{rem}

\section{Comparison with previous methods}\label{sec:NumExp}
In this section, we compare our method with two other approaches: the method based on Taylor approximation~\cite{berry2017quantum,krovi2023improved,dong2024investigationquantumalgorithmlinear} and the method based on linear combination of Hamiltonian simulation (LCHS)~\cite{An2023QuantumAF}, which achieves near-optimal dependence on all parameters. First, we compare the theoretical query complexity of the three methods. Then, we conduct numerical comparisons between our method and the Taylor approximation-based method in two distinct scenarios.
\begin{itemize}
    \item \textbf{Scenario I:} For a given desired precision $\epsilon$ and a fixed approximation order $k$, we compare the smallest values of $m$ required by both methods to achieve the condition $\frac{\|\widehat \x - \x(T)\|_2}{\|\x(T)\|_2} < \epsilon$.
    \item \textbf{Scenario II:} For a given desired precision $\epsilon$ and $m = \left\lceil \left\|\A T\right\|_2\right\rceil$, we compare the smallest approximation order $k$ required by both methods to achieve the condition $\frac{\|\widehat \x - \x(T)\|_2}{\|\x(T)\|_2} < \epsilon$.
\end{itemize}
\subsection{Theoretical comparison}
We use the result from \cite{dong2024investigationquantumalgorithmlinear} as the theoretical query complexity for the method based on Taylor approximation. In \cite{dong2024investigationquantumalgorithmlinear}, an oracle $O_{\A}$ is used to compute the non-zero entries of $\A$. However, in our setting, we assume the availability of a $(\alpha, \mathfrak d)$-block-encoding of $\A$, denoted as $\U_{\A}$. To make their results and ours comparable, we modify their setting accordingly. In this modified setting, their algorithm requires
\begin{equation}\label{eq:old_query}
    \mathcal O\left(\alpha T \cdot C(\A) \cdot g \cdot k\cdot \log \frac{1}{\delta}\right),
\end{equation}
queries to $U_{\A}$, $O_{\x}$ and $O_{\bb}$, where the parameters are defined as
\[\delta = \frac{\epsilon}{25\sqrt{m} g}, \quad k = \left\lfloor \frac{2\log\Omega}{\log\log \Omega}\right\rfloor, \quad \Omega = \frac{2e^3T}{\delta} \left(\|\A\|_2 + \frac{e^2 \|\A T\|_2 \|\bb\|_2}{\|\x(T)\|_2}\right),\quad  m = \left\lceil\|\A T\|_2\right\rceil,\]
and $g = \frac{\max\left\{\max_{0\le t\le T}\|\x(t)\|_2, \|\bb\|_2\right\}}{\|\x(T)\|_2}$.\footnote{Note that the definition of $g$ in this paper differs slightly from that in~\cite{dong2024investigationquantumalgorithmlinear} and~\cite{krovi2023improved}. Specifically, equation (4.22) of \cite{krovi2023improved} should include an additional term, $\sum_{i=0}^m|i,1,h\bb\rangle$, and equation (5.87) should have an extra term, $mh^2\|\bb\|_2^2$, in the denominator. These necessary adjustments lead to the definition of $g$ adopted in this work.} By comparing the query complexity in \cref{eq:old_query} with that in \cref{eq:query_complexity}, our approach provides the following improvements.
\begin{itemize}
    \item The dependence on the approximation order $k$ is improved. This improvement arises from the factor $\frac1{\sqrt{k+1}}$ introduced in the linear system \cref{eq:full_W}. Due to the special structure of the linear system, incorporating this factor reduces the spectral norm of $\LL_{m,k,p}(\A h)$ by a factor of $\frac{1}{\sqrt{k+1}}$, while only increases the spectral norm of $\LL_{m,k,p}(\A h)^{-1}$ by a constant factor. Since $k = \widetilde{\mathcal O}\left(\log\left(\frac1\epsilon\right)\right)$, this improvement leads to a better dependence on precision, which will be later shown in \Cref{tab:comparison}.
    
    \item The parameter $\Omega$ is larger than the corresponding parameter $M$ in our analysis by a factor of $\sqrt m g$, as we use \cref{eq:err_noQLSA} instead of the condition $\left\|\left|\x(T)\right\rangle - \left|\widehat \x\right\rangle\right\|_2\le \delta$ (equation (118) in \cite{berry2017quantum}), where $\delta = \mathcal O\left(\frac{\epsilon}{\sqrt m g}\right)$. Additionally, $\Omega$ includes an extra term $\|\A T\|_2$ compared to $M$. We eliminate this term by employing an optimized upper bound on $\|\left(\I - e^{-i\A h}R^i_{kk}(\A h)\right)\A^{-1}\|_2$, which is applied in \Cref{lem:sol_err} and \Cref{lem:sol_err0}. The corresponding result in \cite{krovi2023improved} is Lemma 10.
    
    \item Even if we assume $M = \Omega$, the choice of $k$ in the original method is twice as large as our approach. This difference arises because the remainder of the Pad\'e approximation includes a factor of $\frac{k!k!}{(2k)!(2k+1)!}$, whereas the remainder of Taylor approximation only has a factor of $\frac1{(k+1)!}$. This improvement directly reduces the gate complexity, as the gate complexity for block-encoding $\LL_{m,k,p}(\A h)$ (or $\C_{m,k,p}(\A h)$) scales linearly with $k$.

    \item If matrix $\A$ is Hermitian and negative semi-definite, we can choose $m < \left\lceil \|\A T\|_2\right\rceil$ while still maintaining a theoretical guarantee for our algorithm. This advantage arises from the special structure of $\W_k(\A h)$. \Cref{lem:W_inv_norm} provides an upper bound on the spectral norm of $\W_k(\A h)^{-1}$ that is independent of $\|\A h\|_2$, whereas the spectral norm of $\M_k(\A h)$ may depend exponentially on $\|\A h\|_2$. A smaller $m$ directly reduces the total query complexity and gate complexity. More importantly, the numerical results indicate that a smaller $m$ is also associated with a lower condition number and a higher success probability $\mathbb P_{succ}$, which can significantly decrease the overall query complexity.
\end{itemize}
\newcolumntype{C}[1]{>{\centering}p{#1}}
\setlength{\parindent}{15pt}
\begin{table}[t]
    \renewcommand{\arraystretch}{2.5}
    \centering
    \begin{tabular}{m{3cm}<{\centering}|m{6cm}<{\centering}|m{4cm}<{\centering}}\hline\hline
         \textbf{Method} & \textbf{Queries to $\U_{\A}$} & \textbf{Queries to $O_{\x}$ and $O_{\bb}$} \\\hline 
        
         Taylor approximation~\cite{berry2017quantum,krovi2023improved,dong2024investigationquantumalgorithmlinear} & 
        \multicolumn{2}{c}{\large$\widetilde {\mathcal O}\left(\frac{\max\left\{\max_{0\le t\le T}\|\x(t)\|_2, \|\bb\|_2\right\}}{\|\x(T)\|_2} C(\A)\alpha T\left(\log\left(\frac1\epsilon\right)\right)^2\right)$}\\\hline
        
        LCHS~\cite{PhysRevLett.131.150603,An2023QuantumAF}  &  \large 
        $\widetilde {\mathcal O}\left(\frac{\|\x(0)\|_2 + \|\bb\|_2 T}{\|\x(T)\|_2} \alpha T \left(\log\left(\frac1{\epsilon}\right)\right)^{1/\beta}\right)$ & \large 
        $\mathcal O\left(\frac{\|\x(0)\|_2 + \|\bb\|_2 T}{\|\x(T)\|_2}\right)$ \\\hline
        
        Pad\'e approximation & 
        \multicolumn{2}{c}{\large $\widetilde{\mathcal O}\left(\frac{\max\left\{\max_{0\le t\le T}\|\x(t)\|_2, \|\bb\|_2\right\}}{\|\x(T)\|_2}C(\A)\alpha T\left(\log\left(\frac1\epsilon\right)\right)^{1.5}\right)$}
        \\\hline\hline
    \end{tabular}
    \caption{Comparison of the query complexity. These complexities apply to ODEs with time-independent $\A$ and $\bb$, with logarithm terms omitted except for $1/\epsilon$. Here, $\alpha \ge \|\A\|_2$, $T$ is the evolution time, $\epsilon$ is the desired accuracy, $\beta\in (0,1)$ is a chosen constant, and $C(\A) = \max_{0\le t\le T}\|\exp(\A t)\|_2$.}
    \label{tab:comparison}
\end{table}

\Cref{tab:comparison} compares our method with the other two methods in terms of query complexity, with logarithm terms omitted except for $1/\epsilon$. The method based on LCHS achieves the best query complexity to $O_{\x}$ and $O_{\bb}$. Regarding the query complexity to $\U_{\A}$, our method has a fixed dependence on precision, whereas the LCHS method's dependence is determined by the parameter $\beta$. Additionally, the query complexity of our method depends on the quantities $\max_{0\le t\le T}\|\x(t)\|_2$ and $C(\A)$, while the LCHS method depends on $\|\bb\|_2 T$. When $\beta = 2/3$, both methods exhibit the same dependence on precision, but the overall complexity is determined by the specific choice of $\A$, $\bb$, and $\x_0$.

\subsection{Numerical comparison}
In this section, we present numerical experiments to compare our method with the method based on Taylor approximation.
\subsubsection{Numerical result of scenario I}
\noindent \textbf{Experiment 1:} Let $T = 30$, $k = 9$, $p = 1$,
\[ \A = \begin{bmatrix}-2&1&0&0&0\\1&-2&1&0&0\\0&1&-2&1&0\\0&0&1&-2&1\\0&0&0&1&-2\end{bmatrix}, \quad\text{and }\quad  \x_0 = \bb = \begin{bmatrix}1\\1\\1\\1\\1\end{bmatrix}.\]
We construct $\C_{m,k,1}(\A h)$ and $\LL_{m,k,1}(\A h)$, and solve the corresponding linear systems for different values of $m$. The results, including relative solution error, condition number, and the success probability $\mathbb P_{succ}$, are shown in \Cref{fig:exp1a}. 
\begin{figure}
    \centering
    \includegraphics[width=1\linewidth]{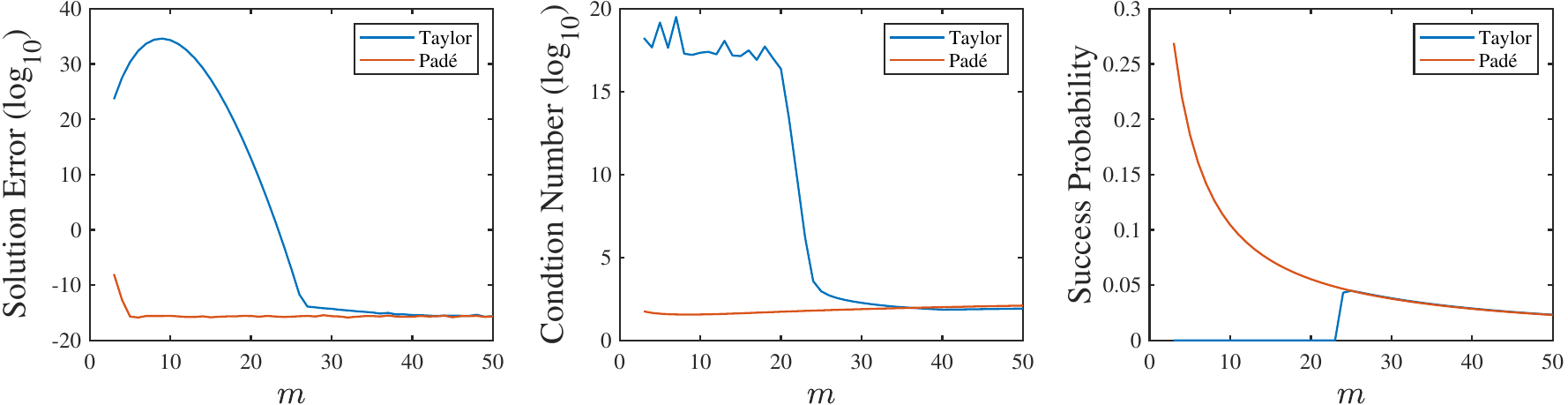}
    \caption{Comparison of the relative solution error, condition number, and success probability $\mathbb P_{succ}$ as $m$ increases, between the method based on Pad\'e approximation and the method based on Taylor approximation.}
    \label{fig:exp1a}
\end{figure}

The left figure demonstrates that the method based on Pad\'e approximation can achieve high relative accuracy with a relatively small $m$, while the method based on Taylor approximation requires a much larger $m$ to achieve similar accuracy. The middle figure reveals that the condition number of our newly proposed method is at a low level for all values of $m$, while the condition number of the Taylor approximation-based method can become extremely large if $m$ is not large enough. This observation aligns with our theoretical analysis. The right figure compares the success probability $\mathbb P_{succ}$. The method based on Pad\'e approximation achieves high success probabilities for smaller values of $m$, while the success probabilities of both methods converge to nearly the same value as $m$ increases sufficiently.

In summary, the first experiment demonstrates that the method based on Taylor approximation becomes unstable when $m$ is not large enough, while the method based on Pad\'e approximation provides satisfactory results even for relatively small $m$.

\vskip 1em
\noindent \textbf{Experiment 2:} The advantage of our newly proposed method becomes more pronounced as the time interval length $T$ increases. To demonstrate this, we fix the accuracy $\epsilon = 10^{-10}$ and the approximation order $k = 9$, and investigate how the optimal choice $m^*$ (the smallest $m$ required to achieve the desired accuracy) changes with $T$ for both methods. With the parameters $k = 9, p = 1$, and $m = m^*$, we construct and solve the corresponding linear systems. We also compute the condition number and success probability $\mathbb P_{succ}$ for these linear systems. The results are shown in \Cref{fig:exp1b}.
\begin{figure}
    \centering
    \includegraphics[width=1.0\linewidth]{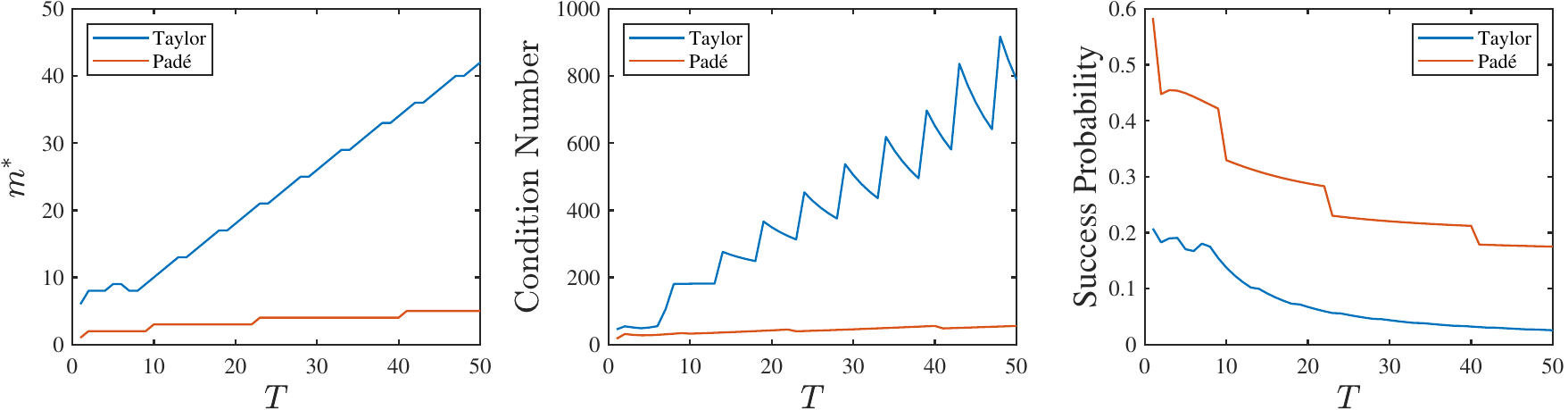}
    \caption{Comparison of the minimum choice of $m$ for both methods under the same accuracy requirement, with varying time interval length $T$. The condition number and success probability $\mathbb P_{succ}$ of the corresponding linear system at these values of $m$ are also compared.}
    \label{fig:exp1b}
\end{figure}

From the results, we observe that the method based on Pad\'e approximation consistently outperforms the method based on Taylor approximation, with the gap growing more significant as $T$ increases. A smaller $m^*$ directly benefits the block-encoding of the linear operator, since the gate complexity in \cref{eq:gate_comp} depends on $m$. Moreover, the query complexity of the QLSA is linearly dependent on the condition number of the corresponding linear system, so a smaller condition number significantly improves the solving process. In this experiment, we set $p = 1$. However, to further boost the success probability $\mathbb P_{succ}$, a suitable choice of $p > 1$ could be employed for both methods. A larger success probability in the $p=1$ case implies that the corresponding method requires a smaller $p$ to achieve the desired success probability. Furthermore, the condition number of $\LL_{m,k,p}(\A h)$ and $\C_{m,k,p}(\A h)$ depends linearly on the parameter $p$, and thus the efficiency of QLSA is thus also linearly dependent on $p$.

The first two experiments compare the two methods using a specific choice of matrix $\A$. However, since all Hermitian matrices are unitarily diagonalizable, the results from these experiments capture the behavior of all Hermitian and negative semi-definite matrices. These results support our theoretical analysis for this class of matrices: the condition number and success probability of the proposed method remain well-controlled, even when $m \ll \|\A T\|_2$. 

\vskip 1em
\noindent \textbf{Experiment 3:} The set of Hermitian and negative semi-definite matrices is a subset of the broader class of matrices whose eigenvalues have negative real parts. Within this broader class, the method based on Pad\'e approximation is also expected to outperform the method based on Taylor approximation. However, providing a theoretical analysis for this more general case is challenging. To investigate this further, we conduct a numerical comparison between the two methods on matrices from this broader class.

Let $\epsilon = 10^{-10}$, $k = 9$, $p = 1$ and $\x_0 = \bb = [1,1,1,1,1]^T$. We generate one hundred matrices $\A \in \mathbb C^{5\times 5}$ randomly, where each matrix's eigenvalues have negative real parts. For each matrix $\A$ and each time interval $T= 1,\dots, 50$, we determine the smallest value of $m$ that achieves the desired accuracy $\epsilon$, denoted by $m^*$. With the parameters $k = 9$, $p = 1$, and $m = m^*$, we construct the corresponding linear systems and compute the condition number and success probability for each linear system. Finally, we calculate the mean value and standard deviation over the one hundred samples. The results are plotted in \Cref{fig:exp2}.
\begin{figure}
    \centering
    \includegraphics[width=1.0\linewidth]{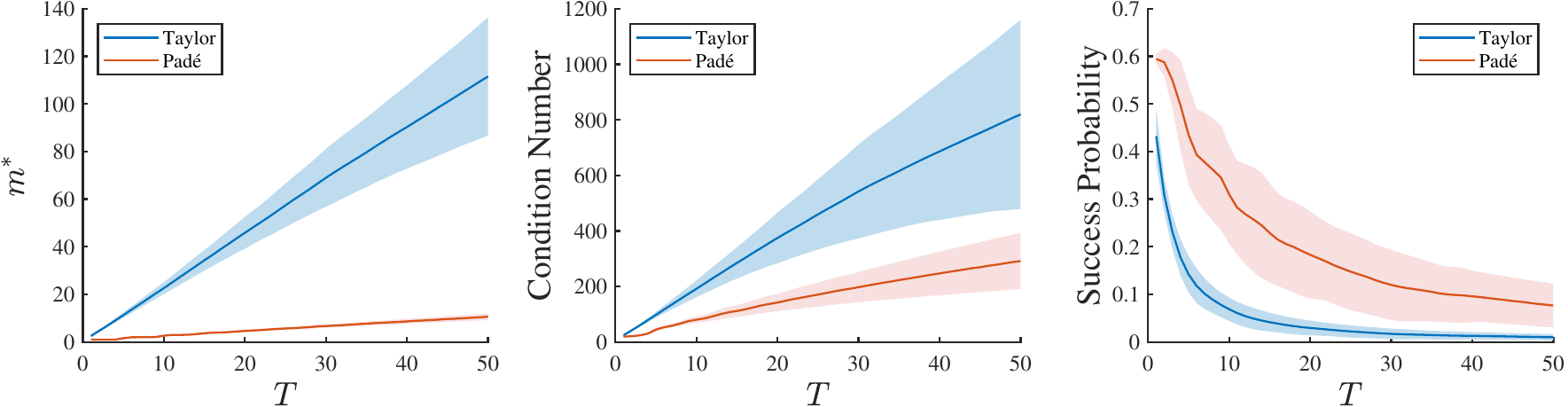}
    \caption{For varying time interval length $T$, and with the same accuracy requirement $\epsilon$ and approximation order $k$, the leftmost plot compares the average minimum choice of $m$ for the two methods. The middle plot compares the average condition number, and the rightmost plot compares the average success probability. The shaded area represents the mean value plus or minus the standard deviation.}
    \label{fig:exp2}
\end{figure}

The results in \Cref{fig:exp2} demonstrate that the Pad\'e approximation-based method consistently outperforms the Taylor approximation-based method across all three compared metrics, on average. Notably, the gap between the two methods widens as $T$ increases. This experiment highlights that the proposed method remains more efficient for a broader class of matrices.
\subsubsection{Numerical result of scenario II}
\noindent \textbf{Experiment 4:} In scenario II, we always choose $m = \left\lceil \|\A T\|_2\right\rceil$. Therefore, it suffices to compare the two methods for the case where $T = 1$ and $\|\A\|_2 = 1$. In this experiment, we set $m = p = 1$, $\x_0 = \bb = [1,1,1,1,1]^T$, and generate one hundred matrices $\A\in \mathbb C^{5\times 5}$ with eigenvalues having negative real parts and spectral norm equal $1$. For each $\A$ and a required precision $\epsilon$, we determine the smallest $k$ that achieves the required precision, denoted by $k^*$. With parameters $m=p=1$ and $k = k^*$, we construct the corresponding linear systems. We then evaluate the condition number of each linear system and calculate the mean value and standard deviation across the one hundred samples. The results are shown in \Cref{fig:exp3a}.
\begin{figure}
    \centering
    \includegraphics[width=0.7\linewidth]{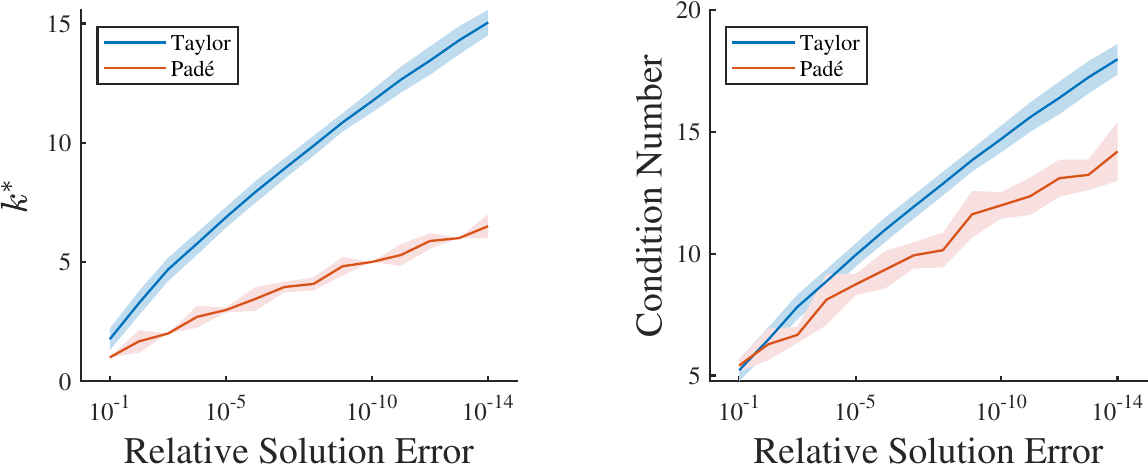}
    \caption{Comparison of the two methods for different precision requirements. The left plot compares the smallest $k$ needed to achieve the required precision for both methods, on average. The right plot compares the average condition number of the linear systems derived from both methods with the smallest choice of $k$ that achieves the required precision. The shaded area represents the mean value plus or minus the standard deviation.}
    \label{fig:exp3a}
\end{figure}

From the left plot of \Cref{fig:exp3a}, we observe that the average magnitude of $k^*$ required by the Pad\'e approximation-based method is roughly half that of the Taylor approximation-based method. This aligns with the third point in our theoretical comparison. The smaller $k$ benefits both the total query complexity (see \cref{eq:query_complexity}) and the gate complexity (as shown in \cref{eq:gate_complexity}). Additionally, the average condition number of the linear system based on the Pad\'e approximation is smaller, further supporting the efficiency of the Pad\'e approximation-based method.

\section{Conclusion}
\label{sec:conclusion}

We propose a quantum algorithm to solve linear autonomous ordinary differential equations (ODEs) of the form~\eqref{eq:main} using the Pad\'e approximation. By discretizing the time step-by-step with step size $h$, the solution to~\eqref{eq:main} can be expressed in terms of $e^{\A h}$. The proposed method then uses the Pad\'e approximation, a rational function of $\A h$, to approximate $e^{\A h}$. The matrix rational function is encoded into a large linear system $\LL_{m,k,p}(\A h)$, as shown in~\eqref{eq:full_W}, which does not explicitly involve matrix inversion. The matrix $\LL_{m,k,p}(\A h)$ has a sparse block structure, allowing it to be efficiently block encoded using an oracle query to $\A$. Details of the quantum circuit for encoding $\LL_{m,k,p}(\A h)$ can be found in \Cref{sec:block_encode}. By applying this block-encoding in QLSA, we obtain the overall quantum algorithm for solving the linear ODEs.

The complexity analysis of the proposed quantum algorithm is conducted in detail. The core of the analysis focuses on deriving asymptotic upper bounds for the condition number of $\LL_{m,k,p}(\A h)$. Unlike existing work, when the matrix $\A$ is negative semi-definite, our upper bound for the condition number is independent of $\|\A h\|_2$, which relaxes the constraint $\|\A h\|_2 < 1$ in the final complexity analysis. Combining the condition number analysis with the complexity analysis of QLSA and the success probability analysis, we obtain the overall complexity analysis, which includes both the oracle query complexity and the gate complexity. Compared to other algorithms of a similar nature~\cite{berry2017quantum,krovi2023improved,dong2024investigationquantumalgorithmlinear}, the proposed algorithm shows improved dependence on the approximation order $k$. Furthermore, thanks to the Pad\'e approximation, the approximating order $k$ can be chosen much smaller than that in the case of using Taylor approximation. Numerical experiments are conducted to validate the theoretical results and illustrate the complexity comparison of the proposed method with other algorithms of the same type.

Several straightforward future directions emerge from this work. First, we could combine the proposed Pad\'e approximation-based algorithm with ODE linearization techniques to tackle nonlinear ODEs. Additionally, by carefully selecting the time step $h$, the proposed quantum algorithm could be extended to handle non-autonomous ODEs; however, the corresponding complexity analysis would need to be revisited. Furthermore, Pad\'e approximation could be applied to other quantum ODE algorithms to enhance their performance, potentially achieving optimal dependency on all parameters and improving the prefactors.

\appendix
\section{Construction and condition number analysis of the linear system}
\subsection{Proof of \Cref{obs:zz}}\label{pf:lem3.2}
\begin{proof}
    From \cref{eq:backfor}, we can observe that 
    \[ \begin{cases}
        \z_1^{(s)} + \beta_1 \A h \z_0^{(s)} = -d_1 h\bb = -\frac12 h\bb,\\
        \z^{(s)}_{j+1} + \beta_{j+1}  \A h \z_{j}^{(s)} = \boldsymbol 0, \quad j = 1,\dots, k-1.
    \end{cases} \]
    Again by \cref{eq:backfor}, we obtain
    \[ \begin{cases}
        \widetilde \z_1^{(s)} - \alpha_1 \A h \z_0^{(s)}= n_1 h\bb = \frac12 h\bb ,\\
        \widetilde \z^{(s)}_{j+1} - \alpha_{j+1} \A h \widetilde  \z_{j}^{(s)} = \boldsymbol 0, \quad j = 1, \dots, k-1.
    \end{cases} \]
    Since $\alpha_j = \beta_j$ for all $j = 1, \dots, k$, we conclude that 
    \[\z_j^{(s)} = (-1)^j \widetilde \z_j^{(s)}, \forall j = 1, \dots, k.\]
\end{proof}
\subsection{Proof of \Cref{lem:invert}}\label{sec:construction}
\begin{proof}
    For matrix $\W_k(\A)$, we have the following decomposition:
    \begin{equation}\label{eq:LU}
        \begin{aligned}
            \begin{bmatrix}
                \I_n & \I_n & \I_n & \cdots & \I_n\\
                \I_n & \beta_k \A\\
                & \I_n & \beta_{k-1} \A\\
                & & \ddots & \ddots\\
                &&& \I_n & \beta_1\A 
            \end{bmatrix} &= \begin{bmatrix}
                \I_n & \I_n & \I_n - \beta_k \A & \cdots & \I_n + \sum_{j=2}^k \beta_2\cdots \beta_j (-\A)^{j-1}\\
                & \I_n \\ && \I_n\\ &&& \ddots \\ &&&& \I_n
            \end{bmatrix} \\
            & \quad \times \begin{bmatrix}
                0 & 0 & \cdots & 0 & \I_n + \sum_{j=1}^k \beta_1\cdots \beta_j (-\A)^{j}\\
                \I_n & \beta_k \A\\
                & \I_n & \beta_{k-1}\A\\
                && \ddots & \ddots \\
                &&& \I_n & \beta_1\A
            \end{bmatrix},
        \end{aligned}
    \end{equation}
    where 
    \[\I_n + \sum_{j=1}^k \beta_1\cdots \beta_j (-\A)^{j} = \sum_{j=0}^k d_j (-\A)^j = D_{kk}(\A).\]
    Therefore, matrix $\W_k(\A)$ is non-singular if and only if $D_{kk}(\A)$ is non-singular.
\end{proof}
\subsection{Proofs of \Cref{lem:W_inv_norm} and \Cref{lem:le1_W_inv_norm}}\label{pf:lem3.45}
\begin{lem}\label{lem:matrix_inv}
    Let $\A\in \mathbb R^{n\times n}$ be an invertible matrix and partitioned as
    \[\A = \begin{bmatrix}
        \A_{11}&\A_{12}\\\A_{21}&\A_{22}
    \end{bmatrix},\]
    where $\A_{11}\in \mathbb R^{n_1\times n_1}$, $\A_{22}\in \mathbb R^{n_2\times n_2}$, and $\A_{11}$ is assumed to be invertible. In this case, define $\boldsymbol S = \A_{22} - \A_{21}\A_{11}^{-1}\A_{12}$. Then, $\boldsymbol S$ is also invertible, and the inverse of $\A$ is given by 
    \[\A^{-1} = \begin{bmatrix}
        \A_{11}^{-1} + \A_{11}^{-1} \A_{12} {\boldsymbol S}^{-1} \A_{21} \A_{11}^{-1} & -\A_{11}^{-1} \A_{12} {\boldsymbol S}^{-1}\\
        -{\boldsymbol S}^{-1} \A_{21} \A_{11}^{-1} & {\boldsymbol S}^{-1}
    \end{bmatrix}.\]
\end{lem}
\begin{lem}\label{lem:block_norm}
    Let $\M$ be a matrix of the form 
    \[\M = \begin{bmatrix}\M_{11}&\cdots& \M_{1n}\\\vdots & & \vdots\\\M_{m1}&\cdots &\M_{mn}\end{bmatrix},\]
    where each $\M_{ij}$ is a matrix block. Define the matrix $\widetilde \M$ such that $\widetilde \M_{ij} \ge \|\M_{ij}\|_2$ for each $i, j$,. Then, we have the inequality
    \[\|\M\|_2 \le \|\widetilde \M\|_2.\]
\end{lem}
\begin{proof}
    Let $\x^T=[\x_1^T, \cdots, \x_n^T]$ be an arbitrary unit vector. Then, we have 
    \[\begin{aligned}
        \left\|\M \x\right\|_2 &= \left\|\begin{bmatrix}\sum_{j=1}^n\M_{1j}\x_j\\\vdots \\\sum_{j=1}^n \M_{mj}\x_j\end{bmatrix}\right\|_2 = \sqrt{\sum_{i=1}^m \left\|\sum_{j=1}^n \M_{ij}\x_j\right\|_2^2}\\
        &\le \sqrt{\sum_{i=1}^m \left(\sum_{j=1}^n \left\|\M_{ij}\right\|_2\left\|\x_j\right\|_2\right)^2}\\
        &\le \left\|\widetilde \M \begin{bmatrix}\|\x_1\|_2\\\vdots\\\|\x_n\|_2\end{bmatrix}\right\|_2 \le \|\widetilde \M\|_2.
    \end{aligned}\]
    By the definition of the matrix $2$-norm, we have $\|\M\|_2 \le \|\widetilde \M\|_2$.
\end{proof}
\begin{lem}\label{lem:Dkk-1}
    Let $D_{kk}(\cdot)$ be the denominator of the $(k,k)$ Pad\'e approximation given by \eqref{eq:pade}. Then, we have 
    \[D_{kk}(-1)\le \sqrt{e}\quad \text{for all}\quad k\in \mathbb N_+.\]
\end{lem}
\begin{proof}
    Using the fact that 
    \[\frac{(2k-j)! k!}{(2k)! (k - j)!} \le\left(\frac12\right)^j,\]
    we obtain 
    \[D_{kk}(-1) = \sum_{j=0}^k \frac{(2k-j)!k!}{(2k)!(k-j)!} \frac{1}{j!}\le \sum_{j=0}^k \left(\frac{1}{2}\right)^j\frac{1}{j!} \le \sqrt e.\]
\end{proof}
\begin{lem}\label{lem:1d_W_inv_norm}
    Let $\lambda > 0$. Then, we have the following inequality 
    \[\left\|\W_k(-\lambda)^{-1}\right\|_2 \le \sqrt{4(k+1)\log(k+1) + 1}.\]
\end{lem}
\begin{proof}
    There is a permutation matrix $P$ such that 
    \[\widetilde \W := P\W_k(-\lambda) = \left[\begin{array}{cccc|c}
        1 & -\beta_k \lambda&&\\
        & \ddots & \ddots&\\
        & & 1 & -\beta_{2} \lambda&\\
        &&& 1 & -\beta_1\lambda\\
        \hline
        \frac1{\sqrt{k+1}} & \cdots & \frac1{\sqrt{k+1}} & \frac1{\sqrt{k+1}} & \frac1{\sqrt{k+1}}
    \end{array}\right]=:\left[\begin{array}{c|c}B&-\beta_1\lambda e_k\\ \hline \frac1{\sqrt{k+1}} \boldsymbol 1^T & \frac1{\sqrt{k+1}}
    \end{array}\right].\]
    It suffices to prove that 
    \[\|\widetilde \W^{-1}\|_2 \le \sqrt{4(k+1)\log(k+1) + 1}.\]
    By \Cref{lem:matrix_inv}, the matrix $\widetilde \W^{-1}$ is given by
    \begin{equation}\label{eq:tilde_W_inv}
        \widetilde \W^{-1} = \begin{bmatrix}B^{-1} - \frac{\beta_1 \lambda B^{-1}e_k \boldsymbol 1^T B^{-1}}{1 + \beta_1 \lambda \boldsymbol 1^T B^{-1} e_k}&\frac{\sqrt{k+1}\beta_1 \lambda B^{-1} e_k}{1 + \beta_1 \lambda \boldsymbol 1^T B^{-1} e_k}\\ \frac{-\boldsymbol 1^T B^{-1} }{1 + \beta_1 \lambda \boldsymbol 1^T B^{-1} e_k} & \frac{\sqrt{k+1}}{1 + \beta_1 \lambda \boldsymbol 1^T B^{-1} e_k}\end{bmatrix}, \quad \text{with } B^{-1} = \begin{bmatrix}1&\beta_k \lambda & \cdots & (\beta_2\cdots \beta_k)\lambda^{k-1}\\
            & 1& \ddots & \vdots \\&& \ddots& \beta_2\lambda \\ &&& 1\end{bmatrix}.
    \end{equation}
    It is easy to check that
    \[\begin{aligned}
        \boldsymbol 1^T B^{-1} &= [1, \frac{\sum_{s=k-1}^k\beta_{k-1}\cdots \beta_s \lambda^{s-k+2}}{\beta_{k-1}\lambda}, \cdots, \frac{\sum_{s=1}^k\beta_{1}\cdots \beta_s \lambda^{s}}{\beta_1 \lambda}]\\
        B^{-1}e_k &= [(\beta_2\cdots \beta_k)\lambda^{k-1}, \cdots, \beta_2\lambda, 1]^T,\\
        1 + \beta_1& \lambda \boldsymbol 1^T B^{-1} e_k = 1 + \sum_{s=1}^{k}\beta_1\cdots \beta_s \lambda^s = D_{kk}(-\lambda).
    \end{aligned}\]
    Therefore, 
    \[\widetilde \W^{-1} = \frac{1}{D_{kk}(-\lambda)}\begin{bmatrix}D_{kk}(-\lambda) B^{-1} - \beta_1 \lambda B^{-1}e_k \boldsymbol 1^T B^{-1}&\sqrt{k+1}\beta_1 \lambda B^{-1} e_k\\ -\boldsymbol 1^T B^{-1} & \sqrt{k+1}\end{bmatrix}.\]
    Next, we compute the $(i,j)$-elements of $\widetilde \W^{-1}$ for $1\le i,j\le k$.
    \[\widetilde \W^{-1}_{ij} = e_i^T B^{-1} e_j - \beta_1 \lambda \frac{ e_i^T B^{-1}e_k \boldsymbol 1^T B^{-1} e_j}{D_{kk}(-\lambda)}.\]
    \begin{itemize}
        \item If $i > j$, the first term $e_i^T B^{-1} e_j = 0$, so we have 
        \[\begin{aligned}
            \widetilde \W_{ij}^{-1} &= -\beta_1\lambda \frac{ e_i^T B^{-1}e_k \boldsymbol 1^T B^{-1} e_j}{D_{kk}(-\lambda)}\\
            &= -\frac1{D_{kk}(-\lambda)}(\beta_1\cdots \beta_{k-i+1})\lambda^{k-i+1}\cdot  \frac{\sum_{s=k-j+1}^k\beta_{k-j+1}\cdots \beta_s \lambda^{s+j-k}}{\beta_{k-j+1}\lambda}\\
            &= -\frac{\beta_1\cdots \beta_{k-i+1}}{\beta_1\cdots \beta_{k-j+1}}\frac{\sum_{s=k-j+1}^k \beta_1\cdots \beta_s \lambda^{s+j-i}}{D_{kk}(-\lambda)}.
        \end{aligned}\]
        Here, we use the fact that
        \[\frac{\beta_1\cdots \beta_{k-i+1}}{\beta_1\cdots \beta_{k-j+1}} \le \frac{\beta_1\cdots \beta_{k-i+1}\beta_{k-i+2}\cdots \beta_{s+j-i}}{\beta_1\cdots \beta_{k-j+1}\beta_{k-j+2}\cdots \beta_s} = \frac{\beta_1\cdots \beta_{s+j-i}}{\beta_1\cdots \beta_{s}},\]
        which follows from the monotonicity of the $\beta_i$'s. Therefore, we obtain the inequality
        \[\left|\widetilde \W_{ij}^{-1}\right|\le \frac{\sum_{s=k-j+1}^k \beta_1\cdots \beta_{s+j-i}\lambda^{s+j-i}}{D_{kk}(-\lambda)}\le 1.\]
        \item If $i\le j$, the first term $e_i^T B^{-1} e_j = \beta_{k-j+2}\cdots \beta_{k-i+1}\lambda^{j-i} $, and thus we have
        \[\begin{aligned}
            \widetilde \W^{-1}_{ij} &= \beta_{k-j+2}\cdots \beta_{k-i+1}\lambda^{j-i} - \frac{\beta_1\cdots \beta_{k-i+1}}{\beta_1\cdots \beta_{k-j+1}}\frac{\sum_{s=k-j+1}^k \beta_1\cdots \beta_s \lambda^{s+j-i}}{D_{kk}(-\lambda)}\\
            &= \frac{\beta_1\cdots \beta_{k-i+1}}{\beta_1\cdots \beta_{k-j+1}} \cdot \left(\lambda^{j-i} - \frac{\sum_{s=k-j+1}^k \beta_1\cdots \beta_s \lambda^{s+j-i}}{D_{kk}(-\lambda)} \right)\\
            &= \frac{\beta_1\cdots \beta_{k-i+1}}{\beta_1\cdots \beta_{k-j+1}} \cdot  \frac{\sum_{s=0}^k \beta_1\cdots \beta_s \lambda^{s+j-i} - \sum_{s=k-j+1}^k \beta_1\cdots \beta_s \lambda^{s+j-i}}{D_{kk}(-\lambda)}\\
            &= \frac{\beta_1\cdots \beta_{k-i+1}}{\beta_1\cdots \beta_{k-j+1}} \cdot  \frac{\sum_{s=0}^{k-j} \beta_1\cdots \beta_s \lambda^{s+j-i}}{D_{kk}(-\lambda)}.
        \end{aligned}\]
        Here, 
        \[\frac{\beta_1\cdots \beta_{k-i+1}}{\beta_1\cdots \beta_{k-j+1}} = \frac{\beta_1\cdots \beta_{s+j-i}\beta_{s+j-i+1}\cdots \beta_{k-i+1}}{\beta_1 \cdots \beta_s \beta_{s+1}\cdots \beta_{k-j+1}}\le\frac{\beta_1\cdots \beta_{s+j-i}}{\beta_1 \cdots \beta_s}, \]
        which follows from the monotonicity of the $\beta_i$'s. Therefore, we get 
        \[\left|\widetilde \W_{ij}^{-1}\right| \le \frac{\sum_{s=0}^{k-j} \beta_1\cdots \beta_{s+j-i} \lambda^{s+j-i}}{D_{kk}(-\lambda)}\le 1.\]
    \end{itemize}
    Moreover, it is easy to verify that $|\widetilde \W^{-1}_{ij}|\le 1$ for $i=k+1$ and $j \le k$. Next, we consider an upper bound for the summation 
    \begin{equation}\label{eq:0}
        \sum_{j=1}^{k} \sum_{i=1}^{k+1} \left|\widetilde \W_{ij}^{-1}\right|.
    \end{equation}
    From the previous discussion, we know that $\widetilde \W^{-1}_{ij} < 0$ if and only if $i > j$, and we also have $\boldsymbol 1 ^T \widetilde \W^{-1} e_j = e_{k+1}^T \widetilde \W \widetilde \W^{-1} e_j = 0$ for all $j = 1,\dots, k$. This implies that
    \begin{equation}\label{eq:1}
        \begin{aligned}
            \sum_{i=1}^{k+1}\left|\widetilde \W^{-1}_{ij}\right| = 2\sum_{i=1}^j \widetilde \W^{-1}_{ij}&=\frac{2}{D_{kk}(-\lambda)}\sum_{i=1}^j \sum_{s=0}^{k-j} \frac{\beta_1\cdots \beta_{k-i+1}\cdot \beta_1\cdots \beta_s }{\beta_1\cdots \beta_{k-j+1}}\lambda^{s+j-i}\\
            &= \frac{2}{D_{kk}(-\lambda)}\sum_{l=0}^{j-1} \sum_{s=0}^{k-j} \frac{\beta_1\cdots \beta_{k-j + l+1}\cdot \beta_1\cdots \beta_s }{\beta_1\cdots \beta_{k-j+1}}\lambda^{s+l}\\
            &= \frac{2}{D_{kk}(-\lambda)}\sum_{l=0}^{j-1} \sum_{s=0}^{k-j} \frac{\beta_{s+l+1}\cdots \beta_{k-j+l+1} }{\beta_{s+1}\cdots \beta_{k-j+1}}\beta_{1}\cdots \beta_{s+l} \lambda^{s+l}\\
            &= \frac{2}{D_{kk}(-\lambda)}\sum_{r=0}^{k-1}\sum_{l=\max\{0, r+j-k\}}^{\min\{j-1, r\}} \frac{\beta_{r+1}\cdots \beta_{k-j+l+1} }{\beta_{r-l+1}\cdots \beta_{k-j+1}} \beta_1\cdots \beta_r \lambda^r,
        \end{aligned}
    \end{equation}
    where we make the variable substitution $l = j-i$ in the third equality, and $r = s+l$ in the last equality. Thus, we have 
    \begin{equation}\label{eq:3}
        \sum_{j=1}^{k} \sum_{i=1}^{k+1} \left|\widetilde \W_{ij}^{-1}\right|= \frac{2}{D_{kk}(-\lambda)} \sum_{r=0}^{k-1}\left(\sum_{j=1}^k \sum_{l=\max\{0, r+j-k\}}^{\min\{j-1, r\}} \frac{\beta_{r+1}\cdots \beta_{k-j+l+1} }{\beta_{r-l+1}\cdots \beta_{k-j+1}}\right) \beta_1\cdots \beta_r \lambda^r
    \end{equation}
    Next, we give an upper bound for the inner summation, which can be written as
    \begin{equation}\label{eq:sum}
        \begin{aligned}
            \sum_{j=1}^{k-r}\sum_{l=0}^{\min\{j-1, r\}}\frac{\beta_{r+1}\cdots \beta_{k-j+l+1} }{\beta_{r-l+1}\cdots \beta_{k-j+1}} + \sum_{j=k-r+1}^k \sum_{l=r+j-k}^{\min\{j-1, r\}}\frac{\beta_{r+1}\cdots \beta_{k-j+l+1} }{\beta_{r-l+1}\cdots \beta_{k-j+1}}.
        \end{aligned}
    \end{equation}
    For the first term in \eqref{eq:sum}, we have $r\le k-j$, and thus 
    \[\begin{aligned}
        \frac{\beta_{r+1}\cdots \beta_{k-j+l+1} }{\beta_{r-l+1}\cdots \beta_{k-j+1}} &= \frac{\beta_{r+1}\cdots \beta_{k-j+1} \cdots \beta_{k-j+l+1}}{\beta_{r-l+1}\cdots \beta_{r+1}\cdots \beta_{k-j+1}}\le \frac{\beta_{k-j+2}\cdots \beta_{k-j+l+1}}{\beta_{r-l+1}\cdots \beta_{r}}\\
        &\le \left(\frac{\beta_{k-j+2}}{\beta_r}\right)^l \le \left(\frac{r}{k-j+2}\right)^l,\quad \forall l\ge 1,
    \end{aligned}\]
    where we use the assumption that $\beta_j / \beta_i \le i/j\le 1$ for all $1\le i\le j\le k$ in the last two inequalities. Note that the inequality also holds for $l = 0$. Therefore, we obtain 
    \[\sum_{l=0}^{\min\{j-1, r\}}\frac{\beta_{r+1}\cdots \beta_{k-j+l+1} }{\beta_{r-l+1}\cdots \beta_{k-j+1}} \le \sum_{l=0}^{\min\{j-1, r\}} \left(\frac{r}{k-j+2}\right)^l \le \sum_{l=0}^{\infty} \left(\frac{r}{k-j+2}\right)^l = \frac{k-j+2}{k-j-r+2}. \]
    For the second term in \eqref{eq:sum}, we have $r > k-j$. Using the assumption that $\beta_j / \beta_i \le i/j\le 1$ for all $1\le i\le j\le k$ again, we get
    \[\frac{\beta_{r+1}\cdots \beta_{k-j+l+1} }{\beta_{r-l+1}\cdots \beta_{k-j+1}}\le \left(\frac{\beta_{r+1}}{\beta_{k-j+1}}\right)^{k-j+l-r+1} \le \left(\frac{k-j+1}{r+1}\right)^{k-j+l-r+1}.\]
    Thus, we have
    \[\begin{aligned}
        \sum_{l=r+j-k}^{\min\{j-1, r\}}\frac{\beta_{r+1}\cdots \beta_{k-j+l+1} }{\beta_{r-l+1}\cdots \beta_{k-j+1}} &\le \sum_{l=r+j-k}^{\min\{j-1, r\}} \left(\frac{k-j+1}{r+1}\right)^{k-j+l-r+1} = \sum_{l=0}^{\min\{j-1,r\}+k-j-r}\left(\frac{k-j+1}{r+1}\right)^{l+1}\\
        &\le \sum_{l=0}^{\infty} \left(\frac{k-j+1}{r+1}\right)^{l+1}= \frac{k-j+1}{r+j-k}.
    \end{aligned}\]
    Then, the summation in \eqref{eq:sum} is upper bounded by 
    \begin{equation}
    \begin{aligned}
        \sum_{j=1}^{k-r} \frac{k-j+2}{k-j-r+2} + \sum_{j=k-r+1}^{k} \frac{k-j +1}{r+j-k}\le& (k+1)\sum_{j=1}^{k-r} \frac{1}{k-j-r+2} + r\sum_{j=k-r+1}^{k} \frac{1}{r+j-k}\\
        =& (k+1)\left(\frac12 + \cdots + \frac{1}{k-r+1}\right) + r\left(1 + \cdots + \frac{1}{r}\right)\\
        \le& (k+1)\log(k-r+1) + r(1 + \log r)\\
        \le& (k+1) \left(\log\left(r(k-r+1)\right)+1\right)\\
        \le& (k+1)\left(\log\left(\frac{k+1}{2}\right)^2 + 1\right)\\
        \le& 2(k+1) \log(k+1).
    \end{aligned}
    \end{equation}
    Substituting this bound into \eqref{eq:3}, we obtain
    \[\sum_{j=1}^{k} \sum_{i=1}^{k+1} \left|\widetilde \W_{ij}^{-1}\right|\le  4(k+1)\log(k+1)\frac{\sum_{r=0}^{k-1}d_r\lambda^r}{D_{kk}(-\lambda)}\le 4(k+1)\log(k+1).\]
    Finally, using the fact that 
    \[\sum_{i=1}^{k+1}|\widetilde \W_{i,k+1}^{-1}| = \sum_{i=1}^{k+1}\widetilde \W_{i,k+1}^{-1} = \boldsymbol 1^T \widetilde \W^{-1}e_{k+1} = \sqrt{k+1}e_{k+1}^T \widetilde \W \widetilde \W^{-1}e_{k+1} = \sqrt{k+1},\]
    we get
    \[\begin{aligned}
        \|\widetilde \W^{-1}\|_2 & \le \|\widetilde \W^{-1}\|_F = \sqrt{\sum_{j=1}^{k+1} \sum_{i=1}^{k+1} \left|\widetilde \W_{ij}^{-1}\right|^2} = \sqrt{\sum_{j=1}^{k} \sum_{i=1}^{k+1} \left|\widetilde \W_{ij}^{-1}\right|^2 + \sum_{i=1}^{k+1} \left|\widetilde \W_{i,k+1}^{-1}\right|^2}\\ 
        &\le \sqrt{\sum_{j=1}^{k} \sum_{i=1}^{k+1} \left|\widetilde \W_{ij}^{-1}\right| +\left(\sum_{i=1}^{k+1} \left|\widetilde \W_{i,k+1}^{-1}\right|\right)^2} \\
        & \le \sqrt{4(k+1)\log(k+1) + k + 1} = \sqrt{(k+1)(4\log(k+1) + 1)}.
    \end{aligned}\]
    Here, we use the fact that $\left|\widetilde \W^{-1}_{ij}\right| \le 1$ for all $1\le i\le k+1$ and $1\le j \le k$ in the second inequality.
\end{proof}
\begin{proof}[Proof of \Cref{lem:W_inv_norm}]
    Since $\A$ is Hermitian and negative semi-definite, by unitary similarity transformation, $\W_k(\A)$ is unitarily similar to a block diagonal matrix. Each diagonal block has the same structure as $\W_k(\A)$ but with scalar blocks. Therefore, without loss of generality, we assume $\A = -\lambda$, where $\lambda \ge 0$. Next, applying \Cref{lem:1d_W_inv_norm}, we obtain the bound 
    \[\left\|\W_k(\A)^{-1}\right\|_2 \le \sqrt{(k+1)(4\log(k+1) + 1)}.\]
\end{proof}
\begin{lem}\label{lem:W_inv}
    For any matrix $\A \in \mathbb C^{n\times n}$, if the matrix $\W_k(\A)$ defined by \cref{eq:WT} is invertible, then the inverse is given by 
    \[ \begin{aligned}
        &\W_k(\A)^{-1} = \left[\I_{k+1}\otimes \left(D_{kk}(\A)\right)^{-1}\right] \times \\        
        &{\small \begin{bmatrix}
            d_k\sqrt{k+1} (-\A)^k& \frac{d_k}{d_k}\sum_{j=0}^{k-1} d_j (-\A)^j& \cdots& \frac{d_k}{d_2} \sum_{j=0}^1 d_j (-\A)^{j+k-2}& \frac{d_k}{d_1}d_0 (-\A)^{k-1}  \\
            
            d_{k-1}\sqrt{k+1}(-\A)^{k-1}& -\frac{d_{k-1}}{d_k}\sum_{j=k}^k d_j (-\A)^{j-1} & \ddots & \vdots & \vdots  \\
            
            \vdots & \vdots& \ddots& \frac{d_2}{d_2} \sum_{j=0}^1 d_j (-\A)^j& \frac{d_2}{d_1} d_0 (-\A) \\
            
            d_1\sqrt{k+1}(-\A)&-\frac{d_1}{d_k}\sum_{j=k}^k d_j (-\A)^{j-k+1}& \cdots& -\frac{d_1}{d_2}\sum_{j=2}^k d_j (-\A)^{j-1}& \frac{d_1}{d_1}d_0\I_n\\
            
            d_0\sqrt{k+1}\I_n & -\frac{d_0}{d_k} \sum_{j=k}^k d_j (-\A)^{j-k} & \cdots &-\frac{d_0}{d_2}\sum_{j=2}^k d_j (-\A)^{j-2} & -\frac{d_0}{d_1}\sum_{j=1}^k d_j (-\A)^{j-1}
        \end{bmatrix}}.
    \end{aligned}\]
\end{lem}
\begin{proof}
    Applying the discussion in \Cref{lem:1d_W_inv_norm}, we can directly write each block of $\W_k(\A)^{-1}$.
\end{proof}
\begin{proof}[Proof of \Cref{lem:le1_W_inv_norm}]
    First, using \Cref{lem:D_inv_norm}, we have 
    \[\left\|D_{kk}(\A)^{-1}\right\|_2 \le \frac{2}{3 - e}.\]
    Next, we construct a $(k+1)\times (k+1)$ matrix $\widehat W$ satisfying
    \[\widehat \W_{ls} = \begin{cases}
        \frac{d_{l-1}}{d_t}\sum_{j=0}^{t-1}d_j, & s \ge 2, t+1\le l \le k+1,\\
        \frac{d_{l-1}}{d_t}\sum_{j=t}^k d_j, & s\ge 2, 1\le l \le t,\\
        \sqrt{k+1} d_{k+1-l}, & s=1,
    \end{cases}\]
    where we define $t = k+2-s$ for simplicity. Then, by \Cref{lem:W_inv} and \Cref{lem:block_norm}, we have 
    \[\left\|\W_k(\A)^{-1}\right\|_2 \le \left\|D_{kk}(\A)^{-1}\right\|_2\cdot \left\|\widehat \W\right\|_2\le \frac{2}{3-e}\left\|\widehat \W\right\|_2.\]
    Now, we focus on bounding the spectral norm of matrix $\widehat \W$, which is indeed a special case of the discussion in \Cref{lem:1d_W_inv_norm}. Setting $\lambda = 1$ in \Cref{lem:1d_W_inv_norm}, and get 
    \[\|\widehat \W\|_2 \le D_{kk}(-1) \cdot \sqrt{(k+1)(4\log(k+1) + 1)},\]
    where $D_{kk}(-1) \le \sqrt e$ by \Cref{lem:Dkk-1}. Thus, we obtain 
    \[\left\|\W_k(\A)^{-1}\right\|_2 \le \frac{2\sqrt e}{3-e}\sqrt{(k+1)(4\log(k+1) + 1)}.\]
\end{proof}
\subsection{Proof of \Cref{thm:cond}}\label{sec:cond}
\begin{lem}\label{obs:2.3}
    Suppose that $\alpha$ and $\beta$ are defined as in \cref{eq:alpha_beta}. For indexes $i > j$, it holds that $\alpha_i < \alpha_j$ and $\beta_i < \beta_j$.
\end{lem}
\begin{proof}
    We examine the quantity $\alpha_{j} / \alpha_{j+1}$
    \[ \begin{aligned}
        \frac{\alpha_j}{\alpha_{j+1}} = \frac{(j+1)(p+q-j)(p-j+1)}{j(p-j)(p+q-j+1)} = \left(1 + \frac1j\right)\left(1 + \frac{q}{(p-j)(p+q-j+1)}\right) > 1.
    \end{aligned} \]
    By induction, we conclude that $\alpha_i < \alpha_j$ for $i > j$. Similarly, we can show that $\beta_i < \beta_j$ for $i > j$.
\end{proof}
\begin{lem}\label{lem:1W_inv_norm}
    Suppose $\A \in \mathbb C^{n\times n}$ is Hermitian and negative semi-definite, and $\W_k(\A)$ is defined by \cref{eq:WT}. Then, we have 
    \begin{equation}\label{eq:1W_inv_norm}
        \left\|\widetilde \E^T \W_k(\A)^{-1}\right\|_2 \le \sqrt{5k + 1},
    \end{equation}
    where $\widetilde \E = \widetilde {\boldsymbol 1} \otimes \I_n$, and $\widetilde {\boldsymbol 1}$ is defined in \cref{eq:zy1}.
\end{lem}
\begin{proof}
    For simplicity, we write $\W = \W_k(\A)$ in this proof. Similar to the proof of \Cref{lem:W_inv_norm}, we only need to consider the case $\A = -\lambda \le 0$. At first, we consider solving the linear system
    \begin{equation}\label{eq:A.1}
        \W^T \x = \widetilde {\boldsymbol 1},
    \end{equation}
    where $\x = (x_{k+1}, \dots, x_1)^T$. We aim to show that the solution to this linear system \cref{eq:A.1} is given by 
    \begin{equation}\label{eq:A.3}
        \begin{cases}
            x_{k+1} = -\sqrt{k+1} \frac{D(\lambda)}{D(-\lambda)},\\
            x_j = \frac{(-1)^{j+1}}{D(-\lambda)} \sum_{i=0}^{k-j} \frac{d_{j+i}}{d_j} (-\lambda)^i P_{j+i+1}(\lambda), 1\le j \le k.
        \end{cases}
    \end{equation}
    Here, the notations are defined as follows
    \[ \begin{cases}
        d_0 = 1,\\
        d_j = d_0\beta_1\cdots \beta_j, & 1\le j \le k,
    \end{cases} \quad D(\lambda) = \sum_{j=0}^k d_j (-\lambda)^j, \quad P_j(\lambda) = D(-\lambda) +(-1)^{j}D(\lambda). \]
    Using the explicit expression for $\W^{-1}$ from \Cref{lem:W_inv}, we can verify that $x_{k+1} = -\sqrt{k+1}\frac{D(\lambda)}{D(-\lambda)}$. Moreover, by the first row of the linear system \eqref{eq:A.1}, we have 
    \[ x_{k} = (-1)^{k+1} - \frac1{\sqrt{k+1}}x_{k+1} = \frac{(-1)^{k+1}}{D(-\lambda)}P_{k+1}(\lambda) \]
    which satisfies the expression \eqref{eq:A.3}. 
    
    Next, we verify for general $j$ by induction. Using the recurrence relation 
    \begin{equation}\label{eq:recu}
         \frac1{\sqrt{k+1}} x_{k+1} - \beta_{j+1}\lambda x_{j+1} + x_{j} = (-1)^{j+1},\quad \forall j = 0, \dots, k-1,
    \end{equation}
    we obtain 
    \[ \begin{aligned}
        x_{j} &= (-1)^{j+1} - \frac1{\sqrt{k+1}}x_{k+1} + \beta_{j+1}\lambda x_{j+1} \\
        &= \frac{(-1)^{j+1}}{D(-\lambda)}\left(P_{j+1}(\lambda) + (-1)^{j+1} D(-\lambda) \beta_{j+1} \lambda x_{j+1}\right)\\
        &= \frac{(-1)^{j+1}}{D(-\lambda)}\left(P_{j+1}(\lambda) + (-\beta_{j+1} \lambda) \sum_{i=0}^{k-j-1} \frac{d_{j+i+1}}{d_{j+1}} (-\lambda)^i P_{j+i+2}(\lambda)\right)\\
        &= \frac{(-1)^{j+1}}{D(-\lambda)} \sum_{i=0}^{k-j} \frac{d_{j+i}}{d_{j}} (-\lambda)^i P_{j+i+1}(\lambda).
    \end{aligned} \]
    This shows that $x_{j}$ also satisfies the expression in \eqref{eq:A.3}. 
    
    The next step is to show that $(-1)^{j+1} x_{j}\ge 0$. For the summation in the expression of $x_j$, we have
    \begin{equation}\label{eq:A.5}
        \begin{aligned}
            & \sum_{i=0}^{k-j} \frac{d_{j+i}}{d_j} (-\lambda)^i \left(\sum_{l=0}^k d_l \lambda^l \left(1 + (-1)^{j+i+1+l}\right)\right) \\=&\frac1{d_j}\sum_{i=0}^{k-j} \sum_{l=0}^k \left((-1)^i + (-1)^{j+l+1}\right) d_{j+i}d_l \lambda^{l+i} \\
            =& \frac1{d_j}\sum_{r=0}^{2k-j}\sum_{i=\max\{0, r-k\}}^{\min{\{k-j, r\}}} \left((-1)^i + (-1)^{j+r-i+1}\right) d_{j+i}d_{r-i} \lambda^{r}\\
            =& \frac1{d_j}\sum_{r=0}^{2k-j}\left(1+(-1)^{j+r+1}\right)\left(\sum_{i=\max\{0, r-k\}}^{\min{\{k-j, r\}}}(-1)^i  d_{j+i}d_{r-i}\right) \lambda^{r},
        \end{aligned}
    \end{equation} 
    where we performed the variable substitution $r = l+i$ in the second equality. In the last line of \eqref{eq:A.5}, all terms are obviously non-negative except for the inner summation term:
    \begin{equation}\label{eq:inner_sum}
        \sum_{i=\underline i}^{\overline i}(-1)^i  d_{j+i}d_{r-i},
    \end{equation}
    where $\underline i = \max\{0, r-k\}$ and $\overline i = \min\{k-j, r\}$. 
    
    We now show that this summation is also non-negative under the condition that $j+r+1$ is an even number. This condition is reasonable because the coefficient $1 + (-1)^{j+r+1}$ vanishes when $j+r+1$ is odd. First, we consider the summation 
    \begin{equation}\label{eq:gene_sum}
        \sum_{i=0}^{s} (-1)^i d_{j+i}d_{r-i}.
    \end{equation}
    \begin{itemize}
        \item If $j>r$, we have 
        \[\begin{aligned}
            \sum_{i=0}^{s} (-1)^{i} d_{j+i}d_{r-i} &= \sum_{i=0}^{\lceil s/2\rceil - 1}\left(d_{j+2i}d_{r-2i} - d_{j+2i+1}d_{r-2i-1}\right) + d_{j+s}d_{r-s} \delta_{\{s\text{ is even}\}}\\
            &\ge \sum_{i=0}^{\lceil s/2\rceil - 1}d_{j+2i}d_{r-2i} \left( 1- \frac{\beta_{j+2i+1}}{\beta_{r-2i}}\right)\ge 0,
        \end{aligned}\]
        where we use the decreasing property of the sequence $\beta_i$.
        \item If $j < r$ but $j + s \ge r$, we can find $i^*\in[0, s]$ such that $j+i^* = r$. The integer $i^*$ is odd, and we have 
        \[\begin{aligned}
            \sum_{i=0}^{i^*} (-1)^i d_{j+i} d_{r-i} &= \sum_{l=0}^{i^*} (-1)^{i^* - l} d_{j+i^* -l} d_{r-i^* + l} = -\sum_{l=0}^{i^*} (-1)^{l} d_{r -l} d_{j + l} = 0,
        \end{aligned}\]
        where we reverse the summation order and make the variable substitution $l = i^* - i$ in the first equality. Taking the common factor $(-1)^{i^*} = -1$ outside the summation, we see that the sum equals zero. Therefore, we obtain
        \[\sum_{i=0}^{s} (-1)^{i} d_{j+i}d_{r-i} = \sum_{i=i^* + 1}^s (-1)^{i} d_{j+i}d_{r-i} = \sum_{i=0}^{s-i^* -1} (-1)^{i+i^*+1}d_{j+i^*+1+i} d_{r-i^* -1 - i}.\]
        Let $s' = s-i^* -1$, $j' = j+i^*+1$, and $r' = r - i^* -1$. Then, we have $r'+j'+1$ as an even number, with $j' > r'$, and thus the above summation is non-negative as in the first case.
    \end{itemize}
    In summary, the summation \eqref{eq:gene_sum} is non-negative if $j+s \ge r$. Specifically, it equals zero when $j+s = r$. We now rewrite the summation \eqref{eq:inner_sum} as
    \begin{equation}\label{eq:mod_inner_sum}
        \sum_{i=\underline i}^{\overline i} (-1)^i d_{j+i}d_{r-i} = (-1)^{\underline i}\sum_{i=0}^{\overline i - \underline i} (-1)^{i} d_{j+\underline i+i}d_{r-\underline i -i}.
    \end{equation}
    Defining $s' = \overline i - \underline i$, $j' = j+\underline i$, and $r' = r-\underline i$, we observe $r'+j'+1$ is an even number. Moreover, we compute  
    \[j'+s'-r' = j + \min\{k-j, r\} + \max\{0, r-k\} - r = \min\{k, r+j\} - \min\{k, r\} \ge 0.\]
    Thus, the sign of \eqref{eq:inner_sum} is determined by 
    \[(-1)^{\underline i} = (-1)^{\max\{0, r-k\}}.\]
    Notably, if $r\ge k$, we have $j'+s'=r'$, leading to the summation equaling zero. Consequently, we conclude that the summation \eqref{eq:inner_sum} is always non-negative. 
    
    Hence, we have established that $(-1)^{j+1} x_j \ge 0$. Moreover, from \eqref{eq:recu}, we obtain 
    \[ 0\le (-1)^{j+1}x_j = 1 - (-1)^{j+1} \frac{x_{k+1}}{\sqrt{k+1}} + (-1)^{j+1} \beta_{j+1}\lambda x_{j+1}. \]
    Using the fact that $|x_{k+1}|< \sqrt{k+1}$, we conclude that
    \[1 - (-1)^{j+1} \frac{x_{k+1}}{\sqrt{k+1}} \ge 0, \quad \text{while }\quad (-1)^{j+1} \beta_{j+1}\lambda x_{j+1} \le 0.\]
    Thus, we obtain
    \[ \begin{aligned}
        |x_j| &= (-1)^{j+1} x_j \le 1 - (-1)^{j+1} \frac{x_{k+1}}{\sqrt{k+1}} \le 2.
    \end{aligned} \]
    This implies the final bound
    \[ \|\widetilde {\boldsymbol 1}^T \W^{-1}\|_2 \le \sqrt{5k+1}. \]
\end{proof}
\begin{lem}\label{lem:le11W_inv_norm}
    Suppose $\A\in \mathbb C^{n\times n}$ satisfies $\|\A\|_2 \le 1$, and let $\W_k(\A)$ be defined as in \cref{eq:WT}. Then, we have
    \[\left\|\widetilde \E ^T \W_k(\A)^{-1}\right\|_2 \le \left(\sqrt e + \frac{2 e}{3 - e}\right)\sqrt{2k+1}.\]
    Here, $\widetilde \E = \widetilde {\boldsymbol 1} \otimes \I_n$, where $\widetilde {\boldsymbol 1}$ is define in \cref{eq:zy1}.
\end{lem}
\begin{proof}
    As in \Cref{lem:1W_inv_norm}, we consider solving the linear system
    \[\begin{bmatrix}
        \frac{\I_n}{\sqrt{k+1}}&\I_n\\
        \frac{\I_n}{\sqrt{k+1}}& \beta_k \A&\I_n\\
        \vdots && \ddots & \ddots \\
        \frac{\I_n}{\sqrt{k+1}}&&&\beta_2\A&\I_n\\
        \frac{\I_n}{\sqrt{k+1}}&&&&\beta_1\A 
    \end{bmatrix}\begin{bmatrix}\X_{k+1}    \\ \X_k\\ \vdots \\ \X_2\\ \X_1 \end{bmatrix}= \begin{bmatrix}(-\I_n)^{k+1}\\(-\I_n)^k\\\vdots \\ (-\I_n)^2\\ -\I_n\end{bmatrix}.\]
    Using the notation in \Cref{lem:1W_inv_norm}, we obtain 
    \[\begin{cases}
        \X_{k+1} = -\sqrt{k+1}D_{kk}(\A)^{-1} D_{kk}(-\A),\\
        \X_j = (-1)^{j+1}\sum_{i=0}^{k-j} \frac{d_{j+i}}{d_j} \A^i \left(\I_n + (-1)^{j+i+1}D_{kk}(\A)^{-1} D_{kk}(-\A)\right),\quad  1\le j \le k.
    \end{cases}\]
    Given the assumption $\|\A\|_2\le 1$, along with \Cref{lem:Dkk-1} and \Cref{lem:D_inv_norm}, we obtain 
    \[\left\|D_{kk}(\A)^{-1} D_{kk}(-\A)\right\|_2 \le \frac{2}{3-e}D_{kk}(-1) \le \frac{2\sqrt e}{3 - e}.\]
    Thus, we have $\|\X_{k+1}\|_2 \le \frac{2\sqrt e}{3 - e} \sqrt{k+1}$, and 
    \[\begin{aligned}
        \left\|\X_j\right\|_2 &\le \left(1 + \frac{2\sqrt e}{3 - e}\right)\sum_{i=0}^{k-j}\frac{d_{j+i}}{d_j} \le \left(1 + \frac{2\sqrt e}{3 - e}\right)\sum_{i=0}^{k}d_{i}\\
        &\le \sqrt e + \frac{2 e}{3 - e}.
    \end{aligned}\]
    Finally, applying \Cref{lem:block_norm}, we obtain
    \[\left\|\widetilde \E ^T \W_k(\A)^{-1}\right\|_2 \le \left(\sqrt e + \frac{2 e}{3 - e}\right)\sqrt{2k+1}.\]
\end{proof}
\begin{proof}[Proof of \Cref{thm:cond}]
    For simplicity, we use $\LL$ and $\W$ to denote the matrices $\LL_{m,k,p}(\A h)$ and $\W_{k}(\A h)$, respectively. By decomposing $\LL$ into the sum of its block diagonal and block sub-diagonal parts, we obtain 
    \[ \|\LL\|_2 \le \max\left\{\|\W\|_2, 1\right\} + \| \widetilde \T\|_2, \]
    where $\| \widetilde \T\|_2 = 1$. Furthermore, we decompose $\W$ as follows 
    \[ \W = \begin{bmatrix} \frac1{\sqrt{k+1}}\I_n & \frac1{\sqrt{k+1}}\I_n &\cdots &\frac1{\sqrt{k+1}}\I_n\\ \0 & \0\\ &\ddots&\ddots \\ && \0&\0 \end{bmatrix} + \begin{bmatrix} \0 & \0 &\cdots &\0 \\ \I_n & \0\\ &\ddots&\ddots \\ && \I_n&\0 \end{bmatrix} + \begin{bmatrix} \0 & \0 &\cdots &\0 \\ \0 & \beta_k\A h\\ &\ddots&\ddots \\ && \0&\beta_1\A h  \end{bmatrix}.\]
    From this decomposition, a reasonable bound on $\|\W\|_2$ is 
    \[ \|\W\|_2 \le \beta_1 h \|\A\|_2 + 2. \] 
    Thus, we obtain 
    \[ \|\LL\|_2\le \beta_1 h \|\A\|_2 + 3. \]
    To bound the spectral norm of $\LL^{-1}$, we decompose it as
    \[ \begin{aligned}
        \LL &= \footnotesize \begin{bmatrix}
            \W\\ & \ddots \\ && \W \\&&& \frac{1}{\sqrt{k+1}}\I_n\\
            &&&& \I_n\\ &&&&& \ddots \\ &&&&&& \I_n
        \end{bmatrix} \begin{bmatrix}\I\\ \W^{-1} \widetilde \T & \I \\&\ddots & \ddots \\ &&\W^{-1} \widetilde \T & \I  \\&&& \widetilde {\boldsymbol 1}^T\otimes \I_n & \I_n    \\
            &&&& -\I_n & \I_n \\ &&&&& \ddots & \ddots \\ &&&&&& -\I_n & \I_n\end{bmatrix}\\
        &=: \D\left(\I - \N\right),
    \end{aligned} \]
    where $\N$ is a nilpotent matrix. Moreover, defining $\widetilde \V = - \frac{1}{\sqrt{k+1}}\W^{-1}\left(\e_1 \otimes \I_n\right)$, and $\widetilde \E = \widetilde {\boldsymbol 1} \otimes \I_n$, we obtain
    \[ -\W^{-1} \widetilde \T =  \left(-\frac{1}{\sqrt{k+1}}\W^{-1} \left(\e_1 \otimes \I_n \right) \right)\left(\widetilde {\boldsymbol 1}^T \otimes \I_n\right) = \widetilde \V \widetilde \E^T.\]
    With these notations, we can express the inverse of $\LL$ as
    \[ \LL^{-1} = \sum_{j=0}^{m+p} \N^{j} \D^{-1}, \]
    where $\N^j$ for all $j \ge 1$ is a $j$-th block sub-diagonal matrix. The blocks in this sub-diagonal are given by
    \[\begin{cases}
        \left(\widetilde \V \widetilde \E^T\right)^{j}, \cdots, \left(\widetilde \V \widetilde \E^T\right)^{j}, -\widetilde \E^T  \left(\widetilde \V \widetilde \E^T\right)^{j-1}, \cdots , -\widetilde \E^T, \I_n,\cdots, \I_n, & 1\le j \le m\\
        -\widetilde \E^T  \left(\widetilde \V \widetilde \E^T\right)^{m}, \cdots , -\widetilde \E^T, \I_n,\cdots, \I_n, & j > m.
    \end{cases}\]
    The spectral norm of $\LL^{-1}$ can be bounded as follows:
    \[ \left\|\LL^{-1}\right\|_2 \le \sum_{j=0}^{m+p} \left\|\N^j \D^{-1}\right\|_2, \]
    where 
    \[ \left\|\N^j \D^{-1}\right\|_2 \le \max \left\{\left\|\left(\widetilde \V \widetilde \E^T\right)^{l}\W^{-1}\right\|_2, \left\|\widetilde \E^T \left(\widetilde \V \widetilde \E^T\right)^{l}\W^{-1}\right\|_2, \sqrt{k+1}\bigg |0\le l \le m \right\}, \]
    for all $0\le j\le m+p$. Since $\widetilde \V$ is proportional to the first block column of $\W^{-1}$, we obtain 
    \[ \widetilde \E^T \widetilde \V = \left[D_{kk}(\A h)\right]^{-1}\left(\sum_{j=0}^k d_j (\A h)^j\right) = \left[D_{kk}(\A h)\right]^{-1} N_{kk}(\A h)=R_{kk}(\A h).\]
    Substituting this into $\N^j$, we derive
    \[ \left(\widetilde \V \widetilde \E^T\right)^{j} = \widetilde \V \left(\widetilde \E^T \widetilde \V\right)^{j-1} \widetilde \E^T = \left(\I_{k+1}\otimes \left[R_{kk}(\A h )\right]^{j-1}\right) \widetilde \V \widetilde \E^T, \]
    where, in the second equality, we interchange the multiplication order of $R_{kk}(\A )$ and $\widetilde \V$. For $l\ge 1$, we obtain the following bounds
    \[ \begin{aligned}
        \left\|\left(\widetilde \V \widetilde     \E^T\right)^{l}\W^{-1}\right\|_2 &\le \left\|\left[R_{kk}(\A h)\right]^{l-1}\right\|_2 \left\|\widetilde \V \widetilde \E^T \W^{-1}\right\|_2, \\
        \left\|\widetilde \E^T \left(\widetilde \V     \widetilde \E^T\right)^{l-1}\W^{-1}\right\|_2 & \le \left\|\left[R_{kk}(\A h)\right]^{l-1}\right\|_2 \left\|\widetilde \E^T \W^{-1}\right\|_2.
    \end{aligned} \]
    Since $m$ and $k$ satisfy the condition $\left\|\I - R_{kk}^l(\A h)\exp(-l\A h)\right\|_2 \le 1$ for all $1\le l \le m$, there is a uniform bound 
    \[ \left\|\left[R_{kk}(\A h)\right]^{l}\right\|_2 \le 2\cdot C(\A),\quad \forall l = 1,\dots, m,\]
    where $C(\A) = \sup_{t\in[0, T]}\left\|\exp(\A t)\right\|_2$. Thus, we derive the bound for $\|\LL^{-1}\|_2$: 
    \begin{equation}
        \|\LL^{-1}\|_2 \le 2\cdot C(\A) (m+p) \max\left\{\|\W^{-1}\|_2, \left\|\widetilde \V \widetilde \E^T \W^{-1}\right\|_2, \left\|\widetilde \E^T \W^{-1}\right\|_2, \sqrt{k+1}\right\}.
    \end{equation}
    Up to this point, no special properties of $\A$ have been assumed. 
    \begin{itemize}
        \item If we assume $\A$ is Hermitian and negative semi-definite, we can apply \Cref{lem:W_inv_norm} and \Cref{lem:1W_inv_norm} to obtain a quite tight bound on the terms within the $\max\{\cdot\}$ operator. We have  
        \[ \begin{aligned}
            &\left\|\W^{-1}\right\|_2 \le \sqrt{(k+1)(4\log(k+1) + 1)},\\
            &\left\|\widetilde \E^T \W^{-1}\right\|_2 \le \sqrt{5k + 1}.
        \end{aligned} \] 
        Since $\|\widetilde \V\|_2 \le 1$, we can further bound
        \[\left\|\widetilde \V \widetilde \E^T \W^{-1}\right\|_2 \le \left\|\widetilde \V \right\|_2 \left\|\widetilde \E^T \W^{-1}\right\|_2 \le \left\|\widetilde \E^T \W^{-1}\right\|_2.\]
        For $k\ge 3$, we obtain
        \[\max\left\{\|\W^{-1}\|_2, \left\|\widetilde \V \widetilde \E^T \W^{-1}\right\|_2, \left\|\widetilde \E^T \W^{-1}\right\|_2, \sqrt{k+1}\right\}\le 3\sqrt{k\log k}.\]
        Moreover, in the case where $\A$ is Hermitian and negative semi-definite, we have $C(\A) = 1$, leading to  
        \[ \left\|\LL^{-1}\right\|_2 \le  6(m+p)\sqrt {k\log  k}.\]
        Finally, the condition number satisfies
        \[\kappa \le 3(m+p) \sqrt{k\log  k}\left(6 + \|\A h\|_2 \right).\]
        \item If $\A$ is an arbitrary matrix, but we choose $m = \left\lceil \|\A T\|_2\right\rceil$ such that $\|\A h\|_2 \le 1$, we can apply \Cref{lem:le1_W_inv_norm} and \Cref{lem:le11W_inv_norm} to obtain the following bounds 
        \[\begin{aligned}
            \left\|\W^{-1}\right\|_2&\le \frac{2\sqrt e}{3-e}\sqrt{(k+1)(4\log(k+1) + 1)}\\
            \left\|\widetilde \E ^T \W^{-1}\right\|_2 &\le \left(\sqrt e + \frac{2 e}{3 - e}\right)\sqrt{2k+1}.
        \end{aligned}\]
        Since $\|\widetilde \V\|_2$ can also be bounded by a constant, we conclude 
        \[ \left\|\LL^{-1}\right\|_2 = \mathcal O\left(C(\A)(m+p)\sqrt {k\log  k}\right). \]
        Finally, the condition number satisfies
        \[\kappa = \mathcal O\left(C(\A)(m+p) \sqrt{k\log  k}\right).\]
    \end{itemize}
\end{proof}
\section{Definition and Lemmas for block-encoding}\label{sec:lem_block_encoding}
\begin{defn}[block-encoding]
    Given an $n$-qubit matrix $\A\in \mathbb C^{N\times N}$ with $N = 2^n$, if we can find $\alpha, \epsilon\in \mathbb R_+$, and an $(m+n)$-qubit unitary matrix $U_{\A}$ such that
    \[ \left\|\A - \alpha\left(\langle 0^m| \otimes I_N\right) U_{\A} \left(|0^m\rangle \otimes I_N\right)\right\|_2 \le \epsilon, \]
    then $U_{\A}$ is called an $(\alpha, m, \epsilon)$-block-encoding of $\A$. In particular, if the block-encoding is exact, i.e., $\epsilon = 0$, then $U_{\A}$ is referred to as an $(\alpha, m)$-block-encoding of $\A$. 
\end{defn}
\begin{lem}\label{lem:2.1}
    Suppose $\A\in \mathbb C^{N_1\times N_1}$ , and $\B\in \mathbb C^{N_2\times N_2}$, where $U_{\A}$ is an $(\alpha_1, m_1)$-block-encoding of $\A$ and $U_{\B}$ is an $(\alpha_2, m_2)$-block-encoding of $\B$. Then, the tensor product $U_{\A}\otimes U_{\B}$ is an $(\alpha_1\alpha_2, m_1 + m_2)$-block-encoding of $\A\otimes \B$. 
\end{lem}
\begin{proof}
    By definition of block-encoding, we have 
    \[ \begin{aligned}
        \A &= \alpha_1 \left(\langle 0^{m_1}| \otimes I_{N_1}\right) U_{\A} \left(|0^{m_1}\rangle \otimes I_{N_1}\right),\\
        \B &= \alpha_2 \left(\langle 0^{m_2}| \otimes I_{N_2}\right) U_{\B} \left(|0^{m_2}\rangle \otimes I_{N_2}\right).
    \end{aligned} \]
    Thus, their tensor product satisfies
    \[ \begin{aligned}
        \A\otimes \B &= \alpha_1 \alpha_2 \left[\left(\langle 0^{m_1}| \otimes I_{N_1}\right) U_{\A} \left(|0^{m_1}\rangle \otimes I_{N_1}\right)\right]\otimes \left[\left(\langle 0^{m_2}| \otimes I_{N_2}\right) U_{\B} \left(|0^{m_2}\rangle \otimes I_{N_2}\right)\right] \\
        &= \alpha_1 \alpha_2 \left[\langle 0^{m_1}| \otimes I_{N_1} \otimes \langle 0^{m_2}| \otimes I_{N_2}\right] \left[U_{\A}\otimes U_{\B}\right] \left[|0^{m_1}\rangle \otimes I_{N_1}\otimes |0^{m_2}\rangle \otimes I_{N_2} \right].
    \end{aligned}\]
    This confirms that $U_{\A}\otimes U_{\B}$ is an $(\alpha_1\alpha_2, m_1 + m_2)$-block-encoding of $\A\otimes \B$. 
\end{proof}
\begin{lem}[Lemma 53 in \cite{10.1145/3313276.3316366}]\label{lem:prod}
    If $U_{\A}$ is an $(\alpha, a, \delta)$-block-encoding of an $s$-qubit operator $\A$, and $U_{\B}$ is a $(\beta, b, \epsilon)$-block-encoding of an $s$-qubit operator $\B$, then $(I_b\otimes U_{\A})(I_a\otimes U_{\B})$ is an $(\alpha\beta, a + b, a\epsilon + \beta\delta)$-block-encoding of $\A\B$.
\end{lem}
\begin{lem}[Linear Combination of Unitaries \cite{7354428}]\label{lem:2.2}
    Let $M = \sum_{i}\alpha_i U_i$ be a linear combination of unitaries $U_i$ with $\alpha_i > 0$ for all $i$. Define $V$ as an operator satisfying $V|0^m\rangle := \frac1{\sqrt \alpha}\sum_{i}\sqrt{\alpha_i}|i\rangle$, where $\alpha:= \sum_i \alpha_i$. Then, the operator $W:= V^\dagger U V$ satisfies 
    \[W|0^m\rangle |\psi\rangle = \frac1\alpha |0^m\rangle M|\psi\rangle + |\perp\rangle\]
    for all states $|\psi\rangle$, where $U:= \sum_i|i\rangle \langle i|\otimes U_i$ and $(|0^m\rangle \langle 0^m|\otimes I)|\perp\rangle = 0$.
\end{lem}
\begin{lem}[Adjust the block-encoding parameters]\label{lem:adjust}
    Suppose $U_{\A}$ is an $(\alpha, m, \epsilon)$-block-encoding of $\A$. Then. there exist a $\theta$ such that $R_Y(\theta)\otimes U_{\A}$ is a $\left(\beta, m+1, \epsilon\right)$-block-encoding of $\A$, where $\beta > \alpha$ and 
    \[R_Y(\theta) = \begin{bmatrix}\cos\left(\frac\theta2\right)&-\sin\left(\frac\theta2\right)\\\sin\left(\frac\theta2\right)&\cos\left(\frac\theta2\right)\end{bmatrix}.\]
\end{lem}
\begin{proof}
    Let $U_{\A}' = R_Y(\theta)\otimes U_{\A}$, where $\theta = 2\arccos \frac{\alpha}{\beta}$. Then, we have 
    \[ \begin{aligned}
        &\left\|\A - \beta\left(\langle 0^{m+1}| \otimes I_N\right) U_{\A}' \left(|0^{m+1}\rangle \otimes I_N\right)\right\|_2\\
        =& \left\|\A - \beta\cos \frac\theta2 \left(\langle 0^{m}| \otimes I_N\right) U_{\A} \left(|0^{m}\rangle \otimes I_N\right)\right\|_2\\
        =& \left\|\A - \alpha \left(\langle 0^{m}| \otimes I_N\right) U_{\A} \left(|0^{m}\rangle \otimes I_N\right)\right\|_2\le \epsilon. 
    \end{aligned} \]
\end{proof}
\section{Analysis of the approximation accuracy and success probability}
\subsection{Proofs of \Cref{lem:sol_err0} and \Cref{lem:sol_err}}\label{pf:lem4.12}
As a preliminary step for proving \Cref{lem:sol_err0} and \Cref{lem:sol_err}, we introduce two lemmas that are slight modifications of the result found in Appendix A of \cite{moler2003nineteen}.
\begin{lem}\label{lem:D_inv_norm}
    If $\|\A\|_2\le 1$ and $k \ge 1$, then $\left\|D_{kk}^{-1}(\A)\right\|_2 \le \frac{2}{3-e}$.
\end{lem}
\begin{proof}
    From the definition of $D_{pq}$ as in \cref{eq:pade_D}, we have $D_{kk}(\A) = \I + \F$, where
    \[\F = \sum_{j=1}^k \frac{(2k-j)! k!}{(2k)!(k-j)!}\frac{(-\A)^j}{j!}.\]
    Using the fact that 
    \[\frac{(2k-j)! k!}{(2k)!(k-j)!} \le \left(\frac12\right)^j,\]
    we find that
    \[\left\|\F\right\|_2 \le \sum_{j=1}^k \left[\frac12 \left\|\A\right\|_2\right]^j\frac1{j!}\le \frac12\|\A\|_2 (e-1) \le \frac{e-1}{2}.\]
    Therefore,
    \[\left\|D_{kk}(\A)^{-1}\right\|_2=\left\|(\I+\F)^{-1}\right\|_2 \le 1 /(1-\|\F\|_2) \le \frac{2}{3-e}.\]
\end{proof}
\begin{lem}\label{lem:A2}
    If $\|\A\|_2\le 1$ and $k\ge 1$. Define
    \[\G := e^{-\A} R_{kk}(\A) - \I.\]
    Then, the following bounds hold
    \[\|\G\|_2 \le 20\cdot  \frac{k! k!}{(2k)!(2k + 1)!}, \quad \text{and} \quad \|\G\A^{-1}\|_2 \le 20\cdot  \frac{k! k!}{(2k)!(2k + 1)!}.\]
\end{lem}
\begin{proof}
    Using the remainder theorem for Pad\'{e} approximation \cite{https://doi.org/10.1002/sapm1961401220,moler2003nineteen}, we have 
    \begin{equation}\label{eq:remainder}
        R_{kk}(\A)=e^{\A}-\frac{(-1)^k}{(2k)!} \A^{2k+1} D_{kk}(\A)^{-1} \int_0^1 e^{(1-u) \A } u^k(1-u)^k \diff u.
    \end{equation}
    Therefore, 
    \[\G := e^{-\A}R_{kk}(\A) - \I = \frac{(-1)^k}{(2k)!} \A^{2k+1} D_{kk}(\A)^{-1} \int_0^1 e^{-u \A } u^k(1-u)^k \diff u.\]
    Using \Cref{lem:D_inv_norm}, we obtain 
    \[\begin{aligned}
        \left\|\G\right\|_2 &\le  \frac{2e}{3-e}\frac{1}{(2k)!} \int_0^1 u^k(1-u)^k \diff u
        \le 20 \frac{k! k!}{(2k)!(2k+1)!}.
    \end{aligned}\]
    Similarly, the inequality for $\|\G \A^{-1}\|_2$ can be derived in the same manner.
\end{proof}

\begin{proof}[Proof of \Cref{lem:sol_err0}]
    It is straightforward to verify that 
    \[\widehat \x(ih) = R^i_{kk}(\A h) \x_0 + \left(R^i_{kk}(\A h) - \I \right)\A^{-1}\bb,\quad\forall i = 1, \dots, m.\]
    Since
    \[\x(ih) = e^{i \A h} \x_0 + \left(e^{i \A h} - \I \right)\A^{-1}\bb,\quad\forall i = 1, \dots, m,\]
    we have 
    \[\begin{aligned}
        \x(ih) - \widehat \x(ih) &= \left(\I - e^{-i \A h}R^i_{kk}(\A h)\right)\x(ih) + \left(\I - e^{-i \A h}R^i_{kk}(\A h)\right)\A^{-1}\bb.
    \end{aligned}\]
    Thus, we obtain the following inequality
    \begin{equation}\label{eq:diff}
        \begin{aligned}
            \left\|\x(ih) - \widehat \x(ih)\right\|_2 &\le \left\|\I - e^{-i \A h} R^i_{kk}(\A h) \right\|_2 \left\|\x(ih)\right\|_2 + \left\|\left(\I - e^{-i \A h} R^i_{kk}(\A h)\right)\A^{-1}\right\|_2 \|\bb\|_2.
        \end{aligned}
    \end{equation}
    Let 
    \[\G = e^{-\A h} R_{kk}(\A h) - \I,\]
    then by \Cref{lem:A2} and condition \cref{eq:cond_k}, we have  $\|\G\|_2 \le \frac\delta 5$. Using the conditions $\delta \in \left(0, \frac1m\right)$ and $\|\A T\|_2 \ge 1$, we deduce that 
    \[m = \lceil \|\A T\|_2\rceil \le 2\|\A T\|_2 \quad \text{and} \quad m\|\G\|_2 \le \frac{m\delta}{5} < 1.\]
    Thus, for all $i = 1, \dots, m$,   
    \begin{equation}\label{ieq:G}
        \begin{aligned}
        \left\|\I - e^{-i \A h}R^i_{kk}(\A h) \right\|_2 &= \left\|\I - (\I + \G)^i\right\|_2\\
        &\le \|\G\|_2\cdot \left(1 + (1 + \|\G\|_2) + \cdots (1 + \|\G\|_2)^{m-1}\right)\\
        &= (1 + \|\G\|_2)^m - 1\\
        &\le e^{m \|\G\|_2} - 1\le (e-1) m\|\G\|_2 \\
        &\le \frac{2(e-1)\delta}{5} \|\A T\|_2 \le \delta \|\A T\|.
    \end{aligned}
    \end{equation}
    For the second term in \cref{eq:diff}, we have 
    \[\begin{aligned}
        \left(\I - e^{-i \A h}R^i_{kk}(\A h)\right)\A^{-1} &= \left(\I - e^{-\A h}R_{kk}(\A h)\right)\A^{-1}\cdot \sum_{l=0}^{i-1} e^{-l \A h} R_{kk}^l(\A h)\\
        &= \G \A^{-1}\cdot \sum_{l=0}^{i-1} (\I + \G)^l.
    \end{aligned}\]
    Thus,
    \[\begin{aligned}
        \left\|\left(\I - e^{-i \A h} R^i_{kk}(\A h)\right)\A^{-1}\right\|_2 &\le h\cdot \left\|\G\cdot (\A h)^{-1} \right\|_2 \cdot \left(1 + (1 + \|\G\|_2) + \cdots (1 + \|\G\|_2)^{m-1}\right)  \\
        &\le h \left((1 + \frac\delta 5)^m - 1\right)\le  \delta h \|\A T\|\le \delta T, \quad \forall i = 1,\dots m,
    \end{aligned}\]
    where we use \Cref{lem:A2} again, yielding $\left\|\G\cdot (\A h)^{-1} \right\|_2\le \frac\delta 5$. Finally, we obtain 
    \[\begin{aligned}
        \left\|\x(ih) - \widehat \x(ih)\right\|_2 &\le \delta T \left(\|\A\|_2 \left\|\x(ih)\right\|_2 + \|\bb\|_2\right), \quad \forall i = 1, \dots, m.
    \end{aligned}\]
\end{proof}
\begin{proof}[proof of \Cref{lem:sol_err}]
    As in the proof of \Cref{lem:sol_err0}, we need to bound the terms $\left\|\I - e^{-i \A h} R^i_{kk}(\A h) \right\|_2$ and $\left\|\left(\I - e^{-i \A h}R^i_{kk}(\A h)\right)\A^{-1}\right\|_2$ for all $i = 1, \dots, m$. Let $\G = e^{-\A h} R_{kk}(\A h) - \I$. Then, we have $\|\G\|\le f_k(\theta)$, where $\theta:= \|\A h\|$. 
    
    For the first term, as shown in \cref{ieq:G}, we have 
    \[\begin{aligned}
        \left\|\I - e^{-i \A h} R^i_{kk}(\A h) \right\|_2 &\le e^{m\|\G\|_2} - 1 \le e^{mf_k(\theta)} - 1 \le (e-1) m f_k(\theta) \le \delta \|\A T\|_2, \quad \forall i = 1, \dots, m
    \end{aligned}\]
    where we use the condition $\delta\in \left(0, \frac{1}{\|\A T\|_2}\right)$ and $f_k(\theta) \le \frac{\theta \delta}{e-1}$.
    
    For the second term, we have 
    \[\left\|\G\cdot (\A h)^{-1} \right\|_2 \le \frac{f_k(\theta)}{\theta} \le \frac{\delta}{e - 1},\]
    and thus 
    \[\begin{aligned}
        \left\|\left(\I - e^{-i \A h} R^i_{kk}(\A h)\right)\A^{-1}\right\|_2 &\le h\cdot \left\|\G\cdot (\A h)^{-1} \right\|_2 \cdot \left(1 + (1 + \|\G\|_2) + \cdots (1 + \|\G\|_2)^{m-1}\right)  \\
        &\le \frac h\theta \left(\left(1 + \frac{\delta \theta}{e - 1}\right)^m - 1\right)\\
        &\le  \frac h\theta (e-1) m \frac{\delta \theta}{e - 1} = \delta T, \quad \forall i = 1,\dots m.
    \end{aligned}\]
    In conclusion, we obtain the same error bound as in \Cref{lem:sol_err0}.
\end{proof}

\subsection{Proof of \Cref{thm:succ_prob}}\label{sec:succ_prob}
\begin{lem}\label{obs:dkk}
    Let $D_{kk}(\cdot)$ be the denominator of the $(k,k)$ Pad\'{e} approximation, as defined in \cref{eq:pade}. For any $\lambda \ge 0$, we have the following inequalities: 
    \[D_{kk}(-\lambda) \ge 1, \quad \text{and} \quad 0\le \frac{D_{kk}(-\lambda)-1}{\lambda D_{kk}(-\lambda)} \le 1.\]
\end{lem}
\begin{proof}
    The first inequality, $D_{kk}(-\lambda) \ge 1$, is straightforward to verify. For the second inequality, we observe that 
    \[\begin{aligned}
        \frac{D_{kk}(-\lambda)-1}{\lambda D_{kk}(-\lambda)} & = \frac1{\lambda + \frac{\lambda}{D_{kk}(-\lambda) - 1}}  \le \min \left\{\frac1\lambda, \frac{D_{kk}(-\lambda) - 1}{\lambda}\right\}.
    \end{aligned}\]
    Now, consider the function
    \[ h(\lambda) := \frac{D_{kk}(-\lambda) - 1}{\lambda} = \sum_{j=1}^k d_j \lambda^{j-1}.\]
    This function is increasing on $[0,+\infty)$, while $1/\lambda$ is decreasing. Furthermore, we have the limits 
    \[\lim_{\lambda \to +\infty }h(\lambda) = +\infty ,\quad \lim_{\lambda \to +\infty}\frac1\lambda = 0.\]
    Additionally, evaluating at $\lambda = 1$,
    \[h(1) = D_{kk}(-1) - 1\le \sqrt e - 1 < 1 = \frac1\lambda \Big|_{\lambda=1},\]
    where the first inequality follows from \Cref{lem:Dkk-1}. Since $h(\lambda)$ is increasing and $1/\lambda$ is decreasing, there exists a unique $\lambda^* \in(1,+\infty)$ such that $h(\lambda^*)=1/\lambda^*$. Thus, for all $\lambda\ge 0$, we conclude that
    \[\frac{D_{kk}(-\lambda)-1}{\lambda D_{kk}(-\lambda)} \le \frac1{\lambda^*} <  1.\]
\end{proof}
\begin{proof}[Proof of \Cref{thm:succ_prob}]
    At first, we consider the norm of $\z^{(1)}$. It is easy to verify that 
    \begin{equation}
        \z_0^{(1)} = D_{kk}^{-1}\left(\A h\right) \x_0 + \left(D_{kk}^{-1}(\A h) - \I\right)\A^{-1} \bb,
    \end{equation}
    since $\z_0^{(1)}$ corresponds to the auxiliary variable $\vv$ in \cref{eq:122}. 
    \begin{itemize}
        \item If $\A$ is Hermitian and negative semi-definite, it admits a spectral decomposition $\A = \U\Lambda \U^\dagger$, where $\Lambda = \text{diag}(\lambda_1,\dots, \lambda_n)$ with $\lambda_1\le\cdots \le \lambda_n\le 0$. Then, we have 
        \[U^\dagger \z_0^{(1)} = D_{kk}^{-1}\left(\Lambda h\right) \left(U^\dagger \x_0\right) + \left(D_{kk}^{-1}(\Lambda h) - \I\right)\Lambda^{-1} \left(U^\dagger \bb\right),\]
        and 
        \[\begin{aligned}
            \left\|\z_0^{(1)}\right\|_2 &\le \left\|D_{kk}^{-1}\left(\Lambda h\right)\right\|_2 \left\|\x_0\right\|_2 + \left\|\left((D_{kk}^{-1}(\Lambda h) - \I\right)(\Lambda h)^{-1}\right\|_2  \left\|h \bb\right\|_2\\
            &\le \left\|\x_0\right\|_2 + h\left\|\bb\right\|_2,
        \end{aligned}\]
        where the second inequality follows from \Cref{obs:dkk}. For $j = 1, \dots, k$, we have 
        \[\begin{aligned}
            \z_j^{(1)} = -d_j D_{kk}^{-1}\left(\A h\right)\cdot \left(-\A h\right)^{j-1}  \left(\A h \x_0 + h\bb\right), \quad j = 1, \dots, k.
        \end{aligned}\]
        This follows from the explicit expression of $\W_k(\A h)^{-1}$. Similarly,
        \[\widehat \z_j^{(1)} = -h d_j D_{kk}^{-1}\left(\Lambda h\right)\cdot \left(-\Lambda h\right)^{j-1}  \left(\Lambda \widehat \x_0 + \widehat \bb\right), \quad j = 1, \dots, k,\]
        where we use the notations
        \[\widehat \z_j^{(1)} = U^\dagger \z_j^{(1)}, \quad \widehat \x_0 = U^\dagger \x_0, \quad \widehat \bb = U^\dagger \bb.\]
        Therefore, we obtain
        \[ \begin{aligned}
            & \quad \sum_{j=1}^k \left\|\z^{(1)}_j\right\|_2^2 = \sum_{j=1}^k \left\|\widehat \z^{(1)}_j\right\|_2^2 = \sum_{j=1}^k \sum_{l=1}^n \left|\widehat \z^{(1)}_{j,l}\right|^2\\
            &= h^2 \sum_{j=1}^k \sum_{l=1}^n d_j^2 \frac{(-\lambda_l h)^{2j-2}}{D_{kk}^2(\lambda_l h)} \left| \lambda_l \widehat \x_{0,l} + \widehat \bb_l\right|^2\\
            &= h^2 \sum_{l=1}^n \left| \lambda_l \widehat \x_{0,l} + \widehat \bb_l\right|^2 \frac{\sum_{j=1}^k d_j^2 (-\lambda_l h)^{2j-2}}{D_{kk}^2(\lambda_l h)}\\
            &= h^2 \sum_{|\lambda_l|\le 1}\left| \lambda_l \widehat \x_{0,l} + \widehat \bb_l\right|^2 \frac{\sum_{j=1}^k d_j^2 (-\lambda_l h)^{2j-2}}{D_{kk}^2(\lambda_l h)} + \sum_{|\lambda_l|> 1}\left|\widehat \x_{0,l} + \lambda_l^{-1}\widehat \bb_l\right|^2 \frac{\sum_{j=1}^k d_j^2 (-\lambda_l h)^{2j}}{D_{kk}^2(\lambda_l h)}\\
            &\le 2h^2 \sum_{|\lambda_l|\le 1}\left(\left| \widehat \x_{0,l}\right|^2 + \left|\widehat \bb_l\right|^2\right) \frac{\left(\sum_{j=1}^k d_j (-\lambda_l h)^{j-1}\right)^2}{D_{kk}^2(\lambda_l h)} + 2\sum_{|\lambda_l|> 1}\left(\left| \widehat \x_{0,l}\right|^2 + \left|\widehat \bb_l\right|^2\right) \frac{\left(\sum_{j=1}^k d_j (-\lambda_l h)^{j}\right)^2}{D_{kk}^2(\lambda_l h)}\\
            & \le 2(h^2 + 1)\left(\|\widehat \x_0\|_2^2 + \|\widehat \bb\|_2^2\right) = 2(h^2 + 1)\left(\|\x_0\|_2^2 + \|\bb\|_2^2\right), 
        \end{aligned} \]
        where the first inequality follows from the facts that $\lambda_l \le 0, \forall l$ and $|a+b|^2 \le 2|a|^2 + 2|b|^2, \forall a, b\in \mathbb C$, while the second equation follows from \Cref{obs:dkk}. Combining these results, we obtain 
        \begin{equation}\label{eq:z_norm}
            \begin{aligned}
                \left\|\z^{(1)}\right\|_2^2 = \sum_{j=0}^k \left\|\z_j^{(1)}\right\|_2^2 &\le \left(\left\|\x_0\right\|_2 + h\left\|\bb\right\|_2\right)^2  + 2(h^2 + 1)\left(\|\x_0\|_2^2 + \|\bb\|_2^2\right)\\
                &\le 2\left((h^2 + 2)\left\|\x_0\right\|_2^2 + (2h^2 + 1) \left\|\bb\right\|_2^2\right).
            \end{aligned}
        \end{equation}
        For $i = 1, \dots, m-1$, we have
        \[\left\|\z^{(i+1)}\right\|_2^2 \le 2\left((h^2 + 2)\left\|\widehat \x(ih)\right\|_2^2 + (2h^2 + 1) \left\|\bb\right\|_2^2\right).\]
        Under the assumptions in \Cref{lem:sol_err}, we obtain inequality \cref{eq:sol_err}. Consequently, we derive 
        \[\begin{aligned}
            \left\|\widehat \x(ih) - \x(ih)\right\|_2 &\le \delta T \cdot \left(\left\|\A\right\|_2\|\x(ih)\|_2 + \|\bb\|_2\right)\\
            &\le \delta T \cdot \left(\left\|\A\right\|_2\max_{0\le t\le T} \|\x(t)\|_2 + \|\bb\|_2\right)\\
            &\le \delta' \cdot \max_{0\le t\le T} \|\x(t)\|_2,
        \end{aligned} \]  
        for all $i = 1, \dots, m$, where 
        \[\delta' =  \delta T \cdot \left(\left\|\A\right\|_2 + \frac{\|\bb\|_2}{\|\x(T)\|_2}\right) < \frac18.\]
        In particular, we obtain $ \left\|\widehat \x(T) - \x(T)\right\|_2 \le \delta' \|\x(T)\|_2$. Thus, we have 
        \[\frac{\left\|\widehat \x(ih)\right\|_2}{\|\widehat \x(T)\|_2} \le \frac{\left\|\x(ih)\right\|_2 + \left\|\x(ih) - \widehat \x(ih)\right\|_2}{\left\|\x(T)\right\|_2 - \left\|\x(T) - \widehat \x(T)\right\|_2} \le \frac{(1 + \delta') \max_{0\le t\le T} \|\x(t)\|_2}{(1 - \delta') \|\x(T)\|_2}.\]
        Finally, we derive
        \[\begin{aligned}
            \mathbb P_{succ} \ge \quad &\frac{p\|\widehat \x(T)\|_2^2}{2 \sum_{i=0}^{m-1}\left((h^2 + 2)\left\|\widehat \x(ih)\right\|_2^2 + (2h^2 + 1) \left\|\bb\right\|_2^2\right) + p\|\widehat \x(T)\|_2^2}\\
            = \quad &  \frac{p}{2 \sum_{i=0}^{m-1} \left((h^2 + 2)\frac{\left\|\widehat \x(ih)\right\|_2^2}{\|\widehat \x(T)\|_2^2} + (2h^2+1) \frac{\left\|\bb\right\|_2^2}{\left\|\widehat \x(T)\right\|_2^2}\right) + p}\\
            \ge \quad & \frac{p}{2 \sum_{i=0}^{m-1} \left((h^2 + 2)\frac{(1 + \delta')^2 \max_{0\le t\le T}\left\|\x(t)\right\|_2^2}{(1 - \delta')^2\|\x(T)\|_2^2} + (2h^2+1) \frac{\left\|\bb\right\|_2^2}{(1-\delta')^2\left\|\x(T)\right\|_2^2}\right) + p}\\
            \ge  \quad & \left(\frac{1-\delta'}{1 + \delta '}\right)^2\cdot \frac{p}{6 m g^2 \left(h^2+1\right) + p}\\
            \ge \quad & \frac{1}{2}\cdot \frac{p}{6 m g^2 \left(h^2+1\right) + p}
        \end{aligned}\]
        where 
        \[g := \frac{\max\left\{\max_{0\le t\le T}\left\|\x(t)\right\|_2, \|\bb\|_2\right\}}{\left\|\x(T)\right\|_2}.\]
        \item If $\A$ is an arbitrary matrix satisfying $\|\A h\|_2\le 1$, then we have 
        \[\begin{aligned}
            \left\|\z_0^{(1)}\right\|_2 &\le  \left\|D_{kk}\left(\A h \right)^{-1}\right\|_2 \left( \left\| \x_0\right\|_2 +\left\|\sum_{j=1}^kd_j (\A h)^{j-1}\right\|_2 h \left\|\bb\right\|_2\right)\\
            &\le \frac{2}{3 - e} \left(\left\| \x_0\right\|_2 +\left(D_{kk}(-1) - 1\right)h \left\|\bb\right\|_2\right)\\
            &\le \frac{2}{3 - e} \left(\left\| \x_0\right\|_2 + h \left\|\bb\right\|_2\right).
        \end{aligned}\]
        The second inequality follows from \Cref{lem:D_inv_norm} and the bound
        \[\left\|\sum_{j=1}^kd_j (\A h)^{j-1}\right\|_2 \le \sum_{j=1}^k d_j \left\|\A h\right\|_2^{j-1} \le  \sum_{j=1}^k d_j = D_{kk}(-1) - 1.\]
        The last inequality is deriving using \Cref{lem:Dkk-1}, which ensures $D_{kk}(-1) - 1\le 1$. Moreover, applying \Cref{lem:D_inv_norm} again, we obtain 
        \[\left\|\z_j^{(1)}\right\|_2 = d_j \left\|D_{kk}^{-1}\left(\A h\right)\right\|_2\cdot \left\|\A h\right\|_2^{j-1}  h \left\|\A \x_0 + \bb\right\|_2 \le \frac{2 h d_j}{3 - e}\left\|\A \x_0 + \bb\right\|_2, \]
        for all $j = 1, \dots, k$. Summing over all $j$, we get 
        \[\begin{aligned}
            \sum_{j=1}^k \left\|\z^{(1)}_j\right\|_2^2 &= \left(\frac{2 h}{3 - e}\left\|\A \x_0 + \bb\right\|_2\right)^2 \sum_{j=1}^k d_j^2 \le \left(\frac{2 h}{3 - e}\right)^2 \left\|\A \x_0 + \bb\right\|_2^2\\
            &\le 2 \left(\frac{2}{3 - e}\right)^2\left(\|\A h\|_2^2 \|\x_0\|_2^2 + h^2 \|\bb\|_2^2\right) \le 2 \left(\frac{2}{3 - e}\right)^2\left(\|\x_0\|_2^2 + h^2\|\bb\|_2^2\right).
        \end{aligned}\] 
        More generally, we obtain
        \[\begin{aligned}
            \sum_{j=0}^k \left\|\z^{(i+1)}_j\right\|_2^2 \le 4 \left(\frac{2}{3 - e}\right)^2\left(\|\widehat \x(ih)\|_2^2 + h^2\|\bb\|_2^2\right),\quad \forall i = 0, \dots, m-1,
        \end{aligned}\] 
        where the constant factor satisfies $\left(\frac{2}{3-e}\right)^2\le 51$. Following a similar argument as in the previous case, we establish 
        \[\begin{aligned}
            \mathbb P_{succ} \ge \quad &\frac{p\|\widehat \x(T)\|_2^2}{204 \sum_{i=0}^{m-1}\left(\left\|\widehat \x(ih)\right\|_2^2 + h^2 \left\|\bb\right\|_2^2\right) + p\|\widehat \x(T)\|_2^2}\\
            \ge \quad &\frac12\cdot \frac{p}{204mg^2(h^2 + 1) + p},
        \end{aligned}\]
        with the aid of \Cref{lem:sol_err0}, where 
        \[ g := \frac{\max\left\{\max_{0\le t\le T}\left\|\x(t)\right\|_2, \|\bb\|_2\right\}}{\left\|\x(T)\right\|_2}. \]
    \end{itemize}
\end{proof}
\section{Lemmas for quantum states}
The following three lemmas from \cite{berry2017quantum} are used to bound the distance between quantum states.
\begin{lem}[Lemma 13 in \cite{berry2017quantum}]\label{lem:berry13}
    Let $\boldsymbol \psi$ and $\boldsymbol \varphi$ be two vectors such that $\|\boldsymbol \psi\|_2\ge \alpha > 0$ and $\|\boldsymbol \psi - \boldsymbol \varphi\|_2\le \beta$. Then
    \[\left\|\frac{\boldsymbol \psi}{\left\|\boldsymbol \psi\right\|_2} - \frac{\boldsymbol \varphi}{\left\|\boldsymbol \varphi\right\|_2}\right\|_2 \le \frac{2\beta}{\alpha}.\]
\end{lem}
\begin{lem}[Lemma 14 in \cite{berry2017quantum}]\label{lem:berry14}
    Let $|\boldsymbol \psi\rangle = \alpha|0\rangle |\boldsymbol \psi_0\rangle + \sqrt{1-\alpha^2}|1\rangle |\boldsymbol \psi_1\rangle$ and $|\boldsymbol \varphi\rangle = \beta|0\rangle |\boldsymbol \varphi_0\rangle + \sqrt{1-\beta^2}|1\rangle |\boldsymbol \varphi_1\rangle$, where $|\boldsymbol \psi_0\rangle, |\boldsymbol \psi_1\rangle, |\boldsymbol \varphi_0\rangle, |\boldsymbol \varphi_1\rangle$ are unit vectors, and $\alpha, \beta\in [0,1]$. Suppose $\left\||\boldsymbol \psi\rangle - |\boldsymbol \varphi\rangle \right\|_2 \le \delta < \alpha$. Then $\left\||\boldsymbol \psi_0\rangle - |\boldsymbol \varphi_0\rangle \right\|_2 \le \frac{2\delta}{\alpha-\delta}$.
\end{lem}
\begin{lem}[Lemma 15 in \cite{berry2017quantum}]\label{lem:berry15}
    Let $|\boldsymbol \psi\rangle = \alpha|0\rangle |\boldsymbol \psi_0\rangle + \sqrt{1-\alpha^2}|1\rangle |\boldsymbol \psi_1\rangle$ and $|\boldsymbol \varphi\rangle = \beta|0\rangle |\boldsymbol \varphi_0\rangle + \sqrt{1-\beta^2}|1\rangle |\boldsymbol \varphi_1\rangle$, where $|\boldsymbol \psi_0\rangle, |\boldsymbol \psi_1\rangle, |\boldsymbol \varphi_0\rangle, |\boldsymbol \varphi_1\rangle$ are unit vectors, and $\alpha, \beta\in [0,1]$. Suppose $\left\||\boldsymbol \psi\rangle - |\boldsymbol \varphi\rangle \right\|_2 \le \delta < \alpha$. Then $\beta \ge \alpha - \delta$.
\end{lem}

\bibliography{reference}

\begin{thebibliography}{39}
\providecommand{\natexlab}[1]{#1}
\providecommand{\url}[1]{\texttt{#1}}
\expandafter\ifx\csname urlstyle\endcsname\relax
  \providecommand{\doi}[1]{doi: #1}\else
  \providecommand{\doi}{doi: \begingroup \urlstyle{rm}\Url}\fi

\bibitem[Al-Mohy and Higham(2010)]{doi:10.1137/09074721X}
Awad~H. Al-Mohy and Nicholas~J. Higham.
\newblock A new scaling and squaring algorithm for the matrix exponential.
\newblock \emph{SIAM Journal on Matrix Analysis and Applications}, 31\penalty0
  (3):\penalty0 970--989, 2010.
\newblock \doi{10.1137/09074721X}.
\newblock URL \url{https://doi.org/10.1137/09074721X}.

\bibitem[An et~al.(2023{\natexlab{a}})An, Childs, and Lin]{An2023QuantumAF}
Dong An, Andrew~M. Childs, and Lin Lin.
\newblock Quantum algorithm for linear non-unitary dynamics with near-optimal
  dependence on all parameters.
\newblock \emph{ArXiv}, abs/2312.03916, 2023{\natexlab{a}}.
\newblock URL \url{https://doi.org/10.48550/arXiv.2312.03916}.

\bibitem[An et~al.(2023{\natexlab{b}})An, Liu, and Lin]{PhysRevLett.131.150603}
Dong An, Jin-Peng Liu, and Lin Lin.
\newblock Linear combination of hamiltonian simulation for nonunitary dynamics
  with optimal state preparation cost.
\newblock \emph{Phys. Rev. Lett.}, 131:\penalty0 150603, Oct
  2023{\natexlab{b}}.
\newblock \doi{10.1103/PhysRevLett.131.150603}.
\newblock URL \url{https://doi.org/10.1103/PhysRevLett.131.150603}.

\bibitem[Berry(2014)]{berry2014high}
Dominic~W Berry.
\newblock High-order quantum algorithm for solving linear differential
  equations.
\newblock \emph{Journal of Physics A: Mathematical and Theoretical},
  47\penalty0 (10):\penalty0 105301, feb 2014.
\newblock \doi{10.1088/1751-8113/47/10/105301}.
\newblock URL \url{https://dx.doi.org/10.1088/1751-8113/47/10/105301}.

\bibitem[Berry and C.~S.~Costa(2024)]{berry2024quantum}
Dominic~W. Berry and Pedro C.~S.~Costa.
\newblock Quantum algorithm for time-dependent differential equations using
  {D}yson series.
\newblock \emph{{Quantum}}, 8:\penalty0 1369, June 2024.
\newblock ISSN 2521-327X.
\newblock \doi{10.22331/q-2024-06-13-1369}.
\newblock URL \url{https://doi.org/10.22331/q-2024-06-13-1369}.

\bibitem[Berry et~al.(2015)Berry, Childs, and Kothari]{7354428}
Dominic~W. Berry, Andrew~M. Childs, and Robin Kothari.
\newblock Hamiltonian simulation with nearly optimal dependence on all
  parameters.
\newblock In \emph{2015 IEEE 56th Annual Symposium on Foundations of Computer
  Science}, pages 792--809, 2015.
\newblock \doi{10.1109/FOCS.2015.54}.
\newblock URL \url{https://doi.org/10.1109/FOCS.2015.54}.

\bibitem[Berry et~al.(2017)Berry, Childs, Ostrander, and
  Wang]{berry2017quantum}
Dominic~W Berry, Andrew~M Childs, Aaron Ostrander, and Guoming Wang.
\newblock Quantum algorithm for linear differential equations with
  exponentially improved dependence on precision.
\newblock \emph{Communications in Mathematical Physics}, 356:\penalty0
  1057--1081, 2017.
\newblock \doi{10.1007/s00220-017-3002-y}.
\newblock URL \url{https://doi.org/10.1007/s00220-017-3002-y}.

\bibitem[Brassard et~al.(2002)Brassard, Høyer, Mosca, and Tapp]{Brassard_2002}
Gilles Brassard, Peter Høyer, Michele Mosca, and Alain Tapp.
\newblock Quantum amplitude amplification and estimation, 2002.
\newblock ISSN 0271-4132.
\newblock URL \url{https://doi.org/10.48550/arXiv.quant-ph/0005055}.

\bibitem[Camps and Van~Beeumen(2022)]{9951292}
Daan Camps and Roel Van~Beeumen.
\newblock { FABLE: Fast Approximate Quantum Circuits for Block-Encodings }.
\newblock In \emph{2022 IEEE International Conference on Quantum Computing and
  Engineering (QCE)}, pages 104--113, Los Alamitos, CA, USA, September 2022.
  IEEE Computer Society.
\newblock \doi{10.1109/QCE53715.2022.00029}.
\newblock URL \url{https://doi.org/10.1109/QCE53715.2022.00029}.

\bibitem[Cao et~al.(2023)Cao, Jin, and
  Liu]{cao2023quantumsimulationtimedependenthamiltonians}
Yu~Cao, Shi Jin, and Nana Liu.
\newblock Quantum simulation for time-dependent hamiltonians -- with
  applications to non-autonomous ordinary and partial differential equations,
  2023.
\newblock URL \url{https://doi.org/10.48550/arXiv.2312.02817}.

\bibitem[Childs and Liu(2019)]{Childs2019QuantumSM}
Andrew~M. Childs and Jin-Peng Liu.
\newblock Quantum spectral methods for differential equations.
\newblock \emph{Communications in Mathematical Physics}, 375:\penalty0 1427 --
  1457, 2019.
\newblock URL \url{https://doi.org/10.1007/s00220-020-03699-z}.

\bibitem[Costa et~al.(2022)Costa, An, Sanders, Su, Babbush, and
  Berry]{costa2022optimal}
Pedro~C.S. Costa, Dong An, Yuval~R. Sanders, Yuan Su, Ryan Babbush, and
  Dominic~W. Berry.
\newblock Optimal scaling quantum linear-systems solver via discrete adiabatic
  theorem.
\newblock \emph{PRX Quantum}, 3:\penalty0 040303, Oct 2022.
\newblock \doi{10.1103/PRXQuantum.3.040303}.
\newblock URL \url{https://doi.org/10.1103/PRXQuantum.3.040303}.

\bibitem[Dong et~al.(2024)Dong, Peng, Tang, Yang, and
  Yu]{dong2024investigationquantumalgorithmlinear}
Xiaojing Dong, Yizhe Peng, Qili Tang, Yin Yang, and Yue Yu.
\newblock Investigation on a quantum algorithm for linear differential
  equations, 2024.
\newblock URL \url{https://doi.org/10.48550/arXiv.2408.01762}.

\bibitem[Gily\'{e}n et~al.(2019)Gily\'{e}n, Su, Low, and
  Wiebe]{10.1145/3313276.3316366}
Andr\'{a}s Gily\'{e}n, Yuan Su, Guang~Hao Low, and Nathan Wiebe.
\newblock Quantum singular value transformation and beyond: exponential
  improvements for quantum matrix arithmetics.
\newblock In \emph{Proceedings of the 51st Annual ACM SIGACT Symposium on
  Theory of Computing}, STOC 2019, page 193–204, New York, NY, USA, 2019.
  Association for Computing Machinery.
\newblock ISBN 9781450367059.
\newblock \doi{10.1145/3313276.3316366}.
\newblock URL \url{https://doi.org/10.1145/3313276.3316366}.

\bibitem[Harrow et~al.(2009)Harrow, Hassidim, and
  Lloyd]{PhysRevLett.103.150502}
Aram~W. Harrow, Avinatan Hassidim, and Seth Lloyd.
\newblock Quantum algorithm for linear systems of equations.
\newblock \emph{Phys. Rev. Lett.}, 103:\penalty0 150502, Oct 2009.
\newblock \doi{10.1103/PhysRevLett.103.150502}.
\newblock URL \url{https://doi.org/10.1103/PhysRevLett.103.150502}.

\bibitem[Higham(2005)]{higham2005scaling}
Nicholas~J. Higham.
\newblock The scaling and squaring method for the matrix exponential revisited.
\newblock \emph{SIAM Journal on Matrix Analysis and Applications}, 26\penalty0
  (4):\penalty0 1179--1193, 2005.
\newblock \doi{10.1137/04061101X}.
\newblock URL \url{https://doi.org/10.1137/04061101X}.

\bibitem[Hu et~al.(2023)Hu, Jin, Liu, and
  Zhang]{hu2023dilationtheoremschrodingerisationapplications}
Junpeng Hu, Shi Jin, Nana Liu, and Lei Zhang.
\newblock Dilation theorem via schr\"odingerisation, with applications to the
  quantum simulation of differential equations, 2023.
\newblock URL \url{https://doi.org/10.48550/arXiv.2309.16262}.

\bibitem[Hu et~al.(2024{\natexlab{a}})Hu, Jin, Liu, and
  Zhang]{Hu2024quantumcircuits}
Junpeng Hu, Shi Jin, Nana Liu, and Lei Zhang.
\newblock Quantum {C}ircuits for partial differential equations via
  {S}chr{\"{o}}dingerisation.
\newblock \emph{{Quantum}}, 8:\penalty0 1563, December 2024{\natexlab{a}}.
\newblock ISSN 2521-327X.
\newblock \doi{10.22331/q-2024-12-12-1563}.
\newblock URL \url{https://doi.org/10.22331/q-2024-12-12-1563}.

\bibitem[Hu et~al.(2024{\natexlab{b}})Hu, Jin, and
  Zhang]{doi:10.1137/23M1566340}
Junpeng Hu, Shi Jin, and Lei Zhang.
\newblock Quantum algorithms for multiscale partial differential equations.
\newblock \emph{Multiscale Modeling \& Simulation}, 22\penalty0 (3):\penalty0
  1030--1067, 2024{\natexlab{b}}.
\newblock \doi{10.1137/23M1566340}.
\newblock URL \url{https://doi.org/10.1137/23M1566340}.

\bibitem[Jin and Liu(2023)]{jin2023quantumsimulationdiscretelinear}
Shi Jin and Nana Liu.
\newblock Quantum simulation of discrete linear dynamical systems and simple
  iterative methods in linear algebra via schrodingerisation, 2023.
\newblock URL \url{https://doi.org/10.48550/arXiv.2304.02865}.

\bibitem[Jin and Liu(2024{\natexlab{a}})]{Jin_2024}
Shi Jin and Nana Liu.
\newblock Analog quantum simulation of partial differential equations.
\newblock \emph{Quantum Science and Technology}, 9\penalty0 (3):\penalty0
  035047, jun 2024{\natexlab{a}}.
\newblock \doi{10.1088/2058-9565/ad49cf}.
\newblock URL \url{https://dx.doi.org/10.1088/2058-9565/ad49cf}.

\bibitem[Jin and
  Liu(2024{\natexlab{b}})]{jin2024analogquantumsimulationparabolic}
Shi Jin and Nana Liu.
\newblock Analog quantum simulation of parabolic partial differential equations
  using jaynes-cummings-like models, 2024{\natexlab{b}}.
\newblock URL \url{https://doi.org/10.48550/arXiv.2407.01913}.

\bibitem[Jin et~al.(2022)Jin, Liu, and
  Yu]{jin2022quantumsimulationpartialdifferential}
Shi Jin, Nana Liu, and Yue Yu.
\newblock Quantum simulation of partial differential equations via
  schrodingerisation, 2022.
\newblock URL \url{https://doi.org/10.48550/arXiv.2212.13969}.

\bibitem[Jin et~al.(2023{\natexlab{a}})Jin, Liu, and
  Ma]{jin2023quantumsimulationmaxwellsequations}
Shi Jin, Nana Liu, and Chuwen Ma.
\newblock Quantum simulation of maxwell's equations via schr\"odingersation,
  2023{\natexlab{a}}.
\newblock URL \url{https://doi.org/10.48550/arXiv.2308.08408}.

\bibitem[Jin et~al.(2023{\natexlab{b}})Jin, Liu, and Yu]{PhysRevA.108.032603}
Shi Jin, Nana Liu, and Yue Yu.
\newblock Quantum simulation of partial differential equations: Applications
  and detailed analysis.
\newblock \emph{Phys. Rev. A}, 108:\penalty0 032603, Sep 2023{\natexlab{b}}.
\newblock \doi{10.1103/PhysRevA.108.032603}.
\newblock URL \url{https://doi.org/10.1103/PhysRevA.108.032603}.

\bibitem[Jin et~al.(2024{\natexlab{a}})Jin, Li, Liu, and Yu]{JIN2024112707}
Shi Jin, Xiantao Li, Nana Liu, and Yue Yu.
\newblock Quantum simulation for partial differential equations with physical
  boundary or interface conditions.
\newblock \emph{Journal of Computational Physics}, 498:\penalty0 112707,
  2024{\natexlab{a}}.
\newblock ISSN 0021-9991.
\newblock \doi{10.1016/j.jcp.2023.112707}.
\newblock URL \url{https://doi.org/10.1016/j.jcp.2023.112707}.

\bibitem[Jin et~al.(2024{\natexlab{b}})Jin, Li, Liu, and
  Yu]{doi:10.1137/23M1563451}
Shi Jin, Xiantao Li, Nana Liu, and Yue Yu.
\newblock Quantum simulation for quantum dynamics with artificial boundary
  conditions.
\newblock \emph{SIAM Journal on Scientific Computing}, 46\penalty0
  (4):\penalty0 B403--B421, 2024{\natexlab{b}}.
\newblock \doi{10.1137/23M1563451}.
\newblock URL \url{https://doi.org/10.1137/23M1563451}.

\bibitem[Jin et~al.(2024{\natexlab{c}})Jin, Liu, and
  Ma]{jin2024schrodingerisationbasedcomputationallystable}
Shi Jin, Nana Liu, and Chuwen Ma.
\newblock Schr\"odingerisation based computationally stable algorithms for
  ill-posed problems in partial differential equations, 2024{\natexlab{c}}.
\newblock URL \url{https://doi.org/10.48550/arXiv.2403.19123}.

\bibitem[Jin et~al.(2024{\natexlab{d}})Jin, Liu, and
  Ma]{jin2024schrodingerizationbasedquantumalgorithms}
Shi Jin, Nana Liu, and Chuwen Ma.
\newblock On schr\"odingerization based quantum algorithms for linear dynamical
  systems with inhomogeneous terms, 2024{\natexlab{d}}.
\newblock URL \url{https://doi.org/10.48550/arXiv.2402.14696}.

\bibitem[Jin et~al.(2024{\natexlab{e}})Jin, Liu, and
  Yu]{jin2024quantumsimulationfokkerplanckequation}
Shi Jin, Nana Liu, and Yue Yu.
\newblock Quantum simulation of the fokker-planck equation via
  schrodingerization, 2024{\natexlab{e}}.
\newblock URL \url{https://doi.org/10.48550/arXiv.2404.13585}.

\bibitem[Jin et~al.(2025)Jin, Liu, and Yu]{jin2025quantumcircuitsheatequation}
Shi Jin, Nana Liu, and Yue Yu.
\newblock Quantum circuits for the heat equation with physical boundary
  conditions via schrodingerisation, 2025.
\newblock URL \url{https://doi.org/10.48550/arXiv.2407.15895}.

\bibitem[Krovi(2023)]{krovi2023improved}
Hari Krovi.
\newblock Improved quantum algorithms for linear and nonlinear differential
  equations.
\newblock \emph{{Quantum}}, 7:\penalty0 913, February 2023.
\newblock ISSN 2521-327X.
\newblock \doi{10.22331/q-2023-02-02-913}.
\newblock URL \url{https://doi.org/10.22331/q-2023-02-02-913}.

\bibitem[Moler and Van~Loan(2003)]{moler2003nineteen}
Cleve Moler and Charles Van~Loan.
\newblock Nineteen dubious ways to compute the exponential of a matrix,
  twenty-five years later.
\newblock \emph{SIAM Review}, 45\penalty0 (1):\penalty0 3--49, 2003.
\newblock \doi{10.1137/S00361445024180}.
\newblock URL \url{https://doi.org/10.1137/S00361445024180}.

\bibitem[Montanaro(2016)]{Montanaro2016}
Ashley Montanaro.
\newblock Quantum algorithms: an overview.
\newblock \emph{npj Quantum Information}, 2\penalty0 (1):\penalty0 15023, Jan
  2016.
\newblock ISSN 2056-6387.
\newblock \doi{10.1038/npjqi.2015.23}.
\newblock URL \url{https://doi.org/10.1038/npjqi.2015.23}.

\bibitem[M\"ott\"onen et~al.(2004)M\"ott\"onen, Vartiainen, Bergholm, and
  Salomaa]{PhysRevLett.93.130502}
Mikko M\"ott\"onen, Juha~J. Vartiainen, Ville Bergholm, and Martti~M. Salomaa.
\newblock Quantum circuits for general multiqubit gates.
\newblock \emph{Phys. Rev. Lett.}, 93:\penalty0 130502, Sep 2004.
\newblock \doi{10.1103/PhysRevLett.93.130502}.
\newblock URL \url{https://doi.org/10.1103/PhysRevLett.93.130502}.

\bibitem[Saff and Varga(1975)]{saff1975zeros}
Edwarwd~B Saff and Richard~S Varga.
\newblock On the zeros and poles of pad{\'e} approximants to ez.
\newblock \emph{Numerische Mathematik}, 25\penalty0 (1):\penalty0 1--14, 1975.
\newblock \doi{10.1007/BF01411842}.
\newblock URL \url{https://doi.org/10.1007/BF01411842}.

\bibitem[Shor(1997)]{doi:10.1137/S0097539795293172}
Peter~W. Shor.
\newblock Polynomial-time algorithms for prime factorization and discrete
  logarithms on a quantum computer.
\newblock \emph{SIAM Journal on Computing}, 26\penalty0 (5):\penalty0
  1484--1509, 1997.
\newblock \doi{10.1137/S0097539795293172}.
\newblock URL \url{https://doi.org/10.1137/S0097539795293172}.

\bibitem[Varga(1961)]{https://doi.org/10.1002/sapm1961401220}
Richard~S. Varga.
\newblock On higher order stable implicit methods for solving parabolic partial
  differential equations.
\newblock \emph{Journal of Mathematics and Physics}, 40\penalty0
  (1-4):\penalty0 220--231, 1961.
\newblock \doi{https://doi.org/10.1002/sapm1961401220}.

\bibitem[Wu and Li(2025)]{wu2025structurepreservingquantumalgorithmslinear}
Hsuan-Cheng Wu and Xiantao Li.
\newblock Structure-preserving quantum algorithms for linear and nonlinear
  hamiltonian systems, 2025.
\newblock URL \url{https://doi.org/10.48550/arXiv.2411.03599}.

\end{thebibliography}
\bibliographystyle{plainnat}

\end{document}